\documentclass[sigconf]{acmart}

\usepackage{colortbl}
\usepackage{mathtools}
\usepackage{amsmath}
\usepackage{thmtools, thm-restate}
\usepackage[ruled, vlined, linesnumbered]{algorithm2e}
\hypersetup{
	pdftitle={Annotated Dependency Pairs for Expected Innermost Runtime Complexity and Strong Almost-Sure Termination in Term Rewriting}, colorlinks=true, linkcolor=blue, citecolor=olive, filecolor=magenta, urlcolor=cyan
}
\usepackage{subcaption}
\usepackage{nicefrac,xfrac}
\usepackage{cleveref}
\usepackage{makecell}
\makeatletter \newcommand*\bigcdot{\mathpalette\bigcdot@{.4}}
\newcommand*\bigcdot@[2]{\mathbin{\vcenter{\hbox{\scalebox{#2}{$\m@th#1\bullet$}}}}}
\makeatother

\renewcommand{\emptyset}{\varnothing}

\newcommand{\oldcomment}[1]{}

\newcommand{\myfootnote}[1]{\footnote{#1}}
\renewcommand{\epsilon}{\varepsilon}

\RequirePackage{tikz}             % Drawing tool
\usetikzlibrary{shapes,calc,arrows,automata,decorations.pathmorphing,backgrounds,arrows.meta,shapes.geometric,shapes.multipart,shapes.misc}
\usetikzlibrary{arrows}
\usetikzlibrary{shapes.geometric}
\usetikzlibrary{arrows.meta}
\usetikzlibrary{calc}
\usetikzlibrary{intersections}
\usetikzlibrary{decorations.markings}
\usetikzlibrary{backgrounds}
\usetikzlibrary{matrix,arrows,decorations.pathmorphing}
\usetikzlibrary{positioning, 
                quotes}
\RequirePackage{xparse}
\RequirePackage{xspace}
\RequirePackage{tikz-cd}
\tikzset{commutative diagrams/.cd}
\RequirePackage{wrapfig}

\renewcommand{\emph}[1]{\index{#1}\textit{#1}}

\renewcommand{\emptyset}{\varnothing}

\newcommand{\IN}{\mathbb{N}}

\newcommand{\IE}{\ensuremath{\mathbb{E}}}

\newcommand{\F}[1]{\mathfrak{#1}}
\newcommand{\C}[1]{\mathcal{#1}}

\makeatletter
\def\moverlay{\mathpalette\mov@rlay}
\def\mov@rlay#1#2{\leavevmode\vtop{%
   \baselineskip\z@skip \lineskiplimit-\maxdimen
   \ialign{\hfil$\m@th#1##$\hfil\cr#2\crcr}}}
\newcommand{\charfusion}[3][\mathord]{
    #1{\ifx#1\mathop\vphantom{#2}\fi
        \mathpalette\mov@rlay{#2\cr#3}
      }
    \ifx#1\mathop\expandafter\displaylimits\fi}
\makeatother

\newcommand{\aprove}{\textsf{AProVE}}
\newcommand{\muterm}{\textsf{MU-TERM}}
\newcommand{\natt}{\textsf{NaTT}}
\newcommand{\ttttwo}{\textsf{T\kern-0.15em\raisebox{-0.55ex}T\kern-0.15emT\kern-0.15em\raisebox{-0.55ex}2}}
\newcommand{\tct}{\textsf{TcT}}

\newcommand{\tool}[1]{\textsf{#1}}

\newcommand{\Var}{\mathcal{V}}

\newcommand{\Trans}{\mathcal{T}}

\newcommand{\I}{\mathcal{I}}
\newcommand{\ISum}{\mathcal{I}_{\mbox{\tiny $\sum$}}}

\renewcommand{\P}{\mathcal{P}}

\renewcommand{\S}{\mathcal{S}}
\newcommand{\K}{\mathcal{K}}

\newcommand{\bat}{\text{(\textbf{at})}}
\newcommand{\baf}{\text{(\textbf{af})}}
\newcommand{\bnt}{\text{(\textbf{nt})}}
\newcommand{\bnf}{\text{(\textbf{nf})}}

\newcommand{\D}{\C{D}}

\newcommand{\Proc}{\operatorname{Proc}}

\newcommand{\TSet}[2]{\mathcal{T}\left(#1,#2\right)}

\newcommand{\VSet}{\mathcal{V}}

\renewcommand{\O}{\mathcal{O}}
\newcommand{\R}{\mathcal{R}}

\newcommand{\ADPair}[1]{\mathcal{A}(#1)}

\newcommand{\FDist}{\operatorname{FDist}}
\newcommand{\Supp}{\operatorname{Supp}}

\newcommand{\rootsym}{\operatorname{root}}
\newcommand{\rootterm}{\operatorname{root}}

\newcommand{\rules}{\operatorname{Rules}}
\newcommand{\urules}{\mathcal{U}}

\newcommand{\geo}{\operatorname{geo}}
\newcommand{\edh}{\operatorname{edh}}
\newcommand{\idh}{\operatorname{dh}}
\newcommand{\idl}{\operatorname{dl}}
\newcommand{\edl}{\operatorname{edl}}
\newcommand{\CPI}{\operatorname{CPI}}
\newcommand{\Pol}{\operatorname{Pol}}
\newcommand{\pl}{\operatorname{pol}}

\newcommand{\capt}{\operatorname{capt}}
\newcommand{\Exp}{\operatorname{Exp}}
\newcommand{\DExp}{\operatorname{2-Exp}}
\newcommand{\Fin}{\operatorname{Fin}}

\newcommand{\nonprob}{\normalfont{\text{np}}}
\newcommand{\nonprobDP}{\normalfont{\text{dp}}}
\newcommand{\nonprobDT}{\normalfont{\text{dt}}}
\newcommand{\PP}{\mathcal{P}}
\newcommand{\JJ}{\mathcal{J}}

\newcommand{\SSS}{\mathcal{S}}
\newcommand{\TT}{\mathcal{T}}

\newcommand{\Pre}{\operatorname{Pre}}

\newcommand{\tQ}{\mathsf{Q}}

\newcommand{\tstart}{\mathsf{start}}
\newcommand{\tz}{\mathsf{0}}

\newcommand{\tplus}{\mathsf{plus}}

\newcommand{\ts}{\mathsf{s}}
\renewcommand{\O}{\mathcal{O}}
\newcommand{\tf}{\mathsf{f}}

\newcommand{\tq}{\mathsf{q}}
\newcommand{\tg}{\mathsf{g}}
\renewcommand{\th}{\mathsf{h}}
\newcommand{\ta}{\mathsf{a}}
\newcommand{\tb}{\mathsf{b}}
\newcommand{\tc}{\mathsf{c}}
\newcommand{\td}{\mathsf{d}}
\newcommand{\te}{\mathsf{e}}

\newcommand{\ttrue}{\mathsf{true}}
\newcommand{\tfalse}{\mathsf{false}}

\newcommand{\tF}{\mathsf{F}}

\newcommand{\tgeo}{\mathsf{geo}}
\newcommand{\tStart}{\mathsf{Start}}
\newcommand{\tGeo}{\mathsf{Geo}}
\newcommand{\ruleArr}[3]{
  \mathrel{
    \xrightarrow{{}_{\scriptstyle #1}}
    \!\!{}^{#2}_{#3}
  }
}
\newcommand{\tored}[3]{
  \mathrel{
    \xhookrightarrow{{}_{\mathsf{i}}}
    \!\!{}^{#2}_{#3}
  }
}

\newcommand{\PosDPoss}{\pos_{\mathcal{D} \land \lnot \mathtt{NF}_{\R}}}

\newcommand{\PathComp}{c}

\newcommand{\fun}[1]{\mathrm{#1}}
\newcommand{\pos}{\fun{Pos}}

\newcommand{\posT}{\fun{Pos}_{\SignatureA}}
\newcommand{\posD}{\fun{Pos}_{\SignatureD}}
\newcommand{\NF}{\mathtt{NF}} 
\newcommand{\ANF}{\mathtt{ANF}} 

\newcommand{\anno}{\sharp} 
\newcommand{\annoD}{\anno_{\SignatureD}}
\newcommand{\disannoPos}[1]{\flat^{\uparrow}_{#1}}
\newcommand{\SignatureADC}{\Sigma^\sharp}
\newcommand{\SignatureC}{\mathcal{C}}
\newcommand{\SignatureD}{\mathcal{D}}

\newcommand{\SignatureA}{\mathcal{D}^\sharp}
\newcommand{\ctleaf}{\operatorname{Leaf}}

\newcommand{\AST}{\operatorname{\texttt{AST}}}

\newcommand{\maybePAST}{\operatorname{\texttt{(P)AST}}}
\newcommand{\PAST}{\operatorname{\texttt{PAST}}}

\newcommand{\SAST}{\operatorname{\texttt{SAST}}}

\makeatletter
\NewDocumentCommand{\dparrow}{+O{} +O{0.5cm}}{%
\begin{tikzpicture}[baseline=-0.5ex] {
\node[inner sep=0](@1) at (0,0) {};
\node[inner sep=0](@2) at (#2,0) {};
\draw [arrows={-Triangle[open]},shorten >= 1pt,shorten <= 1pt](@1)--(@2) node[pos=.5,above,inner sep=1pt] {\ensuremath{#1}};}
\end{tikzpicture}\xspace
}

\NewDocumentCommand{\myto}{+O{} +O{0.5cm}}{%
\begin{tikzpicture}[baseline=-0.5ex] {
\node[inner sep=0](@1) at (0,0) {};
\node[inner sep=0](@2) at (#2,0) {};
\draw [arrows={-to}](@1)--(@2) node[pos=.5,above,inner sep=1pt] {\ensuremath{#1}};}
\end{tikzpicture}\xspace
}

\NewDocumentCommand{\paraarrow}{+O{} +O{0.4cm}}{%
\begin{tikzpicture}[baseline=-0.5ex] {
\node[inner sep=0](@1) at (0,0) {};
\node[inner sep=0](@2) at (#2,0) {};
\node[inner sep=0](@3) at (0.07,0) {};
\draw [arrows={-to}](@1)--(@2) node[pos=.5,above,inner sep=1pt] {\ensuremath{#1}};
\draw [arrows={-to}](@1)--(@3);}
\end{tikzpicture}\xspace
}

\NewDocumentCommand{\paradparrow}{+O{} +O{0.4cm}}{%
\begin{tikzpicture}[baseline=-0.5ex] {
\node[inner sep=0](@1) at (0,0) {};
\node[inner sep=0](@2) at (#2,0) {};
\node[inner sep=0](@3) at (0.07,0) {};
\draw [arrows={-Triangle[open]}](@1)--(@2) node[pos=.5,above,inner sep=1pt] {\ensuremath{#1}};
\draw [arrows={-to}](@1)--(@3);}
\end{tikzpicture}\xspace
}

\newcommand{\oset}[2]{%
  {\mathop{#2}\limits^{\vbox to 1\ex@{\kern-\tw@\ex@
   \hbox{\scriptsize #1}\vss}}}}

\newcommand{\osetthree}[2]{%
  {\mathop{#2}\limits^{\vbox to 3\ex@{\kern-\tw@\ex@
   \hbox{\scriptsize #1}\vss}}}}

\newcommand{\osetfive}[2]{%
  {\mathop{#2}\limits^{\vbox to 5\ex@{\kern-\tw@\ex@
   \hbox{\scriptsize #1}\vss}}}}

\newcommand{\osetminus}[2]{%
  {\mathop{#2}\limits^{\vbox to -2\ex@{\kern-\tw@\ex@
   \hbox{\scriptsize #1}\vss}}}}
\makeatother

\newcommand{\ito}{\mathrel{\smash{\stackrel{\raisebox{2pt}{\tiny $\mathsf{i}\:$}}{\smash{\rightarrow}}}}}

\newcommand{\itor}{\mathrel{\ito_{\R}}}
\newcommand{\itorn}{\mathrel{\ito^{n}_{\mathop{\R}}}}

\newcommand{\TB}{\mathcal{B}\!\mathcal{T}}   % Basic terms
\newcommand{\irc}{\mathrm{rc}}    % runtime complexity function
\newcommand{\eirc}{\mathrm{erc}}    % runtime complexity function
\crefname{definition}{Def.}{Def.}
\crefname{example}{Ex.}{Ex.}
\crefname{counterexample}{Counterex.}{Counterex.}
\crefname{appendix}{App.}{App.}
\crefname{ex}{Ex.}{Ex.}
\crefname{theorem}{Thm.}{Thm.}
\crefname{lemma}{Lemma}{Lemmas}
\crefname{remark}{Remark}{Remarks}
\crefname{section}{Sect.}{Sect.}
\crefname{subsection}{Sect.}{Sect.}
\crefname{subsubsection}{Sect.}{Sect.}
\crefname{line}{Line}{Lines}
\crefname{corollary}{Cor.}{Cor.}
\crefname{figure}{Fig.}{Fig.}
\crefname{enumi}{}{}
\crefname{algorithm}{Alg.}{Alg.}

\newcommand{\paper}[1]{}
\newcommand{\report}[1]{#1}

\AtBeginDocument{
  \theoremstyle{remark}
  \newtheorem{remark}[theorem]{Remark}
}

\AtBeginDocument{%
  }

\copyrightyear{2025}
\acmYear{2025}
\setcopyright{rightsretained}
\acmConference[PPDP '25]{Proceedings of the 27th International Symposium on Principles and Practice of Declarative Programming}{September 10--11, 2025}{Rende, Italy}
\acmBooktitle{Proceedings of the 27th International Symposium on Principles and Practice of Declarative Programming (PPDP '25), September 10--11, 2025, Rende, Italy}
\acmPrice{}
\acmDOI{10.1145/3756907.3756917}
\acmISBN{979-8-4007-2085-7/25/09}

\begin{document}

\title[DPs for Expected Innermost Runtime Complexity of Probabilistic Term
  Rewriting]{Dependency Pairs for Expected Innermost Runtime Complexity and Strong
  Almost-Sure Termination of\\Probabilistic Term Rewriting} 

\author{Jan-Christoph Kassing}
\orcid{0009-0001-9972-2470}
\affiliation{%
  \institution{RWTH Aachen University}
  \city{Aachen}
  \country{Germany}
}
\email{kassing@cs.rwth-aachen.de}

\author{Leon Valentin Spitzer}
\affiliation{%
  \institution{RWTH Aachen University}
  \city{Aachen}
  \country{Germany}
}
\email{leon.spitzer@rwth-aachen.de}

\author{Jürgen Giesl}
\orcid{0000-0003-0283-8520}
\affiliation{%
  \institution{RWTH Aachen University}
  \city{Aachen}
  \country{Germany}
}
\email{giesl@informatik.rwth-aachen.de}

\renewcommand{\shortauthors}{Kassing, Spitzer, and Giesl}

\begin{abstract}
  The dependency pair (DP) framework is one of the most powerful tech\-niques for automatic
  termination and complexity analysis of term rewrite systems.
  While DPs were extended to prove almost-sure termination of
  \emph{probabilistic} term rewrite systems (PTRSs), automatic complexity
  analysis for PTRSs  is largely unexplored. We\linebreak introduce the first DP framework for
  analyzing expected complexity
  and for proving \emph{positive} or \emph{strong} almost-sure termination ($\SAST$)\linebreak
  of innermost rewriting with PTRSs, i.e., finite expected runtime.
  We implemented our framework in the tool \aprove{}
  and demonstrate its power compared to 
  existing techniques for proving $\SAST$.
\end{abstract}




\maketitle

\section{Introduction}

\emph{Probabilistic programming} integrates probabilistic 
branching into traditional computer models, with applications in
many areas \cite{Gordon14}.
Probabilities do not only handle uncertainty in data, 
but they can also be used to decrease the expected runtime of algorithms.

\emph{Term rewriting} \cite{baader_nipkow_1999} is a fundamental concept to
transform and evaluate expressions, which is used, e.g., for
symbolic computation,
automated theorem proving, and
automatic program analysis.
There exist many approaches to prove termination or infer bounds on the
runtime complexity of TRSs, for example,
via \emph{ranking functions} like polynomial interpretations
\cite{lankford1979proving}.
One of the most powerful approaches to analyze termination and runtime complexity of TRSs is
the \emph{dependency pair} (DP) framework, see, e.g.,  \cite{arts2000termination,giesl2006mechanizing,
  noschinski2013analyzing,avanzini_combination_2016,hirokawa2005automating}.
It uses a divide-and-conquer approach to transform termination or complexity problems into
simpler subproblems repeatedly (via \emph{DP processors}).
Indeed, DPs are used in essentially all current termination and complexity tools for TRSs,
e.g., \aprove{}~\cite{JAR-AProVE2017}, \muterm{}~\cite{gutierrez_mu-term_2020},
\natt{}~\cite{natt_sys_2014}, \tct{}~\cite{avanzini_tct_2016}, \ttttwo{}~\cite{ttt2_sys}, etc.

\emph{Probabilistic TRSs} (PTRSs) have been introduced in 
\cite{bournez2005proving,BournezRTA02,avanzini2020probabilistic}.
A
PTRS $\R$ is \emph{almost-surely terminating}
($\AST$) if every evaluation (or ``reduction'')
terminates with probability $1$. 
A strictly stronger notion is \emph{positive} $\AST$ ($\PAST$),
where
every reduction must 
consist of a finite expected number of rewrite steps.
An even stronger notion is \emph{strong} $\AST$ ($\SAST$) 
which requires that for every term $t$, the supremum over the expected number of rewrite steps of all possible reductions starting in $t$ is finite.
It is well known that $\SAST$ implies $\PAST$ and that $\PAST$ implies $\AST$.
To see the difference between $\PAST$ and $\SAST$,
note that a term $t$ can start infinitely many reductions of different expected numbers of
rewrite steps due to non-determinism. So even if all reductions are finite, the supremum
of their expected lengths
may be infinite. 
In this paper, we develop an approach to prove
$\SAST$ for PTRSs under an innermost evaluation strategy where 
we only consider reductions starting with \emph{basic} terms
(which represent the application of an algorithm to data objects).
Moreover, our approach computes upper bounds on the \emph{expected innermost runtime complexity} 
of PTRSs.
\emph{Runtime complexity}
is one of the standard notions of complexity for non-probabilistic TRSs \cite{DBLP:conf/cade/HirokawaM08},
and it was adapted to \emph{expected runtime complexity} for PTRSs in \cite{FoSSaCS-Journal}.

\vspace*{.2cm}

\textbf{Related Work:}
There are numerous techniques to
prove $\maybePAST$ for \emph{imperative programs on numbers}, e.g., \cite{ferrerfioritiProbabilisticTerminationSoundness2015,agrawal2017LexicographicRankingSupermartingales,mciver2017new,dblp:journals/pacmpl/huang0cg19,kaminski2018WeakestPreconditionReasoning,chatterjee2020TerminationAnalysisProbabilistic,kaminski2020ExpectedRuntimeAnalyis,abateLearningProbabilisticTermination2021,chatterjee2023LexicographicProofRules,majumdar2025SoundCompleteProof,moosbruggerAutomatedTerminationAnalysis2021a,meyerInferringExpectedRuntimes2021a,lommen2024ControlFlowRefinementComplexity,avanziniModularCostAnalysis2020,ngoBoundedExpectationsResource2018,majumdarPositiveAlmostSureTermination2024}.
In particular, there also exist several \emph{tools} to analyze $\maybePAST$ and expected costs
for imperative probabilistic programs, e.g., 
\tool{Absynth} \cite{ngoBoundedExpectationsResource2018}, 
\tool{Amber} \cite{moosbruggerAutomatedTerminationAnalysis2021a},
\tool{Eco-Imp} \cite{avanziniModularCostAnalysis2020}, and
\tool{KoAT} \cite{meyerInferringExpectedRuntimes2021a,lommen2024ControlFlowRefinementComplexity}.
In addition,
there are also several related approaches
for recursive programs,
e.g., to analyze probabilistic higher-order programs based on\linebreak types or martingales
\cite{rajani2024ModalTypeTheory, dallago2021IntersectionTypesPositive,
  avanzini2019TypebasedComplexityAnalysis,  
dallago2019ProbabilisticTerminationMonadic,DBLP:conf/lics/Kenyon-RobertsO21,beutner2021probabilistic, dallago2024AlmostSureTerminationBinary}, 
or probabilistic imperative languages with recursion 
\cite{kobayashi2020TerminationProblemProbabilistic}.
However,  only few  approaches analyze
probabilistic programs on recursive \emph{data structures}, e.g.,
\cite{wang2020RaisingExpectationsAutomating, leutgebAutomatedExpectedAmortised2022, batzCalculusAmortizedExpected2023}.
While \cite{batzCalculusAmortizedExpected2023} uses pointers to represent data structures like tables and lists, \cite{ wang2020RaisingExpectationsAutomating, leutgebAutomatedExpectedAmortised2022} consider a 
probabilistic programming language with matching similar to term rewriting
and develop an automatic amortized resource analysis via fixed template potential functions.
However, these works are mostly targeted towards specific data structures, whereas our aim is a
fully automatic approach for
general PTRSs that can model arbitrary data structures.

Currently, the only approach
to analyze $\SAST$ of PTRSs automatically
is the direct application of
polynomial or matrix interpretations \cite{DBLP:journals/jar/EndrullisWZ08} to the whole
PTRS \cite{avanzini2020probabilistic}, implemented in \natt{}.
However, already for non-probabilistic TRSs such a direct
application of orderings is limited in power. For a powerful approach, one
should combine orderings in a modular way, as in the DP framework.

Therefore,
we already adapted the DP framework to the probabilistic setting in
order to prove $\AST$, both for innermost \cite{kassinggiesl2023iAST,FLOPS2024} and full
rewriting \cite{JPK60}.
Moreover, in the non-probabilistic setting, DPs were extended to analyze
complexity instead of just termination, see, e.g., \cite{noschinski2013analyzing,avanzini_combination_2016}.
But up to now there did not exist any DP framework to prove $\SAST$ or $\PAST$, or to infer
bounds on the expected runtime of PTRSs. 
In this paper we show that the DP framework for $\AST$ from \cite{FLOPS2024} which uses
\emph{annotated dependency pairs} can be lifted to a novel DP framework for expected
complexity of PTRSs.

Moreover, in  \cite{FoSSaCS2024} we presented criteria for classes of PTRSs where
$\maybePAST$ for innermost rewriting implies $\maybePAST$ for full rewriting, and in
\cite{FoSSaCS-Journal} these criteria were extended to $\SAST$ and expected runtime
complexity.
  Thus,
  they can also be used in order to
 infer
$\SAST$ and expected runtime complexity for full instead of innermost rewriting
via our novel DP framework,
see \Cref{sec:Evaluation}.

\medskip
\textbf{Main Results of the Paper:}
\begin{itemize}
    \item We develop the first DP framework for $\SAST$ and expected innermost 
    runtime complexity of probabilistic TRSs.
    \item We introduce several processors for our novel DP framework.
    \item To evaluate the power of our novel framework, we implemented it in the tool
      \aprove.
\end{itemize}

\medskip

\textbf{Structure:}
We recapitulate (probabilistic) term rewriting in \Cref{sec:Preliminaries}.
In \Cref{sec:ADPs}, we introduce our novel DP framework for $\SAST$ and expected
runtime complexity.
Afterwards, we present several
processors that can be used in our framework in \Cref{sec:Framework}.
In \Cref{sec:Evaluation}, we give experimental results, also
in comparison to the technique of \cite{avanzini2020probabilistic}.\report{ The proofs of
  all our results can be found in \Cref{Additional Theory}.}\paper{ We refer to \cite{PPDPreport}
  for the proofs of all our results.}

\section{Preliminaries}\label{sec:Preliminaries}

We recapitulate ordinary
and probabilistic
TRSs in \Cref{sec:TRS,sec:PTRS}.

\subsection{Term Rewriting}\label{sec:TRS}

We assume some familiarity with term rewriting~\cite{baader_nipkow_1999},
but recapitulate all needed notions.
For any relation $\to \; \subseteq A \times A$ on a set $A$ and\linebreak $n \in \IN$,
we define $\to^{n}$ as 
$\to^{0} \;=\; \{(a,a) \mid a \in A\}$ and $\to^{n+1} \;=\; \to^{n} \circ\linebreak \to$, 
where ``$\circ$'' denotes composition of relations. Moreover,  $\to^* = \bigcup_{n\in \IN}\to^{n}$, 
i.e., $\to^*$ is the reflexive and transitive closure of $\to$.

The set $\TSet{\Sigma}{\VSet}$ of all \emph{terms} over a finite 
set of \emph{function symbols} $\Sigma = \biguplus_{k \in \IN} \Sigma_k$ 
and a (possibly infinite) set of \emph{variables} $\VSet$
is the smallest set with $\VSet \subseteq \TSet{\Sigma}{\VSet}$, 
and if $f \in \Sigma_k$ and $t_1, \dots, t_k \in \TSet{\Sigma}{\VSet}$ then $f(t_1,\dots,t_k) \in \TSet{\Sigma}{\VSet}$.
If $\Sigma$ and $\VSet$ are clear from the context, then we just write $\TT$ instead of $\TSet{\Sigma}{\VSet}$.
The \emph{arity} of a function symbol $f \in \Sigma_k$ is $k$.
For example, consider the signature $\Sigma_{\tq} = \{\tq, \tstart, \ts, \tz\}$, 
where $\tz$ has arity $0$, $\ts$ has arity $1$, $\tstart$ has arity $2$, and $\tq$ has arity $3$.
Then, for $x \in \VSet$,
$\tstart(\ts(\tz),\ts(\tz))$ and $\tq(\tz, x, x)$ are terms in
$\TSet{\Sigma_{\tq}}{\VSet}$.
For a non-variable term $t = f(t_1, \ldots, t_k)$, we define $\rootsym(t) = f$.
A term without variables is called a \emph{ground} term.
The \emph{size} $|t|$ of a term $t$ is the number of occurrences of
function symbols and variables in $t$, i.e.,
$|t| = 1$ if $t \in \VSet$, and $|t| = 1 + \sum_{j=1}^k |t_j|$ if $t = f(t_1, \ldots, t_k)$.
Thus, $|\tstart(\ts(\tz),\ts(\tz))| = 5$ and $|\tq(\tz, x, x)| = 4$.
A \emph{substitution} is a function $\sigma:\VSet \to \TSet{\Sigma}{\VSet}$ with $\sigma(x) = x$ for all but finitely many $x \in \VSet$, 
and we often write $x\sigma$ instead of $\sigma(x)$.
Substitutions homomorphically extend to terms: if $t=f(t_1,\dots,t_k)\in \TT$ then $t \sigma
= f(t_1\sigma,\dots,t_k \sigma)$.
Thus, for a substitution $\sigma$ with $\sigma(x) = \ts(x)$ we obtain $\tq(\tz, x, x) \sigma
= \tq(\tz, \ts(x), \ts(x))$. 
For any term $t \in \TT$, the set of \emph{positions} $\pos(t)$ 
is the smallest subset of $\IN^*$ satisfying $\varepsilon \in \pos(t)$, 
and if $t=f(t_1,\dots,t_k)$ then for all $1 \leq j \leq k$  and all $\pi \in \pos(t_j)$ we
have $j.\pi \in \pos(t)$.  A position $\pi_1$ is \emph{above}
$\pi_2$ if $\pi_1$ is a prefix of $\pi_2$.
If $\pi \in \pos(t)$ then $t|_{\pi}$ denotes the subterm starting at position $\pi$
and $t[r]_{\pi}$ denotes the term that results from replacing the subterm $t|_{\pi}$ at position $\pi$ with the term $r \in \TT$.
We write $s \trianglelefteq t$ if $s$ is a subterm of $t$ and $s \vartriangleleft t$ 
if $s$ is a \emph{proper} subterm of $t$ (i.e., if $s \trianglelefteq t$ and $s \neq t$).
For example, we have $\pos(\tq(\tz, x, x)) = \{\varepsilon, 1, 2, 3\}$, $\tq(\tz, x, x)|_{2} = x$, 
$\tq(\tz, x, x)[\ts(x)]_{2} = \tq(\tz, \ts(x), x)$, and $\ts(x) \vartriangleleft \tq(\tz, \ts(x), x)$.

A \emph{rewrite rule} $\ell \to r$ is a pair of terms $(\ell, r) \in \TT \times \TT$ 
such that $\VSet(r) \subseteq \VSet(\ell)$ and $\ell \notin \VSet$, 
where $\VSet(t)$ denotes the set of all variables occurring in $t \in \TT$.
A \emph{term rewrite system} (TRS) is a finite set of rewrite rules.
As an example, consider the
following TRS $\R_{\tq}$
that is used to compute the rounded
\textbf{q}uotient of two natural numbers
(represented by the successor function $\ts$ and $\tz$) \cite{arts2000termination}.
\[ \begin{array}{@{}r@{\;}c@{\;}l@{\quad}r@{\;}c@{\;}l}
  \tstart(x,y) &\to& \tq(x,y,y) & \tq(x,\tz,\ts(z)) &\to& \ts(\tq(x,\ts(z),\ts(z))) \\
\tq(\ts(x),\ts(y),z) &\to& \tq(x,y,z) &\tq(\tz,\ts(y),\ts(z)) &\to& \tz
\end{array}\]
A TRS $\R$ induces a \emph{rewrite relation} ${\to_{\R}} \subseteq \TT
\times \TT$ on terms where $s \to_{\R} t$ holds if there is a position $\pi \in \pos(s)$, 
a rule $\ell \to r \in \R$, and a substitution $\sigma$ 
such that $s|_{\pi}=\ell\sigma$ and $t = s[r\sigma]_{\pi}$.
Let $\NF_{\R}$ denote the set of all terms that are in \emph{normal form}
w.r.t.\ $\to_{\R}$, i.e., $s \in \NF_{\R}$ if there is no term $t$ with $s \to_{\R}  t$. 

A rewrite step $s \to_{\R} t$ is an \emph{innermost} step (denoted $s \itor t$) if
 $\ell\sigma \in \ANF_{\R}$, where $\ANF_{\R}$ is the set of all terms in  \emph{argument
normal form} w.r.t.\ $\to_{\R}$, i.e., $t \in \ANF_{\R}$ iff $t' \in \NF_{\R}$ for all
proper subterms $t'\vartriangleleft t$.
For example, $\tstart(\ts(\tz),\ts(\tz)) \in \ANF_{\R_{\tq}}$ and
$\tstart(\tstart(\tz,\tz),\linebreak
\ts(\tz)) \notin \ANF_{\R_{\tq}}$ since it contains the proper
subterm $\tstart(\tz,\tz)$, which is not a normal form.
Thus, the rewrite step $\tstart(\tstart(\tz,\tz),\linebreak \ts(\tz)) \to_{\R_{\tq}}
\tq(\tstart(\tz,\tz),\ts(\tz),\ts(\tz))$ is not an innermost rewrite step, as it takes
place at the position of a subterm that is not in argument normal form.
The TRS $\R_{\tq}$ computes $\lfloor \frac{n}{m} \rfloor$ when
starting with the term $\tstart(\ts^n(\tz),\ts^m(\tz))$, where $\ts^n(\ldots)$ denotes $n \in \IN$
successive $\ts$-function symbols.
For instance,
we have $\lfloor \frac{1}{1} \rfloor = 1$ and
\[ \begin{array}{@{}l@{\;\;}c@{\;\;}l@{\;\;}c@{\;\;}l}
  \tstart(\ts(\tz),\ts(\tz)) &\ito_{\R_{\tq}}& \tq(\ts(\tz),\ts(\tz),\ts(\tz)) &\ito_{\R_{\tq}}& \tq(\tz,\tz,\ts(\tz))
  \\  
  &\ito_{\R_{\tq}} &  \ts(\tq(\tz,\ts(\tz),\ts(\tz))) &\ito_{\R_{\tq}} &\ts(\tz).
  \end{array}\]
Essentially, innermost rewriting corresponds to a call-by-value strategy. Note that
innermost rewriting allows us to  reduce arbitrary subterms in argument normal form,
i.e., there is no fixed preference among
several parallel subterms in argument normal form, as in a leftmost-innermost strategy.

Already for non-probabilistic TRSs, the techniques for termination and complexity
analysis of \emph{innermost} rewriting are significantly stronger than the ones for ``full''
rewriting where arbitrary rewrite sequences are allowed (the same holds for the
probabilistic DP framework for $\AST$ in \cite{kassinggiesl2023iAST,FLOPS2024,JPK60}).  Moreover, 
innermost evaluation is the standard strategy for most programming languages.
Hence, in the remainder, we restrict ourselves to innermost rewriting.

The
  \emph{derivation height} \cite{DBLP:conf/rta/HofbauerL89} of a term $t$
is the length of the longest $\itor$-sequence starting with $t$, i.e.,
\[\idh_\R(t) = \sup \{ n \in \mathbb{N} \mid 
\exists t' \in \Trans \text{ such that } t \itorn t' \} \; \in \IN \cup \{\omega\}.\] 
For example, 
$\idh_{\R_{\tq}}(\tstart(\ts(\tz),\ts(\tz))) = 4$.
We have  $\idh_\R(t) = 0$ iff
$t \in \NF_\R$ and
$\idh_\R(t) = \omega$ iff $t$ starts an infinite sequence of $\itor$-steps, as
we restricted ourselves to finite TRSs.

We decompose the signature $\Sigma = \SignatureD_\R \uplus \SignatureC_\R$
into \emph{defined symbols} $\SignatureD_\R = \{ \rootsym(\ell) \mid \ell\!\to\! r \in \R \}$
and  \emph{constructors} $\SignatureC_\R$.
If $\R$ is clear from\linebreak the
context, we just write $\SignatureC$ and $\SignatureD$.
A term $f(t_1, \ldots , t_k)$ is  \emph{basic} 
if $f \in \SignatureD_\R$ and $t_1, \ldots, t_k \in \Trans(\SignatureC_\R, \Var)$, i.e.,
$t_1, \ldots , t_k$
do not contain\linebreak defined symbols.
Thus, basic terms represent an algorithm $f$ that is applied to data $t_1, ..., t_k$.
So for $\R_{\tq}$, $\tq(\tz,x,x)$ is basic, but
$\tq(\tq(\tz,x,x),\linebreak x,x)$ is not.
Let $\TB_\R$ denote the set of basic terms for the TRS $\R$.

The \emph{runtime complexity} $\irc_\R$ is a function that maps any $n \in \IN$ to
the maximum derivation height for basic terms of size $\leq n$.

\begin{definition}[Runtime Complexity, $\irc_\R$ \cite{DBLP:conf/cade/HirokawaM08}]\label{def:innermost-runtime-complexity}
    For a TRS $\R$, its \textit{runtime complexity function} $\irc_{\R}:\IN
    \to \IN \cup \{ \omega \}$
    is:
    \[
    \irc_\R(n) = \sup \{ \idh_\R(t) \mid t \in \TB_\R, |t| \leq n \}
    \]
\end{definition}

Given  a TRS $\R$, our goal is to determine an upper bound on the \textit{asymptotic
  complexity} of the function $\irc_\R$.

\begin{definition}[Asymptotic Complexities]
    We consider a set of complexities $\F{C} = \{\Pol_0, \Pol_1, \Pol_2, \ldots, \Exp,
    \DExp, \Fin, \omega\} $ 
    with the order $ \Pol_0 \sqsubset \Pol_1 \sqsubset \Pol_2 \sqsubset \ldots \sqsubset
    \Exp \sqsubset   \DExp \sqsubset \Fin \sqsubset \omega$, where
 $\sqsubseteq$ is the reflexive closure of $\sqsubset$.
    For any function $f : \mathbb{N} \to \mathbb{N} \cup \{\omega\}$, 
    we define its \emph{complexity} $\iota(f) \in \mathfrak{C}$ as follows:
    \[
    \iota(f) =
    \begin{cases} 
        \, \Pol_a & \text{if $a\!\in\!\mathbb{N}$ is the smallest number with $f(n)\!\in\!\O(n^a)$} \\
        \, \Exp & \text{if no such $a$ exists, but there is a} \\[-.1cm]
        & \text{polynomial $\pl(n)$ such that $f(n) \in
          \O(2^{\pl(n)})$}\\
    \, \DExp \!\!\!\!& \text{if no such polynomial exists, but there is a} \\[-.1cm]
        & \text{polynomial $\pl(n)$ such that $f(n) \in
      \O(2^{2^{\pl(n)}})$}\\
    \, \Fin & \text{if no such polynomial exists,}\\[-.1cm]
        &\text{but there is no $n \in \IN$ with
          $f(n) =  \omega$} \\
        \, \omega & \text{if there is an $n \in \IN$ with $f(n) = \omega$}
    \end{cases}
    \]
    For any TRS $\R$, its (asymptotic) \emph{runtime complexity} $\iota_\R$ is $\iota(\irc_\R)$.
     \end{definition}

The TRS $\R_{\tq}$ has linear runtime complexity,
i.e., $\iota_{\R_{\tq}} = \iota(\irc_{\R_{\tq}}) = \Pol_1$. For example,
any rewrite sequence starting with the basic term $\tstart(\ts^n(0), \ts^m(0))$ 
has at most $2n + 2$ rewrite steps.

Finally, we 
recapitulate a first approach to prove termination 
and to infer upper bounds on the runtime complexity via 
polynomial\footnote{In this paper, we focus on polynomial interpretations with natural coefficients for simplicity, 
but our results can be extended to other interpretations 
where one can define an addition and an expected value operation,
e.g., matrix interpretations \cite{DBLP:journals/jar/EndrullisWZ08}.} interpretations.
A \emph{polynomial interpretation}  is a $\Sigma$-algebra $\I: \Sigma \to \IN[\Var]$ 
that maps every function symbol $f \in \Sigma_k$ to a polynomial $\I_f$ over $k$ variables
with natural coefficients. As usual, $\I$
is homomorphically extended to terms.
$\I$ is \emph{monotonic} if $x > y$ implies $\I_f(\ldots,x,\ldots) > \I_f(\ldots,y,\ldots)$ 
for all $f \in \Sigma$ and $x,y \in \IN$.
We call $\I$ a \textit{complexity polynomial interpretation} (CPI)
if 
for all constructors $f \in \SignatureC$ we have $\I_{f}(x_1, \ldots, x_k) = a_1 x_1 + \ldots + a_k x_k + b$,
where $b \in \IN$ and $a_j \in \{0,1\}$
for all $1 \leq j \leq k$.\footnote{For monotonic CPIs one must have $a_j =
1$, but in \Cref{sec:red-pair} we will consider weakly monotonic CPIs, where $a_j \in \{0,1\}$ is
possible.}
While arbitrary monotonic polynomial interpretations can be used to 
prove termination of TRSs \cite{lankford1979proving},
monotonic CPIs are needed to infer a polynomial runtime bound from 
such a termination proof \cite{bonfante2001algorithms,DBLP:conf/rta/HofbauerL89}.
So the reason for
  the name CPI is that only such polynomial interpretations give rise to polynomial upper bounds when analyzing runtime complexity.

More precisely, if
there is a monotonic  polynomial interpretation $\I$ such that
$\I(\ell) > \I(r)$  holds
for every rule $\ell \to r \in \R$, then the TRS
$\R$ is terminating and $\iota_\R \sqsubseteq \DExp$,
and if all constructors are interpreted by linear polynomials, then we have
$\iota_\R \sqsubseteq \Exp$\report{, see
  \Cref{subsec:poly-interpret}}\paper{, see  \cite{PPDPreport}}.

But if $\I$ is a monotonic CPI, 
we even have
$\iota_\R \sqsubseteq \Pol_a$
if for all $f \in \SignatureD$, the polynomial $\I_f$
has at most degree $a$.
The reason is that this implies $\I(t) \in \O(|t|^a)$ for all basic ground terms
$t \in \TB_\R$.
(More\linebreak  precisely, the function  that maps any $n \in \IN$ to $\sup \{ \I(t) \mid t \in
\TB_\R,\linebreak \VSet(t) = \emptyset, 
|t| \leq n  \}$ is in $\O(n^a)$.)
Since every
rewrite step decreases the interpretation of the term by at least $1$,
the length of each rewrite sequence starting with a basic ground term of size $\leq n$ 
is  in  $\O(n^a)$.

This direct approach is only feasible for simple examples like the
TRS $\R_{\tplus} = \{\tplus ( \tz , y ) \to y, \;
\tplus ( \ts(x) , y ) \to \ts( \tplus ( x , y ) )\}$,
computing the addition of two natural numbers.
Let $\I$ be a monotonic
CPI with $\I_{\tz} = 0$, $\I_{\ts}(x) = x + 1$, and $\I_{\tplus}(x,y) = 2x + y + 1$.
We have $\I(\tplus(\ts(x), y)) = 2x+y+3 > 2x+y+2 = \I(\ts(\tplus(x,y)))$,
and $\I(\tplus ( \tz , y )) = y + 1 > y = \I(y)$.
Thus,  $\R_{\tplus}$ is terminating
and has at most linear runtime complexity, i.e., $\iota_{\R_{\tplus}} \sqsubseteq \Pol_1$
(in fact, we have $\iota_{\R_{\tplus}} = \Pol_1$).
To automate this approach, one can use SMT solvers  to search for a suitable CPI $\I$.
However, such a direct application of polynomials fails
to prove termination or to infer a
polynomial upper bound on $\iota_{\R_{\tq}}$.
In contrast, 
 $\iota_{\R_{\tq}} \sqsubseteq \Pol_1$
can be proved
by more elaborate techniques like dependency pairs, see, e.g., \cite{noschinski2013analyzing}.
Indeed, we will show how to analyze
a probabilistic version of $\R_{\tq}$ with our novel DP framework.

\subsection{Probabilistic Rewriting}\label{sec:PTRS}

A probabilistic TRS has finite multi-distributions on the right-hand sides of its rewrite rules.
A finite \emph{multi-distribution} $\mu$ on a set $A \neq \emptyset$ is a 
finite multiset of pairs $(p:a)$, where $0 < p \leq 1$ is a probability 
and $a \in A$, such that $\sum _{(p:a) \in \mu}p = 1$.
Let $\FDist(A)$ denote the set of all finite multi-distributions on $A$.
For $\mu\in\FDist(A)$, its \emph{support} is the multiset $\Supp(\mu)\!=\!\{a \mid (p\!:\!a)\!\in\!\mu$ for some $p\}$.
A \emph{probabilistic rewrite rule} $\ell \to \mu$ is a pair $(\ell, \mu) \in \TT \times \FDist(\TT)$ 
such that $\ell \not\in \VSet$ and $\VSet(r) \subseteq \VSet(\ell)$ for every $r \in \Supp(\mu)$.
A \emph{probabilistic TRS} (PTRS)\linebreak is a finite set of probabilistic rewrite rules.
Similar to TRSs, a PTRS $\R$ induces a \emph{(probabilistic) rewrite relation} 
${\to_{\R}} \subseteq \TT \times \FDist(\TT)$ where
$s \to_{\R} \{p_1:t_1, \ldots, p_k:t_k\}$  
if there is a position $\pi \in \pos(s)$, a rule $\ell \to \{p_1:r_1, \ldots, p_k:r_k\} \in \R$, 
and a substitution $\sigma$ such that $s|_{\pi}=\ell\sigma$ 
and $t_j = s[r_j\sigma]_{\pi}$ for all $1 \leq j \leq k$.
We call $s \to_{\R} \mu$ an \emph{innermost} rewrite step (denoted $s \ito_{\R} \mu$)
if $\ell\sigma \in \ANF_{\R}$.
Consider the PTRS $\R_{\tgeo}$ with the only rule $\tgeo(x) \to
\{\nicefrac{1}{2}:\tgeo(\ts(x)), \; \nicefrac{1}{2}:x \}$. When starting with the term
$\tgeo(\tz)$, it
computes the representation $\ts^k(\tz)$ of the number $k \in \IN$ with a probability of $(\nicefrac{1}{2})^{k+1}$, 
i.e., a geometric distribution.

To track innermost rewrite sequences with their probabilities,
we consider (potentially infinite) \emph{rewrite sequence trees (RSTs)}~\cite{FLOPS2024}.
The \pagebreak[3] nodes $v$ of an $\R$-RST are labeled by pairs $(p_v:t_v)$ of a
probability $p_v \in (0,1]$ and a term $t_v$,  where
the root always has the probability $1$. 
For each node $v$ with successors $w_1, \ldots, w_k$, 
the edge relation represents an innermost rewrite step,
i.e., $t_v \ito_{\R} \{\tfrac{p_{w_1}}{p_v}:t_{w_1}, \ldots, \tfrac{p_{w_k}}{p_v}:t_{w_k}\}$.
For an $\R$-RST $\F{T}$,
$V^{\F{T}}$ denotes its set of nodes, $\rootterm(\F{T})$ is the term at its root, 
and $\ctleaf^{\F{T}}$ denotes its set of leaves.
A finite RST for ${\R_{\tgeo}}$ is shown 
in \Cref{fig:RST_example}, but of course this RST could be extended further to an infinite one.

\begin{figure}
	\centering
    \begin{tikzpicture}
        \tikzstyle{adam}=[thick,draw=black!100,fill=white!100,minimum size=4mm,
          shape=rectangle split, rectangle split parts=2,rectangle split horizontal,font={\footnotesize}]
        \tikzstyle{empty}=[rectangle,thick,minimum size=4mm]
        
        \node[adam] at (0, 0)  (a) {$1$\nodepart{two}$\tgeo(\tz)$};

        \node[adam] at (-2, -0.5)  (b) {$\nicefrac{1}{2}$\nodepart{two}$\tgeo(\ts(\tz))$};
        \node[adam,label=below:{\footnotesize $\quad \NF_{\R_{\tgeo}}$}] at (2, -0.5)  (c) {$\nicefrac{1}{2}$\nodepart{two}$\tz$};

        \node[adam] at (-4, -1.3)  (d) {$\nicefrac{1}{4}$\nodepart{two}$\tgeo(\ts(\ts(\tz)))$};
        \node[adam,label=below:{\footnotesize $\quad \NF_{\R_{\tgeo}}$}] at (0, -1.1)  (e) {$\nicefrac{1}{4}$\nodepart{two}$\ts(\tz)$};
        
        \draw (a) edge[->] (b);
        \draw (a) edge[->] (c);
        \draw (b) edge[->] (d);
        \draw (b) edge[->] (e);
    \end{tikzpicture}\vspace*{-.3cm}
    \caption{\vspace*{-.3cm}${\R_{\tgeo}}$-Rewrite sequence tree starting with $\tgeo(\tz)$}\label{fig:RST_example}
\end{figure}

A PTRS $\R$ is \emph{almost-surely terminating} ($\AST$) if $\sum_{v \in \ctleaf^{\F{T}}}
p_v = 1$ holds for all $\R$-RSTs $\F{T}$, i.e., if the probability of termination is always
1. This notion of $\AST$ for PTRSs is equivalent to the ones in
\cite{avanzini2020probabilistic,bournez2005proving,kassinggiesl2023iAST}
where AST is defined via a lifting of $\ito_{\R}$
to multisets or via stochastic processes.
However,  $\AST$ is not sufficient to guarantee that 
the expected runtime complexity of a PTRS is finite. To define this concept
formally, we first introduce
the \emph{expected derivation length} of
an $\R$-RST $\F{T}$
as
\[\textstyle \edl(\F{T}) = \sum_{v \in V^{\F{T}} \setminus \ctleaf^{\F{T}}} p_v.\]
So $\edl(\F{T})$ adds up the probabilities of all
rewrite steps in $\F{T}$.
Thus, for the RST $\F{T}$ in \cref{fig:RST_example} we obtain $\edl(\F{T}) = 1 +
\nicefrac{1}{2} = \nicefrac{3}{2}$, i.e., in expectation we perform $\nicefrac{3}{2}$ rewrite steps in $\F{T}$.
Then a  PTRS $\R$ is \emph{positively almost-surely terminating} ($\PAST$)
if
$\edl(\F{T})$ is finite for all $\R$-RSTs $\F{T}$.
Again, this notion of $\PAST$ for PTRSs is equivalent to the
ones in 
\cite{avanzini2020probabilistic,bournez2005proving}.
Clearly, $\PAST$ implies $\AST$, but not vice versa
(e.g., a PTRS with the rule $\tg \to
\{ \nicefrac{1}{2}: \tc(\tg,\tg), \; \nicefrac{1}{2}:\tz \}$ 
which represents a symmetric random walk is $\AST$, but not $\PAST$). 

To adapt the notions for complexity from TRSs to PTRSs,
recall that in the non-probabilistic setting, the \emph{derivation height} of a term $t$
does not consider a fixed rewrite sequence, but
all possible rewrite sequences starting with $t$ and takes the supremum of
their lengths. Similarly, 
while $\edl$ considers a fixed RST $\F{T}$,
for the \emph{expected derivation height} of a term $t$, we consider all possible
RSTs with root $t$ and take the supremum of their expected derivation lengths. 
So the  \emph{expected derivation height} of a  term $t$ is 
\[\edh_\R(t) = \sup\{\edl(\F{T}) \mid \F{T} \text{ is an $\R$-RST with $\rootterm(\F{T}) =
t$}\}.\]
Now we can
adapt the notion of 
 runtime complexity 
 to  PTRSs.

\begin{definition}[Expected  Runtime Complexity, $\eirc_{\R}$]\label{def-signature}
    For a PTRS $\R$, its \textit{expected runtime complexity 
    function} $\eirc_{\R}:\IN \to \IN \cup \{ \omega \}$ is:
    \[
    \eirc_{\R}(n) = \sup \{ \edh_\R(t) \mid t \in \TB_\R, |t| \leq n \}
    \]
 Moreover, we  define $\R$'s \emph{expected runtime complexity} $\iota_\R$ as $\iota(\eirc_\R)$.
\end{definition}

A PTRS $\R$ is \emph{strongly} or \emph{bounded almost-surely terminating} ($\SAST$)
if $\edh_\R(t)$ is finite for every term $t$. So in contrast to $\PAST$, here one requires
a finite bound on the expected derivation
lengths of all RSTs with the same term at the root. Such notions of $\SAST$ were defined in, e.g., 
\cite{majumdarPositiveAlmostSureTermination2024,avanzini2020probabilistic,dblp:conf/vmcai/fuc19}.
$\SAST$ implies $\PAST$,
but not vice versa (a PTRS with finitely many rules that is $\PAST$ but not
$\SAST$ is given in \cite{FoSSaCS-Journal}). However, as also shown in \cite{FoSSaCS-Journal}, $\SAST$ and $\PAST$ are
``almost always'' equivalent for finite PTRSs (e.g.,  whenever \pagebreak[3] the
signature contains a function symbol of arity $\geq 2$). 

\begin{table}
  \centering
  \small
    \begin{tabular}{@{}cc@{}}
    \toprule
    \textbf{Notation} & \textbf{Description} \\
    \midrule
    $\TT$ & \makecell[c]{Set of all terms} \\
    \midrule
    $\TB_{\R}$ & \makecell[c]{Set of all basic terms for a PTRS $\R$} \\
    \midrule
    $\edl(\F{T})$ & \makecell[c]{Expected derivation length of an RST $\F{T}$ \\
    (Sum of all probabilities of inner nodes)}\\
    \midrule
    $\edh_\R(t)$ & \makecell[c]{Expected derivation height of a term $t$ w.r.t.\ a PTRS
      $\R$\\ (Supremum over all $\edl(\F{T})$ for $\R$-RSTs $\F{T}$ starting with $t$)} \\
    \midrule
    $\eirc_{\R}(n)$ & \makecell[c]{Expected runtime complexity function\\
(Supremum over all $\edh_\R(t)$ for basic terms $t$ 
   of size $\leq n$)} \\
    \midrule
    $\iota_{\R}$ & \makecell[c]{Expected runtime complexity $\iota(\eirc_{\R})$ of a PTRS $\R$} \\
    \midrule
    $\AST$ & \makecell[c]{Sum of all probabilities of leaves is 1 for every $\R$-RST} \\
    \midrule
    $\PAST$ & \makecell[c]{$\edl(\F{T})$ is finite for every $\R$-RST $\F{T}$} \\
    \midrule
    $\SAST$ & \makecell[c]{$\edh_\R(t)$ is finite for every basic term $t$} \\
    \bottomrule
  \end{tabular}
\caption{Key notions}\label{fig:key_notions}
\vspace{-25pt}
\end{table}

As mentioned, in this paper we restrict ourselves to innermost reductions that
start with \emph{basic} terms. So we regard a PTRS $\R$ to be\linebreak
$\SAST$ if 
$\edh_\R(t)$ is finite for all \emph{basic} terms $t$, or equivalently, 
$\eirc_{\R}(n)\linebreak \neq \omega$ for all $n \in \IN$. Thus,
we use the following
definition of $\SAST$.

\begin{definition}[Strong Almost-Sure Termination, $\SAST$]\label{def:SAST}
    A PTRS $\R$ is called \emph{strongly almost-surely terminating} ($\SAST$)
      if $\iota_{\R} \sqsubseteq \Fin$.
\end{definition}

  \vspace*{-.1cm}
We summarized all introduced key notions in \Cref{fig:key_notions}.
  \vspace*{-.1cm}

\begin{example}[Leading Examples]\label{ex:LeadingExamples}
Consider the following PTRS  $\R_1$, which is based
  on the  previously defined systems $\R_{\tq}$ and $\R_{\tgeo}$.

  \vspace*{-.3cm}
  
{\small
  \begin{align*} 
 \R_1: \quad\qquad\qquad   \tstart(x,y) &\to \{1:\tq(\tgeo(x),y,y)\}\\[-.08cm]
    \tgeo(x) &\to \{\nicefrac{1}{2}:\tgeo(\ts(x)), \; \nicefrac{1}{2}:x\}\\[-.08cm]
     \tq(\ts(x),\ts(y),z) &\to \{1:\tq(x,y,z)\}\\[-.08cm]
    \tq(x,\tz,\ts(z)) &\to \{1:\ts(\tq(x,\ts(z),\ts(z)))\}\\[-.08cm]
    \tq(\tz,\ts(y),\ts(z)) &\to \{1:\tz\}
\end{align*}}

When starting with $\tstart(\ts^n(\tz),\ts^m(\tz))$,
$\R_1$ computes $\lfloor \frac{n + \geo(\tz)}{m} \rfloor$,
i.e., it first increases $n$ according to a geometric distribution,
and then computes the quotient like $\R_{\tq}$.
Thus, $\iota_{\R_1} = \Pol_1$, since $\R_{\tq}$ has linear runtime complexity
and the geometric distribution only increases $n$ by $2$ in expectation.
So in particular, $\R_1$ is $\SAST$.

Moreover, consider the PTRS $\R_2$ with the rules:

\vspace*{-.3cm}

{\small \begin{align*}
 \R_2: \quad\qquad\qquad  \tstart &\to \{1:\tf(\tgeo(\tz))\}\\[-.08cm]
    \tgeo(x) &\to \{\nicefrac{1}{2}:\tgeo(\ts(x)), \; \nicefrac{1}{2}:x\}\\[-.08cm]
    \tf(\ts(x)) &\to \{1:\tf(\tc(x,x))\}\\[-.08cm]
    \tf(\tc(x,y)) &\to \{1:\tc(\tf(x),\tf(y))\}\!
  \end{align*}
}

The two $\tf$-rules have exponential runtime complexity,
as a reduction starting in $\tf(\ts^n(\tz))$ creates a full binary tree
of height $n$ and visits every inner node once.
When beginning with the term $\tstart$, $\R_2$ 
first generates the term $\tf(\ts^k(\tz))$
with probability
  $(\nicefrac{1}{2})^{k+1}$ and then takes at least $2^k$ steps to terminate.
The
expected derivation length of the corresponding RST is at least $\sum_{k=0}^{\infty}
(\nicefrac{1}{2})^{k+1} \cdot 2^k = \sum_{k=0}^{\infty} \nicefrac{1}{2} = \omega$. 
Hence,  $\iota_{\R_2} = \omega$, i.e., $\R_2$ is not $\SAST$.
\end{example}

\section{Annotated Dependency Pairs}\label{sec:ADPs}

In \Cref{sec:ADP-problems}
we define \emph{annotated dependency pairs}. While such dependency pairs
were used to prove $\AST$ in
\cite{FLOPS2024} and relative termination of TRSs in \cite{kassing2024DependencyPairFramework}, 
we now develop a new criterion in order to use them for complexity analysis of PTRSs.
Afterwards, in \Cref{sec:ADP-Framework} we introduce the general idea of our novel
framework in order to derive upper bounds on the expected runtime complexity of PTRSs.

\subsection{ADP Problems}\label{sec:ADP-problems}

The core idea of the  \emph{dependency pair framework} for termination of TRSs
\cite{arts2000termination,giesl2006mechanizing} is the following: \emph{a function is terminating iff 
the arguments of each recursive function call are decreasing w.r.t.\ some well-founded ordering}.
Hence, for every defined symbol $f \in \SignatureD$ one introduces a fresh
\emph{tuple} or \emph{annotated}
symbol $f^{\sharp}$ that is used to compare the arguments of two successive calls of $f$.
Let $\Sigma^\sharp = \Sigma \cup \SignatureA$ with $\SignatureA = \{f^\sharp \mid f \in
\D\}$, and for any
$\Sigma' \subseteq \Sigma \cup \VSet$,
let $\pos_{\Sigma'}(t)$ be all positions of $t$ with symbols or variables from   $\Sigma'$.
For any $t = f(t_1, \ldots,t_k) \in \TT$ with $f \in \SignatureD$, let $t^{\sharp} = f^{\sharp}(t_1,\ldots,t_k)$.
For termination analysis, one considers each function call 
in a right-hand side of a rewrite rule on its own, i.e.,  for each rule $\ell \to r$ 
with $\pos_{\D} = \{\pi_1, \ldots, \pi_n\}$, one obtains
$n$  \emph{dependency pairs}
$\ell^\sharp \to r|_{\pi_i}^\sharp$ for all $1 \leq i \leq n$.
However, for complexity analysis, one has to consider all function calls in a right-hand
side simultaneously. Thus,
when adapting DPs for complexity analysis
in
\cite{noschinski2013analyzing},
a single \emph{dependency tuple} (DT) $\ell^\sharp \to [r|_{\pi_1}^\sharp, \ldots, r|_{\pi_n}^\sharp]$ is
constructed instead of the $n$ dependency pairs.
By analyzing the dependency tuples (together with the original rewrite rules),
\cite{noschinski2013analyzing} presented 
a modular \emph{DT framework} that can be used to infer an upper bound on the runtime
complexity.
However, in contrast to the \emph{chain criterion}
of dependency pairs (which states that termination of a TRS is equivalent to the absence of
infinite chains of DPs), the 
\emph{chain criterion}
of this dependency tuple framework yields an over-approximation. More precisely,
the upper bounds on
the runtime complexity obtained
via dependency tuples are only tight 
for \emph{confluent} TRSs.

Recently, we introduced \emph{annotated dependency pairs} (ADPs)
to analyze \emph{almost-sure termination} of PTRSs \cite{FLOPS2024}.
We now show that by
using ADPs instead of dependency tuples, the corresponding chain criterion for (expected)
complexity becomes an equivalence again, i.e., it can be used to 
compute \emph{tight} complexity bounds (irrespective of confluence).
Instead of extracting the function calls of right-hand sides
and coupling them together in a fresh dependency tuple, in ADPs
we annotate these function calls in the original rewrite rule directly, i.e., we
keep the original structure of the rule.

\begin{definition}[Annotations]\label{def:annotations}
     For $t \in \TT^\sharp = \Trans(\Sigma^\sharp, \Var)$ and a set of positions $\Phi
     \subseteq \pos_{\SignatureD \cup \SignatureA}(t)$, let $\sharp_{\Phi}(t)$  
    be the variant of $t$ where the symbols at positions from $\Phi$ in $t$ 
    are annotated and all other annotations are removed.
So $\pos_{\SignatureA}(\sharp_{\Phi}(t)) = \Phi$ and $\sharp_{\emptyset}(t)$ removes all
annotations from $t$.
 We often write $\annoD(t)$ instead of $\sharp_{\pos_\D(t)}(t)$ to annotate all defined symbols in $t$,
    and $\flat(t)$ instead of $\sharp_{\emptyset}(t)$, where
we extend $\flat$ to multi-distributions, rules, and sets of rules 
    by removing the annotations of all occurring terms. 
    Moreover, $\flat^{\uparrow}_{\pi}(t)$ results from removing 
    all annotations from $t$ that are strictly above the position $\pi$.
     We write $t \trianglelefteq_{\sharp}^\pi s$ if  $\pi \in \pos_{\SignatureA}(s)$ 
    and $t = \flat(s|_{\pi})$, i.e., $t$ results from a subterm of $s$ 
    with annotated root symbol by removing its annotation. If
    $\pi$  is not of interest, we just write
    $t \trianglelefteq_{\sharp} s$.
    We often write $\tF$ instead of $\tf^\sharp$ for $\tf \in \SignatureD$
(e.g., $\tGeo$ instead of $\tgeo^\sharp$).
\end{definition}

\begin{example}[Annotations]\label{ex:annotations}
If $\tf\!\in\!\SignatureD$, then we have $\sharp_{\{1\}}(\tf(\tf(x)))\!=\linebreak \sharp_{\{1\}}(\tF(\tF(x)))
= \tf(\tF(x))$, \pagebreak[3] $\annoD(\tf(\tf(x))) = \sharp_{\{\varepsilon,1\}}(\tf(\tf(x))) =\tF(\tF(x))$,
$\flat(\tF(\tF(x))) = \tf(\tf(x))$, $\flat^{\uparrow}_{1}(\tF(\tF(x))) = \tf(\tF(x))$, and
$\tf(x) \, \trianglelefteq_{\sharp} \, \tf(\tF(x))$. 
\end{example}

The annotations indicate which function calls 
need to be regarded for complexity analysis. To transform a PTRS into ADPs, initially we
annotate all defined symbols in the right-hand sides of rules,
since all function calls need to be considered at the start of our analysis.
The left-hand side of an ADP is just
the left-hand side of the original rule (i.e., in contrast to the DPs of \cite{arts2000termination,giesl2006mechanizing},
we do not annotate symbols in left-hand sides).
The DP and the DT framework work on pairs $\langle \P, \R \rangle$, 
where $\R$ contains the original rewrite rules and $\P$ is
the set of dependency pairs or tuples.
In contrast, ADPs already represent the original rewrite rules themselves. 
We simply add a Boolean flag $m \in \{\ttrue, \tfalse\}$ to
indicate whether we still need to consider the corresponding original rewrite rule for our analysis.
Initially, the flag is $\ttrue$ for all ADPs.

\begin{definition}[Annotated Dependency Pairs]\label{def:canonical-annotated-dps}
    An \emph{annotated dependency pair} (ADP) has the form $\ell \to \{p_1 : r_1, \ldots, p_k :
    r_k\}^m$, where $\ell \in \TT$ with $\ell \not\in \Var$, $m \in
    \{\ttrue, \tfalse\}$, and for all $1 \leq j \leq k$ we have $r_j \in \TT^\sharp$ with
    $\Var(r_j) \subseteq \Var(\ell)$.

The \textit{canonical ADP} of
 a probabilistic rule $\ell \to \mu = \{p_1\!: r_1, \dots, p_k\!: r_k\}$
   is 
    $\ADPair{\ell \to \mu} = \ell \to \{p_1 : \annoD(r_1), \dots, p_k :
    \annoD(r_k)\}^{\ttrue}$.
    The ca\-no\-nical ADPs of a PTRS $\R$ are $\ADPair{\R} = \{ \ADPair{\ell \to \mu} \mid \ell \to \mu \in \R\}$.
\end{definition}

For a set of ADPs,
the defined symbols, constructors, and basic terms are defined as for a TRS,
because the left-hand sides of the ADPs are the left-hand sides of the original rewrite rules.

\begin{example}[Canonical Annotated Dependency Pairs]\label{ex:canonical-adps}
  The canonical ADPs $\ADPair{\R_1}$ and $\ADPair{\R_2}$ of $\R_1$ and $\R_2$
from \Cref{ex:LeadingExamples}  are: 

\vspace*{-.3cm}

{\small 
{\allowdisplaybreaks\begin{align}
\ADPair{\R_1}: \quad\qquad\qquad
      \tstart(x,y) &\to \{1:\tQ(\tGeo(x),y,y)\}^{\ttrue} \label{eq:ADP-1-1}\\[-.08cm]
        \tgeo(x) &\to \{\nicefrac{1}{2}:\tGeo(\ts(x)),\nicefrac{1}{2}:x\}^{\ttrue} \label{eq:ADP-1-2}\\[-.08cm]
        \tq(\ts(x),\ts(y),z) &\to \{1:\tQ(x,y,z)\}^{\ttrue} \label{eq:ADP-1-4}\\[-.08cm]
        \tq(x,\tz,\ts(z)) &\to \{1:\ts(\tQ(x,\ts(z),\ts(z)))\}^{\ttrue} \label{eq:ADP-1-5}\\[-.08cm]
        \tq(\tz,\ts(y),\ts(z)) &\to \{1:\tz\}^{\ttrue} \label{eq:ADP-1-3}
        \\
      \ADPair{\R_2}: \quad\qquad\qquad\qquad   \tstart &\to \{1:\tF(\tGeo(\tz))\}^{\ttrue}
      \nonumber\\[-.08cm]  
         \tgeo(x) &\to \{\nicefrac{1}{2}:\tGeo(\ts(x)),\nicefrac{1}{2}:x\}^{\ttrue}
         \nonumber\\[-.08cm]  
         \tf(\ts(x)) &\to \{1:\tF(\tc(x,x))\}^{\ttrue} \label{eq:ADP-2-3}\\[-.08cm]
         \tf(\tc(x,y)) &\to \{1:\tc(\tF(x),\tF(y))\}^{\ttrue} \nonumber 
     \end{align}}}
\end{example}

Since the original rule and all corresponding dependency pairs are encoded in a single
ADP,
when rewriting with ADPs we have to distinguish 
whether we intend
to rewrite with the original rule or with a dependency pair.
This is important as our analysis should only focus on the complexity of rewriting at annotated positions,
i.e., of those function calls that we still need to analyze.

\begin{definition}[$\tored{}{}{\P}$]\label{def:annotated-dps}
    Let $\P$ be a finite set of ADPs.
    A term $s \in \TT^\sharp$ rewrites with $\P$ to $\mu = \{p_1 : t_1, \dots, p_k
    : t_k\}$ (denoted $s \tored{i}{}{\P}\mu$) if there is 
    a position $\pi \in \pos_{\SignatureD \cup \SignatureA}(s)$, a
    rule $\ell \to \{p_1 : r_1, \dots, p_k : r_k\}^m \in \P$, 
    and
    a substitution $\sigma$ such that $\flat(s|_{\pi}) = \ell\sigma \in \ANF_{\P}$ (i.e.,
    all proper
        subterms are in normal form
    w.r.t.\ $\tored{i}{}{\P}$), 
    and for all $1 \leq j \leq k$, $t_j$ is defined as follows,
depending on the flag $m$ and on whether $\pi \in \posT(s)$ holds:
    \begin{equation*}
        \begin{tabular}{c | l@{\,}l | l@{\,}l |}
            & $\pi \in \posT(s)$ & & $\pi \not\in \posT(s)$ &\\
            \hline
            $m = \ttrue$ & $t_j = \quad\; s[r_j\sigma]_{\pi}$ & $(\mathbf{at})$ & $t_j = \quad\; s[\flat(r_j)\sigma]_{\pi}$ & $(\mathbf{nt})$ \\
            \hline
            $m = \tfalse$ & $t_j = \disannoPos{\pi}( s[r_j\sigma]_{\pi})$ & $(\mathbf{af})$ & $t_j = \disannoPos{\pi}( s[\flat(r_j)\sigma]_{\pi})$ & $(\mathbf{nf})$ \\
            \hline 
        \end{tabular}
    \end{equation*}
\end{definition}

Rewriting with $\P$ is like ordinary probabilistic term rewriting while considering and
modifying annotations.
We distinguish between  \textbf{a}-steps (\underline{\textbf{a}}nnotation) and
\textbf{n}-steps (\underline{\textbf{n}}o annotation).
Similar to the  DP and the DT framework for non-probabilistic TRSs,
for complexity we only ``count'' 
\textbf{a}-steps (on positions with \textbf{a}nnotated symbols) that apply dependency pairs, 
and between two \textbf{a}-steps there can be several \textbf{n}-steps 
where rules are applied below the position of the next \textbf{a}-step 
(in order to evaluate the arguments of the function call to normal forms).
The flag $m \in \{\ttrue, \tfalse\}$ indicates whether
the ADP may be used for such \textbf{n}-steps on the arguments before an \textbf{a}-step
on an annotated symbol above.

During an $\bat$-step (for \underline{\textbf{a}}nnotation
and \textsf{\underline{t}rue}), all annotations are kept except 
those in subterms that correspond to variables in the applied rule.
Those subterms are normal forms as we consider innermost rewriting.
An $\bat$-step at a position $\pi$ represents an 
\textbf{a}-step as it rewrites at the position of an annotation, 
but in addition, it can also represent an 
\textbf{n}-step if 
an annotated symbol is later rewritten at a position above $\pi$.
To ease readability, we illustrate \Cref{def:annotated-dps}
using
a rule with non-probabilistic structure, i.e., a probabilistic rewrite rule of the form $\ell \to \{1:r\}$.
An example for an $\bat$-step  is:
\[
    \tF(\ts(\tF(\ts(\tz)))) \tored{i}{}{\ADPair{\R_{2}}} \{1 : \tF(\ts(\tF(\tc(\tz,\tz))))\}
\]
using ADP \eqref{eq:ADP-2-3}. 
Here, we have $\pi = 1.1$, $\flat(s|_{1.1}) = \tf(\ts(\tz)) = \ell \sigma$, 
where $\sigma$ instantiates $x$ with the normal form $\tz$, and $r_1 = \tF(\tc(\tz,\tz))$.

A step of the form $\bnt$ (for \underline{\textbf{n}}o annotation and 
\textsf{\underline{t}rue})
performs a rewrite step at the position of a non-annotated defined symbol.
This represents only an \textbf{n}-step, and thus
all annotations on the right-hand side $r_j$ are removed.
An example for such a step is:
\[
    \tF(\ts(\tf(\ts(\tz)))) \tored{i}{}{\ADPair{\R_{2}}} \{1 : \tF(\ts(\tf(\tc(\tz,\tz))))\}
\]
using ADP \eqref{eq:ADP-2-3}. 
Here, we have the same $\pi$, $\ell \sigma$, and $\sigma$ as above, but use the
right-hand side $\flat(r_1) = \tf(\tc(\tz,\tz))$ without annotations.

An
$\baf$-step (for \underline{\textbf{a}}nnotation
and 
\textsf{\underline{f}alse}) at a position $\pi$
only represents an \textbf{a}-step, but not 
an \textbf{n}-step, i.e., it does not
rewrite the arguments of a function
that is evaluated later on an annotated posi\-tion above.
Therefore, we remove
all annotations above  $\pi$, as no \textbf{a}-step is allowed to occur above $\pi$ afterwards.
If $\ADPair{\R_2}'$  contains $\tf(\ts(x)) \to\linebreak \{1:\tF(\tc(x,x))\}^{\tfalse}$,
then
a step of the form $\baf$ would be:
\[
    \tF(\ts(\tF(\ts(\tz)))) \tored{i}{}{\ADPair{\R_{2}}'} \{1 : \tf(\ts(\tF(\tc(\tz,\tz))))\}
    \]
    
Finally, a step of the form $\bnf$ (for \underline{\textbf{n}}o annotation and 
\textsf{\underline{f}alse}) is irrelevant for proving an upper bound on the expected runtime complexity 
since there can never be another \textbf{a}-step at a position above.
These steps are only included to ensure that the innermost evaluation strategy is not
affected if one modifies the annotations or the flag of ADPs (such modifications will be
done by our ADP processors in \Cref{sec:Framework}).
An example would be 
\[
    \tF(\ts(\tf(\ts(\tz)))) \tored{i}{}{\ADPair{\R_{2}}'} \{1 : \tf(\ts(\tf(\tc(\tz,\tz))))\}
\]
 with $\tf(\ts(x)) \to \{1:\tF(\tc(x,x))\}^{\tfalse}$ at Position 1.1
again.

Next, we lift RSTs to \emph{chain trees} that consider rewriting 
with $\tored{i}{}{\P}$ instead of $\ito_{\R}$.

\begin{definition}[Chain Tree]\label{def:chain-tree}
    Let $\F{T} = (V,E,L)$ be a (possibly infinite) labeled and directed tree with nodes $V
    \neq \emptyset$ and edges $E \subseteq V \times V$,
    where $vE = \{ w \mid (v,w) \in E \}$ is finite for every $v \in V$.
    We say that $\F{T}$ is a $\P$\emph{-chain tree} ($\P$-CT) if
    \begin{itemize}
        \item $L:V\rightarrow(0,1]\times\TT^{\sharp}$ labels every node $v$ 
        by a probability $p_v$ and a term $t_v$.
        For the root $v \in V$ of the tree, we have $p_v = 1$.
      \item If $vE\!=\!\{w_1, \ldots, w_k\}$, then
        $t_v \tored{}{}{\PP}
        \{\tfrac{p_{w_1}}{p_v}\!:t_{w_1}, \ldots, \tfrac{p_{w_k}}{p_v}\!:t_{w_k}\}$.
    \end{itemize}
    For every inner node $v$, let $\P(v) \in \P \times \{\bat, \baf, \bnt, \bnf \}$ 
    be the ADP and the kind of step used for rewriting $t_v$.  
\end{definition}
\noindent
The ${\ADPair{\R_{\tgeo}}}$-CT in
\Cref{fig:CT_example}  corresponds to the 
${\R_{\tgeo}}$-RST in \Cref{fig:RST_example}.

\begin{figure}
	\centering
    \begin{tikzpicture}
        \tikzstyle{adam}=[thick,draw=black!100,fill=white!100,minimum size=4mm, shape=rectangle split, rectangle split parts=2,rectangle split horizontal,font={\footnotesize}]
        \tikzstyle{empty}=[rectangle,thick,minimum size=4mm]
        
        \node[adam] at (0, 0)  (a) {$1$\nodepart{two}$\tGeo(\tz)$};

        \node[adam] at (-1.8, -0.6)  (b) {$\nicefrac{1}{2}$\nodepart{two}$\tGeo(\ts(\tz))$};
        \node[adam,label=below:{\footnotesize $\quad \mathtt{NF}_{\ADPair{\R_{\tgeo}}}$}] at (1.6, -.6)  (c) {$\nicefrac{1}{2}$\nodepart{two}$\tz$};

        \node[adam] at (-4, -1.2)  (d) {$\nicefrac{1}{4}$\nodepart{two}$\tGeo(\ts(\ts(\tz)))$};
        \node[adam,label=below:{\footnotesize $\quad \NF_{\ADPair{\R_{\tgeo}}}$}] at (0, -1.2)  (e) {$\nicefrac{1}{4}$\nodepart{two}$\ts(\tz)$};
        
        \draw (a) edge[->] (b);
        \draw (a) edge[->] (c);
        \draw (b) edge[->] (d);
        \draw (b) edge[->] (e);
    \end{tikzpicture}\vspace*{-.3cm}
	  \caption{\vspace*{-.2cm}${\ADPair{\R_{\tgeo}}}$-Chain tree starting with $\tGeo(\tz)$}\label{fig:CT_example}
\end{figure}

In contrast to the \emph{expected derivation length} of RSTs, for the expected derivation
length of CTs, we only consider  $\bat$- or  $\baf$-steps, i.e., steps at the position of an
annotated symbol.\footnote{Since $\bnt$- and $\bnf$-steps
are disregarded for the expected derivation length of CTs,
in contrast to the chain
trees used for proving $\AST$ in \cite{FLOPS2024}, we do not have to
require that every infinite path
    contains infinitely many  $\bat$- or  $\baf$-steps.}
    Moreover, sometimes we do not want to count the application of
\emph{all} ADPs, but 
only the ADPs from some subset. Thus, similar to the adaption of
DPs for\linebreak
complexity analysis in \cite{noschinski2013analyzing}, 
in our ADP framework we do not consider a single set of ADPs $\P$, but we use a
second set $\SSS \subseteq \P$
of those\linebreak ADPs that (still) have to be taken into account.
So
one only has to add the probabilities of \textbf{a}-steps with ADPs from
$\SSS$ in the chain tree to determine its \emph{expected derivation length}.
Thus, our ADP frame\-work uses \emph{ADP problems} $\langle \P, \SSS \rangle$ where $\SSS \subseteq \P$,
and our analysis\linebreak ends once we have $\SSS = \emptyset$.
  Such
ADP problems are called \emph{solved}.
In the remainder, we fix an arbitrary ADP problem $\langle \P, \SSS \rangle$.

\begin{definition}[Exp.\ Derivation Length for Chain Trees, $\edl_{\langle \P, \SSS \rangle}$]\label{def:expected-derivation-length-extension}
    Let $\F{T} = (V,E,L)$ be a $\P$-chain tree. 
    The \emph{expected derivation length} of $\F{T}$, where we only count
    steps with $\SSS$ at annotated symbols, is
    \[ \textstyle
     \edl_{\langle \P, \SSS \rangle}(\F{T}) = \sum_{v \in V \setminus \ctleaf^{\F{T}}, \;
       \P(v) \in \SSS \times  \{\bat, \baf\}} \; p_v
    \]
\end{definition}

\begin{example}[Expected Derivation Length for Chain Trees]\label{ex:expected-derivation-length-extension}
  Reconsider the PTRS $\R_1$ and the following $\ADPair{\R_1}$-chain tree $\F{T}$.

  
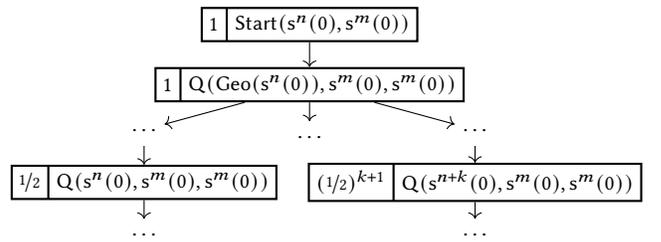
\vspace*{.2cm}

  \begin{center}
        \begin{small}
             \begin{tikzpicture}
                \tikzstyle{adam}=[thick,draw=black!100,fill=white!100,minimum size=4mm, shape=rectangle split, rectangle split parts=2,rectangle split horizontal]
                \tikzstyle{empty}=[rectangle,thick,minimum size=4mm]

                \node[adam] at (0, 0.8)  (z) {$1$\nodepart{two}$\tStart(\ts^n(\tz),\ts^m(\tz))$};

                \node[adam] at (0, 0)  (a) {$1$\nodepart{two}$\tQ(\tGeo(\ts^n(\tz)),\ts^m(\tz),\ts^m(\tz))$};
        
                \node[empty] at (-1.5, -0.8)  (b1) {$\ldots$};
                \node[adam] at (1.5, -0.8)  (b2) {$\nicefrac{1}{2}$\nodepart{two}$\tQ(\ts^n(\tz),\ts^m(\tz),\ts^m(\tz))$};
        
                \node[empty] at (-3, -1.6)  (c1) {$\ldots$};
                \node[adam] at (0.5, -1.6)  (c2) {$(\nicefrac{1}{2})^{k+1}$\nodepart{two}$\tQ(\ts^{n+k}(\tz),\ts^m(\tz),\ts^m(\tz))$};

                \node[empty] at (4, -1.4)  (d1) {$\ldots$};
                \node[empty] at (3.5, -2.2)  (d2) {$\ldots$};
                \node[empty] at (-3, -2.2)  (space) {\phantom{0}};
                
                \draw (z) edge[->] (a);
                \draw (a) edge[->] (b1);
                \draw (a) edge[->] (b2);
                \draw (b1) edge[->] (c1);
                \draw[->] ([rotate=10]b1) -- (c2);
                \draw (b2) edge[->] (d1);
                \draw (c2) edge[->] (d2);
            \end{tikzpicture}
        \end{small}
    \end{center}
    Let $\SSS_{\tgeo} = \{\eqref{eq:ADP-1-2}\}$ contain only the $\tgeo$-ADP 
and let $\SSS_{\tq} = \{  \eqref{eq:ADP-1-4},\eqref{eq:ADP-1-5}, \eqref{eq:ADP-1-3}\}$
contain the $\tq$-ADPs.
    Computing the expected derivation length of $\F{T}$ w.r.t.\ $\SSS_{\tgeo}$ or
    w.r.t.\ all ADPs
    results in

    \[
    \begin{tabular}{r@{\;}c@{\;}l}
        $\edl_{\langle \ADPair{\R_1}, \SSS_{\tgeo} \rangle}(\F{T})$&$=$&$1 + \tfrac{1}{2} +
        \tfrac{1}{4} + \tfrac{1}{8} + \ldots =
          1 + \sum_{k=1}^{\infty} \tfrac{1}{2^k} =
          2$\\[.2cm]
            $\edl_{\langle \ADPair{\R_1}, \ADPair{\R_1} \rangle}(\F{T})$ &$=$&  $\mbox{\small
            $1 + \edl_{\langle
            \ADPair{\R_1}, \SSS_{\tgeo} \rangle}(\F{T}) + \edl_{\langle \ADPair{\R_1},
            \SSS_{\tq} \rangle}(\F{T})$}$\\ 
            &$\leq$& $1 + 2 + \sum_{k=1}^{\infty} (\nicefrac{1}{2})^{k+1} \cdot (2 (n + k) + 2)$\\
            &$=$& $3 +  \sum_{k=1}^{\infty} \nicefrac{k}{2^{k}} +  \sum_{k=1}^{\infty}
            (\nicefrac{1}{2})^{k}  \cdot (n + 1)$\\
            &$=$& $3 + 2 + 1 \cdot (n + 1) = n + 6$
    \end{tabular}\]
\end{example}

Next, we define \emph{expected derivation height} via chain trees
by considering all possible CTs with the root $t^\sharp$ for a basic term $t$, and taking the
supremum of their expected derivation lengths. 

\begin{definition}[Exp.\ Derivation Height via Chain Trees, $\edh_{\langle \P, \SSS \rangle}$]\label{def:expected-complexity-of-terms}
    For $t \in \TB_\P$,
    the \textit{expected derivation height} 
    $\edh_{\langle \P, \SSS \rangle}(t) \in \IN \cup \{\omega\}$
    is the supremum obtained when adding all
    probabilities for \textbf{a}-steps with $\SSS$ 
    in any chain tree $\F{T}$  with root $t^{\sharp}$:
    {\small\[
        \edh_{\langle \P, \SSS \rangle}(t) = \sup \{ \edl_{\langle \P, \SSS \rangle}(\F{T})
        \mid \F{T} \text{ is a } \P \text{-chain tree with } \rootterm(\F{T}) = t^{\sharp} \}
    \]}
\end{definition}

\begin{example}[Expected Derivation Height w.r.t.\ Chain Trees]\label{ex:expected-complexity-of-terms}
    Consider the term $t = \tstart(\ts^n(\tz),\ts^m(\tz))$ 
    and its corresponding $\ADPair{\R_1}$-CT from \cref{ex:expected-derivation-length-extension}.
    As this is the only
    $\ADPair{\R_1}$-CT with root $t^\sharp$, we obtain
    $\edh_{\langle \ADPair{\R_1}, \ADPair{\R_1} \rangle}(t) \leq n + 6$ and $\edh_{\langle \ADPair{\R_1},\SSS_{\tgeo} \rangle}(t) = 2$.
\end{example}

\noindent
Now we can define expected runtime complexity for ADP problems.

\begin{definition}[Expected Runtime Complexity for ADP Problems, $\eirc_{\langle \P,\SSS \rangle}$]\label{def-cplx-ADP-problems}
    The \textit{expected runtime complexity function} of an ADP problem $\langle \P,\SSS \rangle$ is defined as 
    \[
        \eirc_{\langle \P,\SSS \rangle}(n) = \sup \{\edh_{\langle \P, \SSS \rangle}(t) \mid t \in \TB_\P, |t| \leq n\}
    \]
    and we define the \emph{runtime
    complexity} $\iota_{\langle \P,\SSS \rangle}$ of $\langle \P,\SSS \rangle$ as $\iota(\eirc_{\langle \P,\SSS \rangle})$.
\end{definition}

\begin{example}[Expected Runtime Complexity for ADP Problems]\label{ex:expected-innermost-runtime}
  For a basic term  
$\tstart(t_1, t_2)$, $\ADPair{\R_1}$ first computes a geometric distribution starting in $t_1$.
This needs $2$ steps in expectation, and  increases  $t_1$ by only $2$ in expectation.
    In the resulting term $\tq(t_{\tgeo}, t_2, t_2)$, 
    where $t_{\tgeo}$ is the normal form resulting from  $\tgeo(t_1)$,
    we decrease $t_{\tgeo}$ until we reach $\tz$.
    Therefore, the expected derivation height is linear in the size of the start term, 
    i.e., $\iota_{\langle \ADPair{\R_1},\ADPair{\R_1} \rangle} = \Pol_1$.
    If we only consider $\SSS_{\tgeo}$ for the complexity, 
    then $\iota_{\langle \ADPair{\R_1},\SSS_{\tgeo} \rangle} = \Pol_0$.
\end{example}

With our new concepts, we obtain the following novel \emph{chain criterion} for complexity
analysis of PTRSs.
It shows that to analyze the expected runtime complexity of a PTRS $\R$, 
it suffices to analyze the expected runtime complexity of its \emph{canonical ADP problem}
$\langle \ADPair{\R}, \ADPair{\R} \rangle$,
i.e., in the beginning  all ADPs are considered for complexity.
In the canonical ADP problem, all defined symbols in right-hand sides are annotated. Thus, 
one can only perform $\bat$-steps, because due to the innermost 
strategy, annotations are only removed from subterms in normal form. 
Hence, the rewrite steps with $\R$ and the ones with
$\ADPair{\R}$ directly correspond to each other. 

\begin{restatable}[Chain Criterion]{theorem}{ChainCrit}\label{theo:chain-criterion-extended-adps}
    Let $\R$ be a PTRS. Then for all basic terms $t \in \TB_\R$ we have
    \[
        \edh_\R(t) = \edh_{\langle \ADPair{\R}, \ADPair{\R} \rangle}(t)
    \]
    and therefore $\iota_\R = \iota_{\langle \ADPair{\R}, \ADPair{\R} \rangle}$.
\end{restatable}

So in contrast to the chain criterion of \cite{noschinski2013analyzing} 
for complexity analysis in the non-probabilistic setting, \Cref{theo:chain-criterion-extended-adps} 
yields
a \emph{tight} bound ($\iota_\R = \iota_{\langle \ADPair{\R}, \ADPair{\R} \rangle}$)
for arbitrary PTRSs 
due to the usage of ADPs instead of dependency tuples (with 
dependency tuples one would only obtain 
an upper bound, i.e.,
$\iota_\R \sqsubseteq \iota_{\langle \ADPair{\R}, \ADPair{\R} \rangle}$).

\subsection{ADP Framework}\label{sec:ADP-Framework}

Like
the original DP framework of \cite{giesl2006mechanizing},
our ADP framework is a \emph{divide-and-conquer} approach 
which applies \emph{processors} to simplify ADP problems until all
subproblems are solved. As in \cite{noschinski2013analyzing}, a processor also returns a complexity $c \in
\F{C}$. 

\begin{definition}[Processor]\label{def:ADP-processors}
    An \textit{(ADP) processor} $\Proc$ is a function
    $\Proc(\langle \P, \SSS \rangle) 
    = (c, \{\langle \P_1, \SSS_1 \rangle, \ldots, \langle \P_n, \SSS_n \rangle\})$ 
    mapping an ADP problem $\langle \P, \SSS \rangle$ 
    to a complexity $c \in \F{C}$ and a set of ADP problems. 
\end{definition}
The intuition for $\Proc$ is that 
    in addition to the
complexities of $\langle \P_i, \SSS_i \rangle$, the complexity $c$ is also used to obtain
an upper bound on the complexity of $\langle \P, \SSS \rangle$.
During the analysis with our ADP framework, we construct a \emph{proof tree} 
that contains all subproblems and complexities resulting from the application of
processors.

\begin{figure}
	\centering
    \begin{tikzpicture}
        \tikzstyle{adam}=[thick,draw=black!100,fill=white!100,minimum size=4mm, rectangle split parts=2,rectangle split horizontal]
        \tikzstyle{empty}=[rectangle,thick,minimum size=4mm]
        
        \node[empty,label=left:{\footnotesize \textcolor{red}{$\Pol_1$}}] at (0, 0)  (a) {$\langle \P, \SSS \rangle$};

        \node[empty,label=left:{\footnotesize \textcolor{red}{$\Pol_2$}}] at (-2, -0.4)  (b) {$\langle \P_1, \SSS_1 \rangle$};
        \node[empty,label=right:{\footnotesize \textcolor{red}{$\Pol_0$}}] at (2, -0.4)  (c) {$\langle \P_2, \emptyset \rangle$};

        \node[empty,label=left:{\footnotesize \textcolor{red}{$\Pol_0$}}] at (-2, -1.2)  (d) {$\langle \P_3, \emptyset \rangle$};
        
        \draw (a) edge[->] (b);
        \draw (a) edge[->] (c);
        \draw (b) edge[->] (d);
    \end{tikzpicture}\vspace*{-.3cm}
    \caption{\vspace*{-.1cm}Example of a solved proof tree}\label{fig:proof_tree}
\end{figure}

\begin{definition}[Proof Tree]\label{def:proof_tree}
    A \emph{proof tree} is a labeled, finite tree $(V,E,L_{\C{A}},L_{\C{C}})$ 
    with a labeling $L_{\C{A}}$ that maps each node to an ADP problem
    and a second labeling $L_{\C{C}}$ that maps each node to a complexity from $\F{C}$.
    Each edge represents an application of a processor, i.e., 
    if $vE = \{w_1, \ldots, w_n\}$, 
    then $\Proc(L_{\C{A}}(v)) = (L_{\C{C}}(v), \{ L_{\C{A}}(w_1),\linebreak  \ldots, L_{\C{A}}(w_n)\})$
    for some processor $\Proc$,
    and we require that the complexity of all leaves $v$ is $L_{\C{C}}(v) = \Pol_0$ if the corresponding 
    ADP problem $L_{\C{A}}(v)$ is solved, and $L_{\C{C}}(v) = \omega$ otherwise.
    We call a proof tree \textit{solved} if all ADP problems in its leaves are solved.
\end{definition}

\Cref{fig:proof_tree} shows an example of a solved proof tree where
the complexities given by the labeling $L_{\C{C}}$ are depicted
next to the nodes. So here, a processor with $\Proc(\langle \P, \SSS \rangle) = (\Pol_1, \{
\langle \P_1, \SSS_1 \rangle, \langle \P_2, \emptyset \rangle \})$
was used for the step from the root to its two children.

To ensure that  the maximum of all complexities in
a proof tree is an 
upper bound on the complexity of the ADP problem at the root, we require that proof trees are \emph{well
formed}, i.e., that for every node $v$
with $L_{\C{A}}(v) = \langle \P, \SSS \rangle$,
the expected complexity of 
$\langle \P, \SSS \rangle$ is bounded by the $L_{\C{C}}$-labels
of the subtree starting
at $v$ and the $L_{\C{C}}$-labels on the path from the root to $v$.
However, since no processor has been applied on the
leaves of the proof tree
(yet), for leaves
we use the actual expected complexity instead of the $L_{\C{C}}$-label.
Moreover, for well-formed proof trees
we require that 
the
ADPs from $\PP \setminus \SSS$ have already been taken into account in the path from the
root to $v$. This will\linebreak be exploited, e.g.,  in the
\emph{knowledge propagation processor} of
\Cref{sec:knowledge-propagation}.

\begin{definition}[$\oplus$, Well-Formed Proof Tree]\label{def:well-formed_proof_tree}
    Let $\oplus$ be the maximum operator on complexities, i.e.,
    for $c,d \in \F{C}$, let
   $c \oplus d = d$ if $c \sqsubseteq d$ and $c \oplus
    d = c$ otherwise (so, e.g.,  $\Pol_2 \oplus \Pol_1 =
    \Pol_2$).
  
   A proof tree $(V,E,L_{\C{A}},L_{\C{C}})$ is \emph{well formed}
    if for every node $v$ with $L_{\C{A}}(v)\!=\!\langle \P\!, \SSS \rangle$
    and path $v_1,\ldots,v_k\!=\!v$ from the root $v_1$ to $v$,  we have
    \[\mbox{\small \begin{tabular}{r@{\;}c@{\;}l}
        $\iota_{\langle \P,  \SSS \rangle}$ &$\sqsubseteq$& $L_{\C{C}}(v_1) \oplus
        \dots \oplus L_{\C{C}}(v_{k-1})
        \oplus \max \{ L_{\C{C}}'(w) \mid (v,w) \in E^* \}$\\
        $\iota_{\langle \P, \P \setminus \SSS \rangle}$ &$\sqsubseteq$& $L_{\C{C}}(v_1) \oplus
        \dots \oplus L_{\C{C}}(v_{k-1})$
        \end{tabular}}\]
        Here, $E^*$ is the reflexive-transitive closure of the edge relation,
        i.e., $(v,w) \in E^*$ if $v$ reaches  $w$ in the proof tree.
        Moreover, let $L_{\C{C}}'(v) =  L_{\C{C}}(v)$ for inner nodes $v$
        and $L_{\C{C}}'(v) =  \iota_{L_{\C{A}}(v)}$  for leaves $v$.
\end{definition}

The following corollary shows that for well-formed proof trees,
the complexity of the ADP problem at the root  is indeed
bounded by the maximum of all complexities at  the nodes.

\begin{restatable}[Complexity Bound from Well-Formed Proof Tree]{corollary}{ComplexityBoundByTree}\label{theo:approximating_complexities}
    Let $\F{P} = (V,E,L_{\C{A}},L_{\C{C}})$ be a 
    well-formed proof tree
    with $L_{\C{A}}(v_1) = \langle \P, \SSS \rangle$ for the root $v_1$ of $\F{P}$.
    Then, 
    \[\iota_{\langle \P, \SSS \rangle} \sqsubseteq \max\{L_{\C{C}}(v) \mid v \in V\}.\]
\end{restatable}

Now we can define when a processor is \emph{sound}.

\begin{definition}[Soundness of Proc.\ $\!$]\label{def:sound_processors}
     A processor $\Proc(\langle \P, \SSS \rangle) = (c, \{\langle \P_1, \SSS_1 \rangle,
     \ldots, \langle \P_n, \SSS_n \rangle\})$  
    is \textit{sound} 
    if for all well-formed proof trees $(V,E,L_{\C{A}},L_{\C{C}})$ and
    all nodes $v \in V$, we have: if
    $L_{\C{A}}(v) = \langle \P,\linebreak \SSS \rangle$ and $v_1,\ldots,v_k = v$ is the path
    from the root node $v_1$ to $v$, then
   {\small \begin{eqnarray}
      \label{eq:soundness-1}
        \iota_{\langle \P, \SSS \rangle} 
        &\hspace*{-.3cm}\sqsubseteq\hspace*{-.3cm}& L_{\C{C}}(v_1) \oplus ... \oplus L_{\C{C}}(v_{k-1}) \oplus c \oplus
        \iota_{\langle \P_1, \SSS_1 \rangle} \oplus ... \oplus \iota_{\langle \P_n,
          \SSS_n \rangle}\\
    \label{eq:soundness-2} \iota_{\langle \P_i, \P_i \setminus \SSS_i \rangle} 
        &\hspace*{-.3cm}\sqsubseteq\hspace*{-.3cm}& L_{\C{C}}(v_1) \oplus ... \oplus
    L_{\C{C}}(v_{k-1}) \oplus c\qquad \quad
        \text{for all $1 \leq i \leq n$}
        \end{eqnarray}}
\end{definition}
So \eqref{eq:soundness-1} requires that the complexity of
the considered ADPs $\SSS$ must be bounded 
by the maximum of all ``previous'' complexities
$L_{\C{C}}(v_1), \ldots, L_{\C{C}}(v_{k-1})$, the newly derived complexity $c$, 
    and the complexity of the  remaining ADP problems $\langle \P_1, \S_1 \rangle, \ldots,
    \langle \P_n, \S_n \rangle$.
Moreover,  \eqref{eq:soundness-2} ensures that in the remaining ADP problems
$\langle \P_i, \SSS_i \rangle$, the complexity of the ``non-considered'' ADPs $\P_i \setminus \SSS_i$ 
is bounded 
by the maximum of all previous complexities and the newly derived complexity $c$.
This ensures that well-formedness of proof trees is preserved when extending them by
applying sound
processors.

\begin{restatable}[Sound Processors Preserve Well-Formedness]{lemma}{WellFormednessPreservation}\label{lem:continuing-the-proof-chain}
    Let $\F{P} = (V,E,L_{\C{A}},L_{\C{C}})$ be a proof tree with a leaf $v$ where
    $L_{\C{A}}(v)$ is not solved,  
    and let $\Proc$ be a sound processor such that
  $\Proc(L_{\C{A}}(v)) = (c, \{\langle \P_1, \SSS_1 \rangle, \ldots, \langle \P_n, \SSS_n \rangle\})$.
Let $\F{P}'$  result from $\F{P}$ by adding fresh nodes $w_1, \ldots, w_n$ and edges
$(v,w_1), \ldots, (v,w_n)$, where the labeling is extended such that
$L_{\C{A}}(w_i) = \langle \P_i, \SSS_i \rangle$ for all $1 \leq i \leq n$
and $L_{\C{C}}(v) = c$. Then $\F{P}'$ is also well formed.
\end{restatable}

To determine an upper bound on the expected runtime complexity $\iota_\R$ of a PTRS $\R$,
our ADP framework
starts with the canonical ADP problem $\langle \ADPair{\R}, \ADPair{\R} \rangle$
and applies sound processors repeatedly until all problems are solved.
Then by
\Cref{theo:chain-criterion-extended-adps,theo:approximating_complexities,lem:continuing-the-proof-chain}, 
the expected runtime complexity $\iota_\R$ is bounded by the maximum of all complexities
occurring in the corresponding proof tree.

\begin{corollary}[Soundness of the ADP Framework for Expected Runtime Complexity]\label{cor:correctness}
    Let $\R$ be a PTRS and $\F{P} = (V,E,\linebreak L_{\C{A}}, L_{\C{C}})$ be a 
    well-formed solved proof tree
    where $L_{\C{A}}(v_1) = \langle \ADPair{\R},\linebreak \ADPair{\R} \rangle$ for the root $v_1$ of
    $\F{P}$. Then we have
    \[
        \iota_{\R} \sqsubseteq \max\{L_{\C{C}}(v) \mid v \in V\}.
    \]
\end{corollary}

\begin{remark}\label{K-Component}
    While our framework is inspired by the DT
    framework of \cite{noschinski2013analyzing} for
    complexity analysis of non-probabilistic TRSs,
    our adaption to PTRSs differs from \cite{noschinski2013analyzing} in several
    aspects. Apart from using ADPs instead of dependency tuples (which results in a ``tight''
    chain criterion instead of an over-approximation), we use proof trees (instead of
    just proof chains, which allows us to use processors that return several subproblems),
     and we
    introduced the novel concept of \emph{well-formed} proof trees and 
    require that sound processors preserve
    well-formedness. This will allow us to define a 
    \emph{knowledge propagation processor} in \Cref{sec:knowledge-propagation} which takes the
    knowledge provided by well-formed proof trees into account.
    In contrast to \cite{noschinski2013analyzing}, we obtain  such a processor without
   extending our ADP problems by an additional component 
    $\K$ that contains those dependency
    tuples which were already taken into account in the proof up to now, since
    in our setting we would always have $\K = \P \setminus \SSS$.
\end{remark}

\section{ADP Processors}\label{sec:Framework}

In this section, we
adapt the main processors of the DP and the DT framework \cite{giesl2006mechanizing,
  noschinski2013analyzing} in order to analyze expected
runtime complexity of probabilistic TRSs. 
Throughout the section, we illustrate our processors with the PTRS
$\R_1$ from \Cref{ex:LeadingExamples}.
The resulting solved proof tree for the initial
ADP problem $\langle \ADPair{\R_1}, \ADPair{\R_1} \rangle$ 
that we construct during this section is depicted
in \Cref{fig:proof_tree_R1}.

\subsection{Usable Rules Processor}\label{sec:usable-rules}

We start with a processor that considers the \emph{usable rules}.
Usable rules over-approximate
the set of those rules that can be used 
to evaluate the arguments of annotated function symbols if their
variables are instantiated by normal forms (as required in innermost
evaluations). Essentially, the usable rules of a term $t$ consist of all rules for the
defined symbols $f$ occurring in $t$ and all rules that are usable for the
terms in the right-hand sides of $f$-rules. 

\begin{definition}[Usable Rules]\label{def:usable-rules}
  For every $f\!\in\!\SignatureADC$ and set of ADPs $\PP$, 
  let $\rules_{\PP}(f) = \{\ell\!\to\!\mu^{\ttrue}\!\in\!\PP \mid \rootsym(\ell)\!=\!f\}$.
  Moreover, for every  $t \in \TT^{\sharp}$, the \emph{usable rules} 
  $\urules_{\PP}(t)$ of $t$ w.r.t.\ $\PP$ are defined as:
  {\small
   \[\begin{tabular}{r@{\;}c@{\,}l}
    $\urules_\PP(t)$ &$=$& $\emptyset, \hspace*{3.27cm} \text{if $t \in \VSet$ or $\PP = \emptyset$}$\\
    $\urules_\PP(f(t_1, \ldots, t_k))$&$=$&$\rules_{\PP}(f) 
    \cup \bigcup_{1 \leq j \leq k}\!\urules_{\PP'}(t_j)$\\
    &&$\cup \bigcup_{\ell \to \mu^{\ttrue} \in \rules_{\PP}(f),
      \; r \in \Supp(\mu)}  \urules_{\PP'}(\flat(r))$
  \end{tabular}\]
  }
   where $\PP' = \PP \setminus \rules_\PP(f)$. The usable rules of $\P$ are \[\textstyle \mathcal{U}(\P) = \bigcup_{\ell \to \mu^{m} \in \P, \;
  r \in \Supp(\mu), \; t \trianglelefteq_{\sharp} r} \; \mathcal{U}_\P(t^{\sharp}).\]
\end{definition}

Similar to the usable rules processor for $\AST$ in \cite{FLOPS2024}, our
usable rules processor sets the flag of all non-usable rules in $\P$ to $\tfalse$ to indicate
that they cannot be used to evaluate arguments of annotated functions that are
rewritten afterwards. 
The rules in $\P$'s subset $\SSS$ are changed analogously (since the purpose of 
$\SSS$ is only to indicate which ADPs must still be counted for complexity).

\begin{restatable}[Usable Rules Pr.]{theorem}{UsableRules}\label{thm:exUR-proc}
 \hspace*{-.2cm}  For an ADP problem $\langle \P, \SSS \rangle$, let
\[
    \begin{tabular}{r@{\;}c@{\;}c@{\;}l@{\;}l}
        $\P'$ &$=$& $\mathcal{U}(\P)$ &$\cup$& $\{\ell \to \mu^{\tfalse} \, \mid \, \ell \to \mu^m \in \P \setminus \mathcal{U}(\P)\},$ \\
        $\SSS'$ &$=$& $(\SSS \, \cap \, \mathcal{U}(\P))$ &$\cup$& $\{\ell \to \mu^{\tfalse} \, \mid \, \ell \to \mu^m \in \SSS \setminus \mathcal{U}(\P)\}$.
    \end{tabular}
\]
    Then, $\Proc_{\mathtt{UR}}(\langle \P, \SSS \rangle) = (\Pol_0, \{\langle \P', \SSS' \rangle\})$ is sound.
\end{restatable}

\begin{example}[Usable Rules Processor]\label{ex:usable-rules-proc}
  Consider the ADP problem $\langle \ADPair{\R_1}, \ADPair{\R_1}
  \rangle$ from \Cref{ex:canonical-adps}.
  There is only one term $r$ in the right-hand sides with an annotated subterm $t \trianglelefteq_{\sharp} r$ where $t^\sharp$ has
  a defined symbol below the
  annotated root, viz.\ $t^\sharp = \tQ(\tgeo(x),y,y)$.
  Thus,  only the $\tgeo$-ADP
  $\eqref{eq:ADP-1-2}$ is usable and we can set the flag of all other ADPs
  to $\tfalse$.
      Hence, we get $\Proc_{\mathtt{UR}}(\langle \ADPair{\R_1}, \ADPair{\R_1} \rangle) 
      = (\Pol_0, \{ \langle \P_1, \P_1 \rangle \})$
      where $\P_1 = 
      \{\eqref{eq:ADP-1-1-usablerule} -
    \eqref{eq:ADP-1-3-usablerule}\}$
    with

    \vspace*{-.2cm}

    {\small \begin{align}
        \tstart(x,y) &\to \{1:\tQ(\tGeo(x),y,y)\}^{\tfalse} \label{eq:ADP-1-1-usablerule}\\[-.08cm]
        \tgeo(x) &\to \{\nicefrac{1}{2}:\tGeo(\ts(x)), \, \nicefrac{1}{2}:x\}^{\ttrue} \label{eq:ADP-1-2-usablerule}\\[-.08cm]
            \tq(\ts(x),\ts(y),z) &\to \{1:\tQ(x,y,z)\}^{\tfalse} \label{eq:ADP-1-4-usablerule}\\[-.08cm]
        \tq(x,\tz,\ts(z)) &\to
        \{1:\ts(\tQ(x,\ts(z),\ts(z)))\}^{\tfalse} \label{eq:ADP-1-5-usablerule}\\[-.08cm]
            \tq(\tz,\ts(y),\ts(z)) &\to \{1:\tz\}^{\tfalse} \label{eq:ADP-1-3-usablerule}
    \end{align}	} 
\end{example}
\noindent
Having fewer rules with the flag $\ttrue$
has advantages, e.g.,
for the \emph{dependency graph} and the
\emph{reduction pair processor} in \Cref{sec:dependency-graph,sec:red-pair}.

\begin{example}[Basic Start Terms]\label{ex:why-basic}
The restriction to basic 
start terms is not only required to infer polynomial upper bounds
from CPIs (see \Cref{sec:TRS}),
but it is also essential for the soundness of the usable rules processor.
To see this, let $\R_3$ contain all rules of
  $\R_2$ from \Cref{ex:LeadingExamples} except
$\tstart \to \{1:\tf(\tgeo(\tz))\}$.
If we do not require basic start terms, then we can start an
 evaluation with
  the term $\tf(\tgeo(\tz))$,
  i.e., then $\R_3$ is not $\mathtt{SAST}$.
The canonical ADPs $\ADPair{\R_3}$
are the same as for $\R_2$ in \Cref{ex:canonical-adps}
just without the $\tstart$-ADP.
Thus, $\ADPair{\R_3}$ has no usable rules and
$\Proc_{\mathtt{UR}}$ sets the flag of
all ADPs to $\tfalse$.
When starting with a term like $\tF(\tGeo(\tz))$,
then one application of the $\tgeo$-ADP now
removes the annotations of the $\tF$-symbols above it,
as the  $\tgeo$-ADP now has the flag $\tfalse$. So for the resulting ADP problem,
only $\tGeo$ is annotated in chain trees and thus, they all have finite
expected derivation length. 
Hence, we would now falsely infer 
that $\R_3$ is
$\mathtt{SAST}$  w.r.t.\ arbitrary start terms.
\end{example}

\subsection{Dependency Graph Processor}\label{sec:dependency-graph}

The \emph{dependency graph} is a control flow graph that indicates 
which function calls can occur after each other.
This does not depend on the probabilities,
and we can consider each function call on its own.
Hence, we can use the ordinary dependency graph of the corresponding 
(non-probabilistic) dependency pairs. 
To also detect the predecessors of ADPs $\ell \to \mu$ without annotations, we add
a dependency pair $\ell^\sharp \to \bot$ for a fresh symbol $\bot$ in that case.

\begin{definition}[Non-Probabilistic Variant, Dependency
    Pairs]\label{def:non-prob-variants}
  For a set of ADPs $\P$,
let $\nonprob(\PP) = \{\ell \to \flat(r_j) \mid \ell \to \{p_1:r_1, \ldots, p_k:r_k\}^{\ttrue} \in \PP, 1
\leq j \leq k\}$ denote its \emph{non-probabilistic rule variant}. So
$\nonprob(\PP)$ is 
  an ordinary TRS 
    which only considers the ADPs with the flag $\ttrue$.
  For any ADP $\alpha =   \ell \to \{p_1:r_1, \ldots,
p_k:r_k\}^{m}$, let $\nonprobDP(\alpha) = \{ \ell^\sharp \to t^\sharp \mid
1 \leq j \leq k, t \trianglelefteq_{\sharp} r_j\}$.
Moreover, let $\nonprobDP^\bot(\alpha) = \{ \ell^\sharp \to \bot \}$ if $
\nonprobDP(\alpha) = \emptyset$, and
$\nonprobDP^\bot(\alpha) = \nonprobDP(\alpha)$, otherwise.
For $\PP$, let
$\nonprobDP(\PP) = \bigcup_{\alpha \in \PP} \nonprobDP^\bot(\alpha)$
 denote its \emph{non-probabilistic DP variant}, 
    which is a set containing dependency pairs as in
    \cite{arts2000termination,giesl2006mechanizing} and rules of the form $\ell^\sharp \to \bot$.
\end{definition}

So each dependency pair from $\nonprobDP(\alpha)$
corresponds to a single annotation on the right-hand side of the
ADP $\alpha$. In the dependency graph, the edges indicate whether one DP can follow another when
the instantiated arguments are evaluated with $\nonprob(\PP)$.

\begin{definition}[Dependency Graph]\label{def:dependency-graph}
    The $\P$-\emph{dependency graph} has the nodes $\nonprobDP(\PP)$ 
    and there is an edge from $\ell_1^\sharp \to t_1^\sharp$ to $\ell_2^\sharp  \to ...$  
    if there are substitutions $\sigma_1, \sigma_2$ such that 
    $t_1^\sharp  \sigma_1 \ito^*_{\nonprob(\P)} \ell_2^\sharp  \sigma_2$ 
    and both $\ell_1^\sharp  \sigma_1$ and $\ell_2^\sharp  \sigma_2$ are in argument normal form,  
    i.e., $\ell_1^\sharp  \sigma_1, \ell_2^\sharp  \sigma_2 \in \ANF_{\P}$.
\end{definition}
While the dependency graph is not computable in general, several techniques have been
developed to compute over-approxima\-tions of the graph automatically, e.g., \cite{arts2000termination,giesl2006mechanizing,hirokawa2005automating}.

\begin{example}[Dependency Graph]\label{ex:dependency-graph}
  We continue with $\langle \P_1,\P_1 \rangle$    from \cref{ex:usable-rules-proc}, where
  $\P_1 =    \{\eqref{eq:ADP-1-1-usablerule} -
    \eqref{eq:ADP-1-3-usablerule}\}$.
    We  have

\vspace*{-.3cm}
    
    {\small \begin{align}
      \nonprobDP(\PP_1) = \{ \;
\tStart(x,y) &\to \tQ(\tgeo(x),y,y), \label{dpR1-1}\\[-.08cm]
\tStart(x,y) &\to \tGeo(x), \label{dpR1-2}\\[-.08cm]
\tGeo(x) &\to \tGeo(\ts(x)), \label{dpR1-3}\\[-.08cm]
\tQ(\ts(x),\ts(y),z) &\to \tQ(x,y,z), \label{dpR1-4}\\[-.08cm]
\tQ(x,\tz,\ts(z)) &\to \tQ(x,\ts(z),\ts(z)), \label{dpR1-5}\\[-.08cm]
\tQ(\tz,\ts(y),\ts(z)) &\to \bot \label{dpR1-6}\; \}. 
\end{align}}
    The $\P_1$-dependency graph 
    is depicted in \Cref{fig:dependency-graph}.
\end{example}

\begin{figure}
	\centering
    \scriptsize
    \begin{tikzpicture}
        \node[shape=rectangle,draw=black!100] (A) at (1,2) {\eqref{dpR1-2}};
        \node[shape=rectangle,draw=black!100] (B) at (2,2) {\eqref{dpR1-3}};
        
        \node[shape=rectangle,draw=black!100] (C) at (-3,1.3) {\eqref{dpR1-1}};
        \node[shape=rectangle,draw=black!100] (D) at (-4,2) {\eqref{dpR1-4}};
        \node[shape=rectangle,draw=black!100] (E) at (-2,2) {\eqref{dpR1-5}};
        \node[shape=rectangle,draw=black!100] (F) at (-1,2) {\eqref{dpR1-6}};
        
        \path [->] (A) edge (B);
        \path [->,in=20, out=340, looseness=10] (B) edge (B);
        \path [->,in=270, out=165] (C) edge (D);
        \path [->,in=270, out=15] (C) edge (E);
        \path [->,in=190, out=350] (D) edge (E);
        \path [->,in=10, out=170] (E) edge (D);
\path [->, bend right=25] (C) edge (F);
\path [->, bend left=25] (D) edge (F);
\path [->] (E) edge (F);     
     \end{tikzpicture}
    \caption{$\P_1$-dependency graph}\label{fig:dependency-graph}
\end{figure}

The idea of the \emph{dependency graph processor} for termination analysis is to analyze
each strongly connected component (SCC)\footnote{A
set $\mathcal{G}$ of DPs is  an \emph{SCC} if it is a maximal cycle,
i.e., a maximal set where for any $\alpha, \alpha'$
in $\mathcal{G}$ there is
a non-empty path from $\alpha$ to $\alpha'$ only traversing nodes from
$\mathcal{G}$.} of the dependency graph separately. However (already in the
non-probabilistic setting,  e.g., \cite{noschinski2013analyzing}), this is not possible when
analyzing complexity. There are examples where all SCCs have linear
complexity but the full system has quadratic complexity, or where all individual SCCs 
are $\SAST$ but the full system is not (see \Cref{ex:two-sccs-together}).

The problem is that when considering an SCC individually, then
we lose the information how often and with which instantiations of the variables this SCC
is ``called''. For that reason, we now present a novel dependency graph
processor
  which regards each SCC
\emph{together with its ``prefix''}, i.e., together with all nodes of the dependency graph
that can reach the SCC. As usual, we say that a node \emph{reaches} an SCC if there is a
path from the node to the SCC in the dependency graph (where the path has length $\geq 0$,
i.e., each node also reaches itself). However, prefixes which are independent from each
other can be regarded separately, i.e., we only regard SCC-prefixes $\JJ$ where for all
nodes $\alpha, \beta \in \JJ$, $\alpha$ reaches $\beta$ or $\beta$ reaches $\alpha$.

\begin{definition}[SCC-Prefix]
  Let $\P$ be a set of ADPs.
  Then $\JJ$ is an \emph{SCC-prefix} of the
  $\PP$-dependency graph if there exists an SCC $\mathcal{G} \subseteq \JJ$ 
  where $\JJ \subseteq \nonprobDP(\PP)$
is a maximal set such that 
  all DPs of $\JJ$ reach $\mathcal{G}$
  and for all  $\alpha,\beta \in \JJ$,
    $\alpha$ reaches $\beta$ or
    $\beta$ reaches $\alpha$.
\end{definition}

For example, the $\P_1$-dependency graph of \Cref{fig:dependency-graph} has 
two SCCs
 $\{ \eqref{dpR1-3} \}$ and $\{ \eqref{dpR1-4}, \eqref{dpR1-5} \}$, and two SCC-prefixes
$\JJ_1 = \{ \eqref{dpR1-2}, \eqref{dpR1-3} \}$ and $\JJ_2 = \{ \eqref{dpR1-1}, \eqref{dpR1-4},
\eqref{dpR1-5} \}$. In this example,
$\JJ_1$ and $\JJ_2$
represent two completely independent parts
of the dependency graph.

The dependency graph processor now handles
each SCC-prefix $\JJ$ separately. To consider only the effects of the DPs $\JJ$ in the
ADPs of $\PP$, we replace every ADP $\alpha$ by the variant $\alpha|_\JJ$ where only those
symbols are annotated that correspond to the DPs from $\JJ$. For example, for $\JJ_1 = 
\{ \eqref{dpR1-2}, \eqref{dpR1-3} \}$ and the ADP
\begin{align}
  \alpha = \qquad
  \tstart(x,y) \to& \{1:\tQ(\tGeo(x),y,y)\}^{\tfalse}, \tag{\ref{eq:ADP-1-1-usablerule}}
  \end{align}
in $\alpha|_{\JJ_1}$ we only annotate $\tGeo$ but not $\tQ$ in the
right-hand side, due to the DP  $\eqref{dpR1-2}$ which results from its subterm $\tGeo(x)$,
i.e., $\alpha|_{\JJ_1} = \tstart(x,y) \to \{1:\tq(\tGeo(x),y,y)\}^{\tfalse}$.

\begin{restatable}[Dependency Graph Proc.]{theorem}{DependencyGraph}\label{thm:DPG} 
  Let $\langle \P, \SSS \rangle$ be an ADP problem
and let
  $\JJ$ be an SCC-prefix of the
$\PP$-dependency graph.
For any
ADP $\alpha = \ell \to \{p_1:r_1, \ldots, p_k:r_k\}^m \in \PP$
let
$\alpha|_\JJ = \ell \to \{p_1:\sharp_{\Phi_1}(r_1), \ldots, p_k:\sharp_{\Phi_k}(r_k)\}^m$ 
where for 
$1 \leq j \leq k$, we have  $\pi \in \Phi_j$ iff
there exists an $\ell^\sharp \to t^\sharp
\in \JJ$ such that
$t  \trianglelefteq_{\sharp}^\pi r_j$.
Similarly, let $\PP|_\JJ = \{ \alpha|_\JJ \mid \alpha \in
\PP\}$ and  $\SSS|_\JJ = \{ \alpha|_\JJ \mid \alpha \in
\SSS\}$.\myfootnote{Note that if a DP from $\JJ$ could
originate from several subterms $r_j|_\pi$ in several ADPs $\alpha$, then
we annotate all their positions in $\PP|_\JJ$.  
A slightly more powerful variant of the
processor could be obtained by storing for every DP of $\JJ$ from which subterm of
which ADP it originates, and only annotating this 
position. We did not do this in \Cref{thm:DPG} to ease the presentation.}

Then 
$\Proc_{\mathtt{DG}}(\langle \P, \SSS \rangle) =
(\Pol_0, \{ \langle \P|_\JJ, \SSS|_\JJ \rangle \mid \JJ$ is an  SCC-prefix of the
$\PP$-dependency graph $\})$ \pagebreak[3]
is sound.
\end{restatable}

\begin{example}[Dependency Graph Proc.]\label{ex:dependency-graph-processor}
Due to the two SCC-prefixes 
$\JJ_1 = \{ \eqref{dpR1-2}, \eqref{dpR1-3} \}$ and $\JJ_2 = \{ \eqref{dpR1-1},
\eqref{dpR1-4},
\eqref{dpR1-5} \}$, the dependency graph processor transforms the
ADP problem $\langle \P_1,\P_1 \rangle$ from \Cref{ex:usable-rules-proc} into
$\langle {\P_1}|_{\JJ_1},{\P_1}|_{\JJ_1} \rangle$ (corresponding to the $\tGeo$-SCC-Prefix) and
$\langle {\P_1}|_{\JJ_2},{\P_1}|_{\JJ_2} \rangle$ (corresponding to the $\tQ$-SCC-Prefix)
with

\vspace*{-.3cm}

{\small
 \begin{align}
    {\P_1}|_{\JJ_1}: \qquad\quad    \tstart(x,y) &\to \{1:\tq(\tGeo(x),y,y)\}^{\tfalse} \label{eq:ADP-1-1-scc}\\[-.08cm]
        \tgeo(x) &\to \{\nicefrac{1}{2}:\tGeo(\ts(x)), \; \nicefrac{1}{2}:x\}^{\ttrue} \label{eq:ADP-1-2-scc}\\[-.08cm]
            \tq(\ts(x),\ts(y),z) &\to \{1:\tq(x,y,z)\}^{\tfalse} \label{eq:ADP-1-4-scc}\\[-.08cm]
        \tq(x,\tz,\ts(z)) &\to
        \{1:\ts(\tq(x,\ts(z),\ts(z)))\}^{\tfalse} \label{eq:ADP-1-5-scc}\\[-.08cm]
        \tq(\tz,\ts(y),\ts(z)) &\to \{1:\tz\}^{\tfalse}  \label{eq:ADP-1-3-scc}\\
           {\P_1}|_{\JJ_2}: \qquad\quad  \;
   \tstart(x,y) &\to \{1:\tQ(\tgeo(x),y,y)\}^{\tfalse} \label{eq:ADP-1-1-2-scc}\\[-.08cm]
   \tgeo(x) &\to \{\nicefrac{1}{2}:\tgeo(\ts(x)), \;
   \nicefrac{1}{2}:x\}^{\ttrue} \label{eq:ADP-1-2-2-scc}\\[-.08cm]
            \tq(\ts(x),\ts(y),z) &\to \{1:\tQ(x,y,z)\}^{\tfalse} \label{eq:ADP-1-4-2-scc}\\[-.08cm]
        \tq(x,\tz,\ts(z)) &\to
        \{1:\ts(\tQ(x,\ts(z),\ts(z)))\}^{\tfalse} \label{eq:ADP-1-5-2-scc}\\[-.08cm]
            \tq(\tz,\ts(y),\ts(z)) &\to \{1:\tz\}^{\tfalse} \label{eq:ADP-1-3-2-scc}
 \end{align}}	
\end{example}

Our novel dependency graph processor subsumes several previous processors
from the literature, like
the ``leaf removal processor'' of
\cite{noschinski2013analyzing}.
Leaves of the dependency graph like \eqref{dpR1-6}
are not part of any SCC. Hence, they  are never contained in
SCC-prefixes and thus, the
annotations that only correspond to such leaves are always removed. 
For a similar reason,
$\Proc_{\mathtt{DG}}$ from \Cref{thm:DPG}
subsumes the ``rhs simplification processor'' of \cite{noschinski2013analyzing} and the related
``usable terms processor'' of \cite{FLOPS2024}, both of which share the same underlying idea.

\Cref{ex:two-sccs-together} shows that
only considering SCCs without prefixes would be unsound for analyzing
complexity and proving
$\SAST$. 

\begin{example}[Combining Non-Connected SCCs may Increase Complexity]\label{ex:two-sccs-together}
Recall $\R_2$ from \Cref{ex:LeadingExamples} which  is not $\SAST$
and $\ADPair{\R_2}$ from \Cref{ex:canonical-adps}.
The $\ADPair{\R_2}$-dependency graph is depicted
in \Cref{fig:dependency-graph-R2}.
When only considering the SCCs 
  $\mathcal{G}_{\tGeo}$ (containing the $\tGeo$-DP) and 
$\mathcal{G}_{\tF}$ (containing the three  $\tF$-DPs),
then we could falsely ``prove'' $\SAST$ since $\iota_{\langle
  \ADPair{\R_2}|_{\mathcal{G}_{\tGeo}},  \; \ADPair{\R_2}|_{\mathcal{G}_{\tGeo}} \rangle} = \Pol_0$ and
 $\iota_{\langle
  \ADPair{\R_2}|_{\mathcal{G}_{\tF}},  \; \ADPair{\R_2}|_{\mathcal{G}_{\tF}} \rangle} = \Exp$.
The problem is that for the $\tF$-SCC $\mathcal{G}_{\tF}$,
one also has to consider the DP
$\tStart \to \tF(\tgeo(\tz))$
which determines with which instantiations of the variables the $\tF$-SCC is called (i.e.,
it ensures that 
 the $\tgeo$-rules become usable). Indeed, we have
 $\iota(\langle \ADPair{\R_2}, \ADPair{\R_2} \rangle) = \omega$.
\end{example}

\begin{figure}
	\centering
    \small
    \begin{tikzpicture}
        \node[shape=rectangle,draw=black!100] (A) at (-2,1.2) {\footnotesize $\tStart \to \tF(\tgeo(\tz))$};
        \node[shape=rectangle,draw=black!100] (B) at (2,1.2) {\footnotesize $\tF(\ts(x)) \to \tF(\tc(x,x))$};
        \node[shape=rectangle,draw=black!100] (C1) at (-2,2) {\footnotesize $\tF(\tc(x,y)) \to \tF(x)$};
        \node[shape=rectangle,draw=black!100] (C2) at (2,2) {\footnotesize $\tF(\tc(x,y)) \to \tF(y)$};
        
        \node[shape=rectangle,draw=black!100] (D) at (-2,0.4) {\footnotesize $\tStart \to \tGeo(\tz)$};
        \node[shape=rectangle,draw=black!100] (E) at (2,0.4) {\footnotesize $\tGeo(x) \to \tGeo(\ts(x))$ };
        
        \path [->] (A) edge (B);
        \path [->] (A) edge (C1);
        \path [->] (A) edge (C2);
        \path [->] (B) edge (C1);
        \path [->] (B) edge (C2);
        \path [->,in=175, out=185, looseness=5] (C1) edge (C1);
        \path [->,in=5, out=355, looseness=5] (C2) edge (C2);
        \path [->] (C1) edge (C2);
        \path [->] (C1) edge (B);
        \path [->] (C2) edge (C1);
        \path [->] (C2) edge (B);

        \path [->] (D) edge (E);
        \path [->,in=5, out=355, looseness=5] (E) edge (E);
        
    \end{tikzpicture} 
    \caption{$\ADPair{\R_2}$-dependency graph}\label{fig:dependency-graph-R2}
\end{figure}

\subsection{Reduction Pair Processor}\label{sec:red-pair}

Now we lift 
the direct application of polynomial interpretations explained in \Cref{sec:TRS} 
to an ADP processor for complexity analysis of\linebreak PTRSs, which allows us
to apply polynomial interpretations $\I$
in a modular way.
As in the
classical DP approach \cite{arts2000termination,giesl2006mechanizing}, here it suffices\linebreak if 
$\I$ is \emph{weakly} monotonic, i.e., if $x\geq y$ implies $\I_f(...,x,...)\geq
\I_f(...,\linebreak y,...)$ 
for all $f \in \Sigma$ and $x,y \in \IN$.
Moreover, as in \cite{avanzini2020probabilistic,FLOPS2024}, to ensure\linebreak
``weak monotonicity w.r.t.\ expected values'' we restrict ourselves to interpretations
with multilinear polynomials, where all monomials have the form \pagebreak[3]
$c \cdot x_1^{e_1} \cdot \ldots \cdot
x_k^{e_k}$ with $c \in \IN$ and $e_1, \ldots, e_k \in \{0,1\}$.

The  \emph{reduction pair processor} imposes 
three requirements on $\I$:\\
(1) $\,$ All rules with the flag $\ttrue$
must be weakly decreasing in expectation when removing all annotations. Due to weak
monotonicity, this
ensures that 
evaluating the arguments of a function call (i.e., applying the rule in a context)
also decreases weakly in expectation.\\
(2) $\,$ All ADPs must be weakly  decreasing when comparing the
annotated left-hand side $\I(\ell^\sharp)$ with the expected value of the annotated subterms of the
right-hand side $\{ p_1:r_1, \ldots, p_k:r_k \}$. To
measure the value of a term $r_j$, here
we consider all its subterms $t \trianglelefteq_{\sharp} r_j$ at annotated positions
and  sum up the polynomial interpretations of all such $t^\sharp$, i.e., we
consider $\ISum^\sharp(r_j) = \sum_{t \trianglelefteq_{\sharp} r_j} \I(t^{\sharp})$.
Regarding this sum instead of the interpretation $\I(r_j)$ of the whole term $r_j$ is the
reason why only need \emph{weak} monotonicity.\\
(3) $\,$ 
Finally, the processor removes all strictly decreasing
ADPs from the component $\SSS$ of the ADP problem. However, the ADPs are still kept in
$\P$, because they may still be used in reductions.
Moreover, if $\I$ is a CPI, then the processor infers a polynomial bound
corresponding to the degrees of the polynomials used for annotated symbols. Otherwise, it
only infers an exponential bound (which is still useful when analyzing $\SAST$).

\begin{restatable}[Reduction Pair Proc.]{theorem}{ReductionPair}\label{RP-proc}
    Let $\I : \Sigma^{\sharp} \to \IN(\Var)$ be\linebreak a weakly monotonic, multilinear polynomial interpretation.
    Let $\langle \P, \SSS \rangle$ be an ADP problem where 
    $\P = \P_\geq \uplus \P_>$ and $\P_> \cap \SSS \neq \emptyset$ such that:

    \smallskip
    
  {\small  \noindent
    (1) 
     For every $\ell\!\to\! \{p_1\! :\! r_1, \ldots, p_k\! :\! r_k \}^{\ttrue}\! \in \!\P$:
     $\, \I(\ell)\! \geq\! \sum_{1\leq j \leq k} \, p_j \cdot \I(\flat(r_j))$

     \noindent
     (2)
     For every $\ell\! \to\! \{p_1\! :\! r_1, \ldots, p_k\! :\! r_k \}^m \!\in \!\P_\geq$:\,
     $\I(\ell^{\sharp})\! \geq\! \sum_{1\leq j \leq k} \, p_j \cdot \ISum^\sharp(r_j)$

     \noindent
     (3)
     For every $\ell\! \to\! \{p_1\! : r_1, \ldots, p_k\! : r_k \}^m \!\in\! \P_>$:\,
$\I(\ell^{\sharp})\! > \!\sum_{1\leq j \leq k} \,p_j \cdot \ISum^\sharp(r_j)$,

     \noindent
     where
    $\ISum^\sharp(r_j) = \sum_{t \trianglelefteq_{\sharp} r_j} \I(t^{\sharp})$.}

  \smallskip
  
    Then $\Proc_{\mathtt{RP}}(\langle \P, \SSS \rangle)
    = (c, \langle \P, \SSS \setminus \P_>  \rangle)$ is sound, where the complexity $c \in \F{C}$ is
    determined as follows:
    If $\I$ is a CPI and for all annotated symbols $f^\sharp  \in \SignatureA$, the polynomial
    $I_{f^\sharp}$ has at most degree $a$,
    then $c = \Pol_a$. 
    If $\I$
    is not a CPI, then $c = \Exp$ if all constructors are interpreted by linear
    polynomials, and otherwise $c = \DExp$. 
\end{restatable}

In contrast to the reduction pair processor for proving $\AST$ from \cite{FLOPS2024}, 
our processor cannot remove any annotation from $\P$.
The reason is that the ADPs from $\P_>$ may still be needed to reach all annotated
terms that are relevant for the complexity of $\P_{\geq}$.
(This problem already occurs
when analyzing complexity in
the non-probabilistic setting, see, e.g.,  \cite[Ex.\ 25]{noschinski2013analyzing}.)
But we can remove strictly decreasing ADPs from $\SSS$ and,
therefore, do not have to count them anymore for the complexity.
The complexity of the removed ADPs is accounted for by $c \in \{\Pol_a, \, \Exp, \,
\DExp\}$.

\begin{example}[Reduction Pair Processor]\label{ex:ch6-rloop-5}
  Consider 
  $\langle {\P_1}|_{\JJ_1}, {\P_1}|_{\JJ_1} \rangle$ from 
  \Cref{ex:dependency-graph-processor},
  and a polynomial interpretation $\I$ with
  $\I_{\tStart}(x,y) = 2$, $\I_{\tGeo}(x) = \I_\tQ = 1$, 
  and $\I_{\ts}(x) = \I_{\tgeo}(x) = x+1$.
  Then $\tgeo(x) \to \{\nicefrac{1}{2}:\tGeo(\ts(x)), \,
  \nicefrac{1}{2}:x\}^{\ttrue}$ \eqref{eq:ADP-1-2-scc}
  is weakly
  decreasing in expectation when disregarding annotations, since 
  $\I(\tgeo(x)) = x + 1 =
  \nicefrac{1}{2} \cdot
  \I(\tgeo(\ts(x))) +
  \nicefrac{1}{2} \cdot \I(x)$. 
  Moreover, when regarding annotations, then
  all ADPs are strictly decreasing: for \eqref{eq:ADP-1-2-scc} we have
  $\I(\tGeo(x)) = 1 >  \nicefrac{1}{2} = \nicefrac{1}{2} \cdot \I(\tGeo(\ts(x)))$, 
  for $\tstart(x,y) \to \{1:\tq(\tGeo(x), y,\linebreak y)\}^{\tfalse}$ \eqref{eq:ADP-1-1-scc} we have
  $\I(\tStart(x,y)) = 2 > 1 = \I(\tGeo(x))$, and for
  \eqref{eq:ADP-1-4-scc}-\eqref{eq:ADP-1-3-scc} we have $\I(\tQ(\ldots)) = 1 > 0$ (as their right-hand sides do not
  contain annotations). Since $\I$ is a CPI which interprets
  all annotated symbols as constants, we
  obtain 
  $\Proc_{\mathtt{RP}} (\langle {\P_1}|_{\JJ_1}, {\P_1}|_{\JJ_1} \rangle)  
  = (\Pol_0, \langle {\P_1}|_{\JJ_1}, \emptyset \rangle)$, containing
  a solved \pagebreak[3] ADP problem.

  For the other ADP problem $\langle {\P_1}|_{\JJ_2}, {\P_1}|_{\JJ_2} \rangle$ from 
  \Cref{ex:dependency-graph-processor},
  we use a polynomial interpretation with $\I_{\tStart}(x,y) = x+3$,
  $\I_{\tGeo}(x) = 1$, and $\I_\tQ(x,y,z) = \I_\ts(x) =\I_{\tgeo}(x) = x+1$.
  Then \eqref{eq:ADP-1-2-2-scc} is again weakly decreasing  when disregarding
    annotations. When regarding the annotations, then  
  the ADP \eqref{eq:ADP-1-5-2-scc} is weakly decreasing
  (since $\I(\tQ(x,\tz,\ts(z))) = x + 1 = \I(\tQ(x,\ts(z),\ts(z)))$), and all other 
  ADPs are strictly decreasing. 
  Since $\I$ is a CPI where $\tStart$ and $\tQ$ are interpreted as linear polynomials, 
  we get $\Proc_{\mathtt{RP}} (\langle {\P_1}|_{\JJ_2}, {\P_1}|_{\JJ_2} \rangle)  
  = (\Pol_1, \langle {\P_1}|_{\JJ_2}, \{\eqref{eq:ADP-1-5-2-scc}\}  \rangle)$.
  However,  there is no
  polynomial interpretation which orients \eqref{eq:ADP-1-5-2-scc} strictly 
  and the other ADPs weakly.
  Thus, we need another processor to solve the remaining problem.
\end{example}

\subsection{Knowledge Propagation Processor}\label{sec:knowledge-propagation}

The dependency graph can not only be used to decompose an ADP problem $\langle \P, \SSS
\rangle$ according to the
SCC-prefixes via the dependency graph processor, but it can also be used 
to remove an ADP $\alpha$ from $\SSS$ if all ``\emph{predecessors}'' of $\alpha$ have
already been taken into account. 
More precisely, let $\Pre(\alpha) \subseteq \P$
contain all ADPs
that can
``generate'' a redex for a step with $\alpha$ at an annotated position, i.e.,
 $\Pre(\alpha)$
consists of all ADPs $\beta \in \P$ such that there is an edge from some DP in 
$\nonprobDP^\bot(\beta)$ to some DP in $\nonprobDP^\bot(\alpha)$ in the $\P$-dependency
graph.
Note that $\nonprobDP^\bot(\alpha) \neq \emptyset$ for all ADPs $\alpha$.
If $d$ is the maximal number of annotated symbols in any term on a right-hand side of an
ADP from $\PP$, then in any $\P$-chain tree $\F{T}$, the probabilities of $\alpha$-steps
can be over-approximated as follows. Except for 
the very first step,
every $\bat$- or $\baf$-step
with $\alpha$ is preceded by a step with some ADP $\beta$ from $\Pre(\alpha)$.
Every term in  $\beta$'s right-hand side can 
trigger at most
$d$ $\alpha$-steps.
If the $\beta$-step had probability $p$, then adding all probabilities for these $\alpha$-steps
 yields at most $d \cdot p$.
Since the very first step of the tree might also be an $\alpha$-step, one obtains
\[\begin{tabular}{r@{\;\;}c@{\;\;}l}
$\edl_{\langle \P, \{ \alpha \} \rangle}(\F{T})$ &$\leq$& $1 +
\sum_{v \in V^{\F{T}} \setminus \ctleaf^{\F{T}}, \;
  \P(v) \in \Pre(\alpha) \times  \{\bat, \baf\}} \, d \cdot p_v$\\
&$=$& $1 + d
\cdot   \edl_{\langle \P, \Pre(\alpha) \rangle}(\F{T}).$
\end{tabular}
\]
This in turn implies
$\eirc_{\langle \P,\{ \alpha \} \rangle}(n) \leq 1 + d
\cdot \eirc_{\langle
  \P,\Pre(\alpha) \rangle}(n)$
for all $n \in \IN$, and thus,
$\iota_{\langle \P,\{ \alpha \} \rangle}  \sqsubseteq  \iota_{\langle \P,\Pre(\alpha) \rangle}$.

Hence, if $\alpha \in \SSS$ and
$\Pre(\alpha) \cap \SSS = \emptyset$ (i.e., $\Pre(\alpha) \subseteq \P \setminus \SSS$),
then
in any well-formed proof tree with a
node $v$ where $L_{\C{A}}(v) = \langle \P, \SSS \rangle$, the
ADPs from $\Pre(\alpha)$ have already been taken into account 
in the path $v_1,\ldots,v_k = v$ from the root node $v_1$ to $v$, i.e.,
\[ \iota_{\langle \P,\{ \alpha \} \rangle}  \sqsubseteq  \iota_{\langle \P,\Pre(\alpha)
  \rangle} \sqsubseteq 
 \iota_{\langle \P, \P \setminus \SSS \rangle} \sqsubseteq L_{\C{C}}(v_1) \oplus \dots
 \oplus L_{\C{C}}(v_{k-1}).\]
 As the proof tree already contains  knowledge about $\Pre(\alpha)$'s
 com\-plexity, the \emph{knowledge propagation processor} removes $\alpha$
from $\SSS$.

\begin{restatable}[Knowledge Propagation Proc.\ $\!$]{theorem}{KnowledgeProp}\label{def:kp-proc}
  Let $\langle \P,\SSS \rangle$ be an ADP problem, let $\alpha \in \SSS$ 
  and $\Pre(\alpha)
  \cap \SSS = \emptyset$, where  $\Pre(\alpha)$
  consists of all  ADPs $\beta \in \P$ such that there is an edge from some DP in 
  $\nonprobDP^\bot(\beta)$ to some DP in $\nonprobDP^\bot(\alpha)$ in the $\P$-dependency graph.
  Then the following processor is sound:
  \begin{align*}
      \Proc_{\mathtt{KP}}(\langle \P,\SSS \rangle) 
      = (\Pol_0, \langle \P, \SSS \setminus \{\alpha\} \rangle)
  \end{align*}
\end{restatable}

\begin{example}[Knowledge Propagation Processor]\label{ex:ch6-rloop-6}
 We still have to solve the ADP problem $\langle {\P_1}|_{\JJ_2},
 \{\eqref{eq:ADP-1-5-2-scc}\}  \rangle$ from \Cref{ex:ch6-rloop-5}.
 We have $\nonprobDP^\bot(\eqref{eq:ADP-1-5-2-scc}) = \{ \eqref{dpR1-5} \}$ and
 \pagebreak[3]
 the only DPs
with edges to $\eqref{dpR1-5}$ in the dependency graph of \Cref{fig:dependency-graph} are \eqref{dpR1-1} and
\eqref{dpR1-4}, where $\eqref{dpR1-1} \in \nonprobDP^\bot(\eqref{eq:ADP-1-1-2-scc})$ and
$\eqref{dpR1-4} \in \nonprobDP^\bot(\eqref{eq:ADP-1-4-2-scc})$.
Thus,  $\Pre(\eqref{eq:ADP-1-5-2-scc}) = \{ \eqref{eq:ADP-1-1-2-scc},
\eqref{eq:ADP-1-4-2-scc} \}$, i.e.,
in particular $\eqref{eq:ADP-1-5-2-scc} \notin
\Pre(\eqref{eq:ADP-1-5-2-scc})$. Hence, we can 
 apply the knowledge propagation processor and obtain $\Proc_{\mathtt{KP}}(\langle {\P_1}|_{\JJ_2},
 \{\eqref{eq:ADP-1-5-2-scc}\}  \rangle) = (\Pol_0, \langle  {\P_1}|_{\JJ_2}, \emptyset
 \rangle)$, again containing a solved ADP problem.

The 
 solved proof tree is shown in \Cref{fig:proof_tree_R1}.
 Thus, we inferred that $\R_1$ is $\SAST$ and its
 complexity is at most linear, i.e.,
  $\iota_{\R_1} \sqsubseteq \Pol_1$.
 
\end{example}

\begin{figure}
	\centering
    \scriptsize
    \begin{tikzpicture}
        \tikzstyle{adam}=[thick,draw=black!100,fill=white!100,minimum size=4mm, rectangle split parts=2,rectangle split horizontal]
        \tikzstyle{empty}=[rectangle,thick,minimum size=4mm]
        
        \node[empty,label=left:{\small \textcolor{red}{$\Pol_0$}}] at (0, 0)  (a) {$\langle \ADPair{\R_1}, \ADPair{\R_1} \rangle$};

        \node[empty,label=left:{\small \textcolor{red}{$\Pol_0$}}] at (0, -0.8)  (b)
             {$\langle
\{\eqref{eq:ADP-1-1-usablerule} -
    \eqref{eq:ADP-1-3-usablerule}\}, \{\eqref{eq:ADP-1-1-usablerule} -
    \eqref{eq:ADP-1-3-usablerule}\} \rangle$};

        \node[empty,label=left:{\small \textcolor{red}{$\Pol_0$}}] at (-2, -1.6)  (c)
             {$\langle
\{ \eqref{eq:ADP-1-1-scc} - \eqref{eq:ADP-1-3-scc} \}, \{ \eqref{eq:ADP-1-1-scc} - \eqref{eq:ADP-1-3-scc} \}\rangle$};
        \node[empty,label=left:{\small \textcolor{red}{$\Pol_0$}}] at (2, -1.6)  (d)
             {$\langle
\{ \eqref{eq:ADP-1-1-2-scc} - \eqref{eq:ADP-1-3-2-scc} \}, \{ \eqref{eq:ADP-1-1-2-scc} -
\eqref{eq:ADP-1-3-2-scc} \} \rangle$};

        \node[empty,label=left:{\small \textcolor{red}{$\Pol_0$}}] at (-2, -2.4)  (c1)
             {$\langle
\{ \eqref{eq:ADP-1-1-scc} - \eqref{eq:ADP-1-3-scc} \}, \emptyset \rangle$};

        \node[empty,label=left:{\small \textcolor{red}{$\Pol_1$}}] at (2, -2.4)  (d1)
             {$\langle
\{ \eqref{eq:ADP-1-1-2-scc} - \eqref{eq:ADP-1-3-2-scc} \}, \{ \eqref{eq:ADP-1-5-2-scc} \}  \rangle$};

        \node[empty,label=left:{\small \textcolor{red}{$\Pol_0$}}] at (2, -3.2)  (d11)
             {$\langle
\{ \eqref{eq:ADP-1-1-2-scc} - \eqref{eq:ADP-1-3-2-scc} \}, \emptyset  \rangle$};
        
        \draw (a) edge[->] node[right] {$\Proc_{\mathtt{UR}}$} (b);
        \draw (b) edge[->] node[right=19pt] {$\Proc_{\mathtt{DG}}$} (c);
        \draw (b) edge[->] (d);
        \draw (c) edge[->] node[right] {$\Proc_{\mathtt{RP}}$} (c1);
        \draw (d) edge[->] node[right] {$\Proc_{\mathtt{RP}}$} (d1);
        \draw (d1) edge[->] node[right] {$\Proc_{\mathtt{KP}}$} (d11);
    \end{tikzpicture} 
    \caption{Solved proof tree for $\R_1$}\label{fig:proof_tree_R1}
\end{figure}

\subsection{Probability Removal Processor}\label{sec:prob-removal}

Our framework may also yield
ADP (sub)problems with non-probabi\-lis\-tic structure, 
i.e., where every ADP has the form $\ell \to \{1:r\}^{m}$.
Then, the \emph{probability removal processor} can switch
to ordinary (non-probabilistic) DT problems
for complexity analysis from \cite{noschinski2013analyzing}.

These DT problems have four components 
$(\P, \SSS, \K, \R)$:
a set of dependency tuples $\P$,
the subset $\SSS \subseteq \P$
that is counted for complexity,
a subset $\K$ whose complexity has already been taken into account (see \Cref{K-Component}),
and a set of rewrite rules $\R$.

\begin{restatable}[Probability Removal Processor]{theorem}{ProbabilityRemovalProc}\label{theorem:prob-NPP}
	Let $\langle \PP, \SSS \rangle$ be an ADP problem where every ADP in $\PP$ has the form $\ell \to \{1:r\}^{m}$.
    Let
    $\nonprobDT(\ell \to \{1\!:r\}^{m}) = \ell^\sharp \to [t^\sharp_1, \ldots, t^\sharp_n]$ 
    if $\{t \mid t \trianglelefteq_{\sharp} r\} = \{t_1, \ldots,\linebreak
    t_n\}$, and let
   $\nonprobDT(\P) = \{\nonprobDT(\alpha) \mid \alpha \in \PP\}$.
    Then the expected runtime complexity of $\langle \P, \SSS \rangle$ is equal to the
    runtime
    complexity of the
    non-probabilistic DT problem $\beta = (\nonprobDT(\P),\nonprobDT(\SSS),\nonprobDT(\P \setminus \SSS), 
\nonprob(\P))$.  
    So the processor $\Proc_{\mathtt{PR}}(\langle \PP, \SSS \rangle) = (c, \emptyset)$
    is sound if the DT framework returns 
    $c$ as an upper bound on the runtime complexity of
$\beta$.
\end{restatable}

When proving $\AST$ as in \cite{FLOPS2024}, one should
move to the non-probabilistic DP framework for termination whenever possible, because then one can analyze function
calls in  right-hand sides of  rules separately.
In contrast, in
the non-probabilistic DT framework for complexity analysis,
one also has to consider all function calls from a right-hand side
simultaneously.
However, 
the switch to the non-probabilistic setting is still advantageous, since, e.g., the
reduction pair processor
of the non-probabilistic DT framework allows more orderings than just multilinear
polynomial interpretations.

\subsection{Further Processors}\label{sec:further}

In the non-probabilistic setting,
there are several further processors used to derive upper complexity bounds.
For example, the DP and DT frameworks of
\cite{arts2000termination,giesl2006mechanizing,noschinski2013analyzing}
also provide processors that use narrowing,
rewriting, or instantiations to 
\emph{transform} DP/DT problems to 
simplify the task of proving termination or complexity.
To show how to adapt such transformational processors 
for our novel ADP framework,
we
present a corresponding adapted
\emph{narrowing processor} in\report{ \Cref{sec:trans}.}\paper{  \cite{PPDPreport}.}

\pagebreak[3]
\section{Evaluation}\label{sec:Evaluation}

We implemented our new DP framework for upper bounds on the expected innermost runtime complexity 
in our termination prover \tool{AProVE}~\cite{JAR-AProVE2017}.
To this end, \tool{AProVE} first creates the canonical ADPs and then applies
processors according to the following strategy:

\begin{table}
  \centering
  \footnotesize
  \begin{center}
      \begin{tabular}{||c | c | c | c | c ||}
      \hline
      Strategy & Start Terms & \tool{POLO} & \natt{} & \aprove{} \\ [0.5ex] 
      \hline
      Full & Arbitrary & 30 & 33 & 33 \\ 
      \hline
      Full & Basic & 30 & 33 & 44 \\ 
      \hline
      Innermost & Arbitrary & 30 & 33 & 42 \\ 
      \hline
      Innermost & Basic & 30 & 33 & 55 \\ 
      \hline
      \end{tabular}
  \end{center}
  \caption{\vspace*{-.5cm}Number of $\SAST$ proofs found by each tool}\label{fig:tab_One}
\end{table}

First, we try
to apply the dependency graph processor $\Proc_{\mathtt{DG}}$, the usable rules processor
$\Proc_{\mathtt{UR}}$,
the knowledge propagation processor $\Proc_{\mathtt{KP}}$,
and the probability removal processor $\Proc_{\mathtt{PR}}$ in this order.
The advantage of these processors is that they do not rely on searching (i.e., they are very fast)
and they simplify
the ADP problem whenever they are applicable. If none of these processors can be applied
anymore, then we search for 
CPIs for the reduction pair processor  $\Proc_{\mathtt{RP}}$ in order to derive polynomial
complexity bounds for certain ADPs (otherwise, we try to apply
$\Proc_{\mathtt{RP}}$ with a non-CPI
polynomial
interpretation to derive an exponential bound).  
As soon as one of the processors is applied successfully, we re-start the strategy, since
other processors might be applicable again on the simplified 
subproblems.
Moreover, before the first application of the reduction pair processor, we
use\report{ the rule overlap instantiation processor $\Proc_{\mathtt{ROI}}$.}\paper{ our
adaption of the narrowing processor.}
Since it does not always help in inferring an upper bound and often increases the number of ADPs,
we use\report{ $\Proc_{\mathtt{ROI}}$}\paper{ this processor} only once on a fixed number of
terms.

For every PTRS, the user
can indicate whether one wants to analyze termination or complexity, 
consider arbitrary  or only basic start terms,
and whether one wants
to analyze innermost or full rewriting (with an arbitrary rewrite strategy).
Since our novel DP framework only works for innermost rewriting and basic start terms,
if the user asks for complexity analysis or $\SAST$ of full rewriting,
we check
whether the PTRS belongs to a known class where,
e.g., upper bounds on the expected innermost runtime complexity are upper bounds
w.r.t.\
an arbitrary rewrite strategy as well. 
Such properties were studied in \cite{FoSSaCS2024,FoSSaCS-Journal}.
If one wants to consider arbitrary instead of basic
start terms, we perform the transformation
of \cite{fuhs2019transformingdctorc}
(adapted to PTRSs, see 
\cite{FoSSaCS2024,FoSSaCS-Journal})
in order to 
move from derivational to runtime complexity, i.e.,  the PTRS $\R$ is transformed 
into a new PTRS $\R'$ such that the complexity of $\R'$ on basic start terms is a bound
on the 
complexity of $\R$ on all start terms.

For our evaluation,
we used the benchmark set of all 128 PTRSs from
the \emph{Termination Problem Data Base} \cite{TPDB}, i.e., the benchmarks considered for the
annual \emph{Termination and Complexity Competition} \cite{termcomp}, containing 128
typical probabilistic programs, including examples with complicated probabilistic
structure and probabilistic algorithms on lists and 
trees.
Note that this set  also contains many examples that are $\AST$, but not
$\SAST$. Therefore, we extended the collection by  10 additional examples 
that  are interesting
for
expected complexity analysis (including all examples from our paper).

To evaluate our novel framework for proving $\SAST$,
we compare \tool{AProVE}  with its
previous version (called ``\tool{POLO}'' in \Cref{fig:tab_One})
whose 
only way to 
prove 
$\SAST$ was to search for a monotonic, multilinear polynomial interpretation such that
all rules of the PTRS are strictly decreasing \cite{avanzini2020probabilistic}, and with
the
tool \tool{NaTT} \cite{natt_sys_2014} that implements polynomial and matrix interpretations to prove
$\SAST$.
Thus, \tool{POLO} and \tool{NaTT}  neither consider a specific rewrite
strategy nor start terms.  As shown by the
experiments on all 138 PTRSs in \Cref{fig:tab_One},
our novel ADP framework
increases
the power of proving $\SAST$
significantly.
Here, the numbers indicate for how many PTRSs the respective tool proved $\SAST$.
Note that as usual for benchmarks in termination or complexity analysis,
the termination behavior or runtime complexity of the benchmarks
is in general not known, i.e., the best quantitative performance measure for a tool is the comparison to other tools.

\begin{table}
  \centering
  \footnotesize
  \begin{center}
      \begin{tabular}{||c | c | c | c | c | c | c | c ||}
      \hline
      Strategy & Start Terms & $\Pol_0$ & $\Pol_1$ & $\Pol_2$ & $\Exp$ & $\DExp$ & $\omega$ \\ [0.5ex] 
      \hline
      Full & Arbitrary & 2 & 3 & 0 & 26 & 2 & 105 \\ 
      \hline
      Full & Basic & 2 & 13 & 0 & 27 & 2 & 94 \\ 
      \hline
      Innermost & Arbitrary & 4 & 10 & 3 & 24 & 1 & 96 \\ 
      \hline
      Innermost & Basic & 14 & 11 & 6 & 23 & 1 & 83 \\ 
      \hline
      \end{tabular}
  \end{center}
  \caption{\vspace*{-.5cm}Runtime complexities proven by \aprove{}}\label{fig:tab_Two}
\end{table}

\Cref{fig:tab_Two} shows the
upper bounds inferred by \aprove.
So \aprove{} obtains
numerous constant and/or linear
bounds, even for full rewriting and/or arbitrary start terms.
Note that in contrast to the non-probabilistic setting, a PTRS with expected constant
runtime
is not necessarily trivial as it
can have evaluations of unbounded (and even infinite) length.
However, the transformation of \cite{fuhs2019transformingdctorc}
to move  from arbitrary to basic start terms may add rules with\linebreak linear runtime
to the PTRS. This explains
the low number of constant upper bounds for arbitrary start
terms.
Due to the restriction to multilinear polynomial interpretations,
$\Pol_a$ for $a > 1$ can currently only be inferred from multilinear, but non-linear
interpretations like $\I_{f^\sharp}(x,y) = x
\cdot y$. In the future, we intend to extend our implementation to
 also use,
e.g., 
matrix orderings \cite{DBLP:journals/jar/EndrullisWZ08}
in
order to improve the inference of polynomial bounds of higher degrees.

For more details on our experiments, the collection of examples,
and for instructions on how to run our implementation
in \textsf{AProVE} via its \emph{web interface}
  or locally, we refer to:
\begin{center}
  \url{https://aprove-developers.github.io/PTRSExpectedRuntime/}
\end{center}

\section{Conclusion}\label{sec:Conclusion}

In this paper, we presented
the first DP framework 
to infer upper bounds 
on the expected innermost runtime complexity of PTRSs automatically. 
Our implementation in \aprove{}
is the first tool for automatic complexity analysis
of PTRSs and it improves substantially over previous tools to analyze $\SAST$ of PTRSs.

There are several directions for future work,
e.g., by extending the reduction pair processor of  \Cref{RP-proc} to other orderings, by
adapting  further transformational processors to our new ADP framework,
and
by developing variants of our framework that are directly applicable for full instead of
innermost rewriting and/or for arbitrary instead of basic start terms.

\begin{acks}
 {\small This work is supported \pagebreak[3] by the \grantsponsor{ 2236 }{DFG Research
    Training Group} \grantnum{ 2236 }{ UnRAVeL}.}
\end{acks}

\bibliographystyle{ACM-Reference-Format}
\bibliography{biblio}

\report{\appendix

\section{Additional Theory and Proofs}\label{Additional Theory}

In this appendix,
we first present some additional theory
on the  size of a term under a polynomial interpretation
in \Cref{subsec:poly-interpret}. 
Afterwards, we prove all our theorems and lemmas in \Cref{subsec:proofs}.
Finally, we adapt the transformational \emph{narrowing processor} to our novel
ADP framework in \Cref{sec:trans}. 

\subsection{Size of Polynomial Interpretations}\label{subsec:poly-interpret}
We first analyze the size of $\I_0(t)$ for every term $t$ and every polynomial interpretation $\I$.
Here, $\I_0$ denotes the interpretation that behaves like $\I$, and
in addition, maps every variable to $0$.

\subsubsection{Size of Constructor Terms}
Let $\I$ be a CPI.
Recall that a $\CPI$ $\I$ guarantees that $\I(f(x_1, \ldots, x_k)) = a_1 x_1 + \ldots + a_k x_k + b$, 
with $a_j \in \{0,1\}, b\in \IN$ for all constructors $f \in \Sigma_{C}$.
Let $s \in \TSet{\SignatureC}{\VSet}$ be a constructor term, i.e., a term
containing only constructors and variables. 
Let $b_{max}$ be the maximum of all $b$ for all constructors $f \in \Sigma_C$.
For $|s| = 1$ we get $\I_0(s) \leq b_{max}$, 
and for $|s| > 1$ we get $\I_0(s) \leq \I_0(s') + b_{max}$ 
for some constructor term $s'$ with $|s'| = |s| - 1$.
Hence, for every term $s \in \TSet{\SignatureC}{\VSet}$
we obtain $\I_0(s) \leq b_{max} \cdot |s|$.

Next, let $\I$ not be a CPI but linear on constructors,
i.e., $\I(f(x_1,\linebreak \ldots, x_k)) = a_1 x_1 + \ldots + a_k x_k + b$, 
with $a_1, \ldots, a_n, b\in \IN$ for all constructors $f \in \Sigma_{C}$.
Let $b_{max}$ be the maximum of all $b$ and $a_{max}$ be the maximum of all $a$.
Again, let $s \in \TSet{\SignatureC}{\VSet}$.
For $|s| = 1$ we get $\I_0(s) \leq b_{max}$, 
and for $|s| > 1$ we get $\I_0(s) \leq a_{max} \cdot \I_0(s') + b_{max}$ 
for some constructor term $s'$ with $|s'| = |s| - 1$.
This is a first-order non-homogeneous linear recurrence relation, 
which can be upper-bounded by $2^{\pl(|s|)}$, for some polynomial $\pl(n)$ in $n$,
i.e., we have $\I_0(s) \in \mathcal{O}(2^{\pl(|s|)})$.

Finally, if $\I$ is not a CPI and not even linear on constructors, 
e.g., $\I(f(x_1, x_2)) = x_1 \cdot x_2$ or $\I(f(x_1, x_2)) = x_1^2 + x_2^2$,
then we can only guarantee that $\I_0(s) \leq 2^{2^{\pl(|s|)}}$, for some polynomial $\pl(n)$ in $n$.
To be precise,
let  $a_{max}$ be the maximum
coefficient of all non-constant monomials in all interpretations $I_f$ of constructors $f$,
and 
let $b_{max}$ be the maximum constant monomial in all these interpretations. 
Moreover, let $c_{max}$ be the maximum degree of all $\I(f(x_1, \ldots, x_k))$, 
and $d$ be the maximal arity for all constructors $f \in \Sigma_C$.
For $|s| = 1$ we again have $\I_0(s) \leq b_{max}$, 
and for $|s| > 1$ we get $\I_0(s) \leq 2^d \cdot a_{max} \cdot \I_0(s')^{c_{max}} + b_{max}$ 
for some constructor term $s'$ with $|s'| = |s| - 1$.
After conversion to homogeneous form, and taking the $\log$ on both sides,
we result in a first-order linear recurrence relation again.
Hence, we have $\log(\I_0(s)) \leq 2^{\pl(|s|)}$, for some polynomial $\pl(n)$,
i.e., we have $\I_0(s) \in \mathcal{O}(2^{2^{\pl'(|s|)}})$, for some polynomial $\pl'(n)$.

\subsubsection{Size of Basic Terms}
First, let $\I$ be a CPI where
$\I(f(x_1, \ldots,\linebreak x_k)) = a_1 x_1 + \ldots + a_k x_k + b$, 
with $a_i \in \{0,1\}, b\in \IN$ for all constructors $f \in \Sigma_{C}$.
Again, let $b_{max}$ be the maximum of all $b$ for all constructors $f \in \Sigma_C$.
Moreover, let $t = f(t_1, \ldots, t_k) \in \TB_{\R}$.
Here, we have
\begin{align*}
    \I_0(t) &= \I_{f}(\I_0(t_1), \ldots, \I_0(t_k))\\
    &\leq \I_{f}(b_{max} \cdot |t_1|, \ldots, b_{max} \cdot |t_k|)\\
    &\leq \I_{f}(b_{max} \cdot |t|, \ldots, b_{max} \cdot |t|)\\
    &\leq b^a_{max} \cdot \I_{f}(|t|, \ldots, |t|), \quad \text{where $a$ is the degree of } \I_{f}\!
\end{align*}
Hence, we have
$\I_0(t) \in \mathcal{O}(|t|^a)$.

If $\I$ is only linear on constructors,
we get $\I_0(t) \in \mathcal{O}(2^{\pl(|t|)})$, 
and if $\I$ is not linear, we obtain $\I_0(t) \in \mathcal{O}(2^{2^{\pl(|t|)}})$ in a similar fashion.

\subsubsection{Size of Arbitrary Terms}
If we have $\I(f(x_1, \ldots, x_k)) = a_1 x_1 + \ldots + a_k x_k + b$
with $a_j \in \{0,1\}, b\in \IN$ for all $f \in \Sigma$, then $\I_0(t) \in \mathcal{O}(|t|)$.
If $\I$ is linear for all $f \in \Sigma$, then $\I_0(t) \in \mathcal{O}(2^{\pl(|t|)})$,
and if $\I$ is not linear, we obtain $\I_0(t) \in \mathcal{O}(2^{2^{\pl(|t|)}})$.

\subsubsection{Bounds on Expected Complexity via Polynomials}
Assume that there exists a monotonic polynomial interpretation $\I$ such that $\I(\ell) > \I(r)$ for every rule $\ell \to r \in \R$.
Since the interpretation is decreasing for every rule and bounded from below by $0$,
the interpretation $\I(t)$ of a term $t$ gives an upper bound on its derivation height.
Hence, e.g., if $\I$ is a CPI, then the runtime complexity is at most $\Pol_a$, where $a$
is the highest degree of $\I_f$ for
any defined symbol $f$.

\subsection{Proofs}\label{subsec:proofs}

To prove the soundness of the chain criterion, we proceed as follows:
Given an $\R$-RST $\F{T}$ that starts with the basic term $t$, we create an $\ADPair{\R}$-CT $\F{T}'$ that starts with $t^{\sharp}$ and
mirrors every rewrite step by using the corresponding ADP for the used rewrite rule.
We will show that each rewrite step takes place at an annotated position, i.e.,
we always use case $(\mathbf{at})$ with the ADPs.
Since the resulting CT has the same structure as the original RST (same nodes with the same probabilities) 
and we only have $(\mathbf{at})$-steps (that we count for the expected derivation length of a CT),
we get $\edl(\F{T}) = \edl_{\langle \ADPair{\R}, \ADPair{\R} \rangle}(\F{T}')$ for every such RST $\F{T}$,
and hence $\edh_\R(t) \leq \edh_{\langle \ADPair{\R}, \ADPair{\R} \rangle}(t)$ for every basic term $t$.
For the other direction, i.e., $\edh_\R(t) \geq \edh_{\langle \ADPair{\R}, \ADPair{\R} \rangle}(t)$,
we transform every $\ADPair{\R}$-CT $\F{T}$ into an $\R$-RST $\F{T}'$ by simply removing all occurring annotations.
Then, we directly get $\edl(\F{T}') = \edl_{\langle \ADPair{\R}, \ADPair{\R} \rangle}(\F{T})$ again.

When applying an ADP, we may remove annotations of normal forms,
hence, we define the set of all positions of subterms that are not in normal form.
During the construction of the proof of the chain criterion, we show that at least those
positions are annotated.
In the following, as usual we say that two positions $\pi_1$ and $\pi_2$ are \emph{orthogonal}
(or \emph{parallel}) if $\pi_1$ is not above $\pi_2$ and $\pi_2$ is not above $\pi_1$.

\begin{definition}[{\normalfont{$\PosDPoss$}}]
  \label{def:prop-important-sets}
  Let $\R$ be a PTRS.
  For a term $t \in \TT$ we define $\normalfont{\PosDPoss}(t) = \{\pi \mid \pi \in \posD(t), t|_\pi \notin \NF_{\R}\}$.
\end{definition}

\ChainCrit*

\begin{proof}
    In the following, we will often implicitly use that for an annotated term $t \in \TT^\sharp$, we have $\flat(t) \in \ANF_{\R}$ 
    iff $t \in \ANF_{\ADPair{\R}}$ since a rewrite rule and its corresponding canonical annotated dependency pair have the same left-hand side.
    \smallskip

    \noindent
    \textbf{Soundness} ($\edh_\R(t) \leq \edh_{\langle \ADPair{\R}, \ADPair{\R} \rangle}(t)$):
    Let $\F{T}=(V,E,L)$ be an $\R$-RST whose root is labeled with $(1:t)$ for some term $t \in \TB_{\R}$.
    We create an $\ADPair{\R}$-CT $\F{T}'=(V,E,L')$ that starts with $t^{\sharp}$ and
    mirrors every rewrite step by using the ADP corresponding to
    the used rewrite rule.
    Moreover, we ensure that every rewrite step is an $(\mathbf{at})$-step, and hence, 
    we have $\edl(\F{T}) = \edl_{\langle \ADPair{\R}, \ADPair{\R} \rangle}(\F{T}')$.
    
    We construct the new labeling $L'$ for the $\ADPair{\R}$-CT inductively such that for all inner nodes $x \in V \setminus \ctleaf$ 
    with children nodes $xE = \{y_1,\ldots,y_k\}$ we have 
    $t_x' \tored{}{}{\ADPair{\R}} \{\tfrac{p_{y_1}}{p_x}:t_{y_1}', \ldots, \tfrac{p_{y_k}}{p_x}:t_{y_k}'\}$ 
    and use Case $(\mathbf{at})$.
    Let $X \subseteq V$ be the set of nodes $x$ where we have already defined the labeling $L'(x)$.
    During our construction, we ensure that the following property holds for all $x \in X$:
    \begin{equation} \label{chain-crit-1-soundness-induction-hypothesis}
        \flat(t_x) = \flat(t_x') \land \PosDPoss(t_x) \subseteq \posT(t_x').
    \end{equation}
    This means that the corresponding term $t_x$ for the node $x$ in $\F{T}$ has the same structure as the term $t_x'$ in $\F{T}'$, 
    and additionally, all the possible redexes in $t_x$ are annotated in $t_x'$.

    \vspace*{-0.4cm}
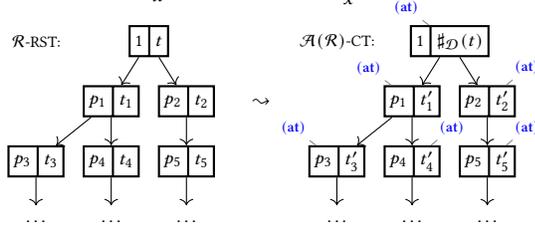
    \begin{center}
    \scriptsize
    \begin{tikzpicture}
        \tikzstyle{adam}=[thick,draw=black!100,fill=white!100,minimum size=4mm, shape=rectangle split, rectangle split parts=2,rectangle split horizontal] \tikzstyle{empty}=[rectangle,thick,minimum size=4mm]

        \node[empty] at (-3, 0) (rst) {$\R$-RST:};
        \node[adam] at (-1.5, 0) (a) {$1$
        \nodepart{two} $t$};
        \node[adam] at (-2, -0.8) (b) {$p_1$
        \nodepart{two} $t_{1}$};
        \node[adam] at (-1, -0.8) (c) {$p_2$
        \nodepart{two} $t_{2}$};
        \node[adam] at (-3, -1.6) (d) {$p_3$
        \nodepart{two} $t_3$};
        \node[adam] at (-2, -1.6) (e) {$p_4$
        \nodepart{two} $t_4$};
        \node[adam] at (-1, -1.6) (f) {$p_5$
        \nodepart{two} $t_5$};
        \node[empty] at (-3, -2.4) (g) {$\ldots$};
        \node[empty] at (-2, -2.4) (h) {$\ldots$};
        \node[empty] at (-1, -2.4) (i) {$\ldots$};

        \node[empty] at (0, -0.8) (arrow) {$\leadsto$};

        \node[empty] at (1, 0) (ct) {$\ADPair{\R}$-CT:};
        \node[adam,pin={[pin distance=0.1cm, pin edge={,-}] 140:\tiny \textcolor{blue}{$(\mathbf{at})$}}] at (2.5, 0) (a2) {$1$
            \nodepart{two} $\annoD(t)$};
        \node[adam,pin={[pin distance=0.1cm, pin edge={,-}] 140:\tiny \textcolor{blue}{$(\mathbf{at})$}}] at (2, -0.8) (b2) {$p_1$
            \nodepart{two} $t'_1$};
        \node[adam,pin={[pin distance=0.1cm, pin edge={,-}] 45:\tiny \textcolor{blue}{$(\mathbf{at})$}}] at (3, -0.8) (c2) {$p_2$
            \nodepart{two} $t'_2$};
        \node[adam,pin={[pin distance=0.1cm, pin edge={,-}] 140:\tiny \textcolor{blue}{$(\mathbf{at})$}}] at (1, -1.6) (d2) {$p_3$
            \nodepart{two} $t'_3$};
        \node[adam,pin={[pin distance=0.1cm, pin edge={,-}] 45:\tiny \textcolor{blue}{$(\mathbf{at})$}}] at (2, -1.6) (e2) {$p_4$
            \nodepart{two} $t'_4$};
        \node[adam,pin={[pin distance=0.1cm, pin edge={,-}] 45:\tiny \textcolor{blue}{$(\mathbf{at})$}}] at (3, -1.6) (f2) {$p_5$
            \nodepart{two} $t'_5$};
        \node[empty] at (1, -2.4) (g2) {$\ldots$};
        \node[empty] at (2, -2.4) (h2) {$\ldots$};
        \node[empty] at (3, -2.4) (i2) {$\ldots$};

        \draw (a) edge[->] (b);
        \draw (a) edge[->] (c);
        \draw (b) edge[->] (d);
        \draw (b) edge[->] (e);
        \draw (c) edge[->] (f);
        \draw (d) edge[->] (g);
        \draw (e) edge[->] (h);
        \draw (f) edge[->] (i);

        \draw (a2) edge[->] (b2);
        \draw (a2) edge[->] (c2);
        \draw (b2) edge[->] (d2);
        \draw (b2) edge[->] (e2);
        \draw (c2) edge[->] (f2);
        \draw (d2) edge[->] (g2);
        \draw (e2) edge[->] (h2);
        \draw (f2) edge[->] (i2);
    \end{tikzpicture}
    \end{center}
    \vspace*{-0.4cm}

    We label the root of $\F{T}'$ with $\annoD(t)$.
    Here, we obviously have $\flat(t) = \flat(\annoD(t))$ and $\PosDPoss(t) \subseteq \posD(t) = \posT(\annoD(t))$.
    As long as there is still an inner node $x \in X$ such that its successors are not contained in $X$, we do the following.
    Let $xE = \{y_1, \ldots, y_k\}$ be the set of its successors.
    We need to define the corresponding terms $t_{y_1}', \ldots, t_{y_k}'$ for the nodes $y_1, \ldots, y_k$.
    Since $x$ is not a leaf, we have $t_x \itor \{\tfrac{p_{y_1}}{p_x}:t_{y_1}, \ldots, \tfrac{p_{y_k}}{p_x}:t_{y_k}\}$.
    This means that there is a rule $\ell \to \{p_1:r_1, \ldots, p_k:r_k\} \in \R$, a position $\pi$, 
    and a substitution $\sigma$ such that ${t_x}|_\pi = \ell\sigma \in \ANF_{\R}$.
    Furthermore, we have $t_{y_j} = t_x[r_j \sigma]_{\pi}$ for all $1 \leq j \leq k$.
    So the labeling of the successor $y_j$ in $\F{T}$ is $L(y_j) = (p_x \cdot p_j: t_x[r_j\sigma]_\pi)$ for all $1 \leq j \leq k$.

    The corresponding ADP for the rule is $\ell \to \{ p_1 : \annoD(r_1), \ldots, p_k : \annoD(r_k) \}^{\ttrue}$.
    Furthermore, $\pi \in \PosDPoss(t_x) \subseteq \posT(t_x')$ and $\flat(t_x) =
    \flat(t_x')$ by the induction hypothesis.
    Hence, we can rewrite $t_x'$ with $\ell \to \{ p_1 : \annoD(r_1), \ldots, p_k : \annoD(r_k) \}^{\ttrue}$, 
    using the position $\pi$ and the substitution $\sigma$, and Case $(\mathbf{at})$ applies.
    We get $t_x' \tored{}{}{\ADPair{\R}} \{p_1: t_{y_1}', \ldots, p_k: t_{y_k}'\}$ with $t_{y_j}' = t_x'[\annoD(r_j) \sigma]_{\pi}$.
    This means that we have $\flat(t_{y_j}) = \flat(t_{y_j}')$.
    
    It remains to prove $\PosDPoss(t_{y_j}) \subseteq \posT(t_{y_j}')$ for all $1 \leq j \leq k$.
    For all $\tau \in \PosDPoss(t_{y_j}) = \PosDPoss(t_x[r_j \sigma]_{\pi})$ that are
    orthogonal to or above $\pi$, 
    we have $\tau \in \PosDPoss(t_{x},\R) \subseteq \posT(t_{x}')$ by the induction hypothesis, 
    and all annotations orthogonal to or above $\pi$ remain in $t_{y_j}'$ as they were in $t_{x}'$.
    For all positions $\tau \in \PosDPoss(t_{y_j}) = \PosDPoss(t_x[r_j \sigma]_{\pi})$ that are below $\pi$, 
    we know that, due to innermost evaluation, at least the defined root symbol of a term that is not in normal form must be inside $r_j$, 
    and thus $\tau \in \posT(t_{y_j}')$, as all defined symbols of $r_j$ are annotated in $t_{y_j}' = t_x'[\annoD(r_j) \sigma]_{\pi}$.
    \smallskip

    \noindent
    \textbf{Completeness} ($\edh_\R(t) \geq \edh_{\langle \ADPair{\R}, \ADPair{\R} \rangle}(t)$): 
    Let $\F{T}=(V,E,L)$ be an $\ADPair{\R}$-RST whose root is labeled with $(1:t^\sharp)$ for some term $t \in \TB_{\R}$.
    We create an $\R$-RST $\F{T}'=(V,E,L')$ that starts with $t$ and
    mirrors every rewrite step by using the corresponding original probabilistic rule for the used ADP.
    Hence, we get $\edl_{\langle \ADPair{\R}, \ADPair{\R} \rangle}(\F{T}) = \edl(\F{T}')$.

    We label all nodes $x \in V$ in $\F{T}'$ with $\flat(t_x)$, where $t_x$ is the term for the node $x$ in $\F{T}$, i.e., we remove all annotations.
    We only have to show that $\F{T}'$ is indeed a valid RST, i.e., that the edge relation represents valid rewrite steps with $\to_{\R}$, 
    but this follows directly from the fact that if we remove all annotations in \Cref{def:annotated-dps}, 
    then we get the ordinary probabilistic term rewrite relation again.

    \vspace*{-0.1cm}
    \begin{center}
    \scriptsize
    \begin{tikzpicture}
        \tikzstyle{adam}=[thick,draw=black!100,fill=white!100,minimum size=4mm, shape=rectangle split, rectangle split parts=2,rectangle split horizontal] \tikzstyle{empty}=[rectangle,thick,minimum size=4mm]

        \node[empty] at (-3, 0) (ct) {$\ADPair{\R}$-CT:};
        \node[adam] at (-1.6, 0) (a) {$1$
        \nodepart{two} $t$};
        \node[adam] at (-2.2, -0.8) (b) {$p_1$
        \nodepart{two} $t_{1}$};
        \node[adam] at (-1, -0.8) (c) {$p_2$
        \nodepart{two} $t_{2}$};
        \node[adam] at (-3.4, -1.6) (d) {$p_3$
        \nodepart{two} $t_3$};
        \node[adam] at (-2.2, -1.6) (e) {$p_4$
        \nodepart{two} $t_4$};
        \node[adam] at (-1, -1.6) (f) {$p_5$
        \nodepart{two} $t_5$};
        \node[empty] at (-3.4, -2.4) (g) {$\ldots$};
        \node[empty] at (-2.2, -2.4) (h) {$\ldots$};
        \node[empty] at (-1, -2.4) (i) {$\ldots$};

        \node[empty] at (0, -0.8) (arrow) {$\leadsto$};

        \node[empty] at (1, 0) (rst) {$\R$-RST:};
        \node[adam] at (2.8, 0) (a2) {$1$
        \nodepart{two} $\flat(t)$};
        \node[adam] at (2.2, -0.8) (b2) {$p_1$
        \nodepart{two} $\flat(t_1)$};
        \node[adam] at (3.4, -0.8) (c2) {$p_2$
        \nodepart{two} $\flat(t_2)$};
        \node[adam] at (1, -1.6) (d2) {$p_3$
        \nodepart{two} $\flat(t_3)$};
        \node[adam] at (2.2, -1.6) (e2) {$p_4$
        \nodepart{two} $\flat(t_4)$};
        \node[adam] at (3.4, -1.6) (f2) {$p_5$
        \nodepart{two} $\flat(t_5)$};
        \node[empty] at (1, -2.4) (g2) {$\ldots$};
        \node[empty] at (2.2, -2.4) (h2) {$\ldots$};
        \node[empty] at (3.4, -2.4) (i2) {$\ldots$};

        \draw (a) edge[->] (b);
        \draw (a) edge[->] (c);
        \draw (b) edge[->] (d);
        \draw (b) edge[->] (e);
        \draw (c) edge[->] (f);
        \draw (d) edge[->] (g);
        \draw (e) edge[->] (h);
        \draw (f) edge[->] (i);

        \draw (a2) edge[->] (b2);
        \draw (a2) edge[->] (c2);
        \draw (b2) edge[->] (d2);
        \draw (b2) edge[->] (e2);
        \draw (c2) edge[->] (f2);
        \draw (d2) edge[->] (g2);
        \draw (e2) edge[->] (h2);
        \draw (f2) edge[->] (i2);
    \end{tikzpicture}
    \end{center}
    \vspace*{-0.2cm}
\end{proof}

\WellFormednessPreservation*

\begin{proof}
  For the first condition required for well-formed proof trees, note that
in $\F{P}$, we had
  $L_{\C{C}}'(v) =
\iota_{\langle \PP, \SSS \rangle}$, because $v$ was a leaf. Since $v$
is not a leaf anymore  in $\F{P}'$, here we have
$L_{\C{C}}'(v) = L_{\C{C}}(v) = c$.
However, $v$ now has the children $w_1, \ldots, w_n$. Let $v_1, \ldots, v_k = v$ be the
path from the root to $v$. Hence, we obtain
\begin{align*}
   & L_{\C{C}}(v_1) \oplus \ldots \oplus L_{\C{C}}(v_{k-1})
\oplus
\iota_{\langle \PP, \SSS \rangle}\\
\sqsubseteq \;& L_{\C{C}}(v_1) \oplus  \ldots \oplus L_{\C{C}}(v_{k-1})
\oplus c \oplus
\iota_{\langle \PP_1, \SSS_1 \rangle}
\oplus \ldots \oplus
\iota_{\langle \PP_n, \SSS_n \rangle} \tag{$\dagger$}\\
=  \;& L_{\C{C}}(v_1) \oplus  \ldots \oplus L_{\C{C}}(v_{k-1})
\oplus L_{\C{C}}(v_{k}) \oplus L_{\C{C}}'(w_1) \oplus \ldots 
L_{\C{C}}'(w_n)
\end{align*}
Here, $(\dagger)$ holds due to the first
condition \eqref{eq:soundness-1} of \cref{def:sound_processors} 
Thus, $\F{P}'$ also satisfies
the first condition for well-formed proof trees.

  Now we consider the second condition required for well-formed proof trees.
       Since $\F{P}$ is well formed and $\Proc$ is sound,
      the second condition \eqref{eq:soundness-2} of \cref{def:sound_processors} 
    implies that $\iota_{\langle \P_{i}, \P_i \setminus \SSS_i \rangle} \sqsubseteq L_{\C{C}}(v_1) \oplus \dots \oplus L_{\C{C}}(v_{k-1}) \oplus L_{\C{C}}(v_{k})$
    with $L_{\C{C}}(v_{k}) = c$ holds for all $1 \leq i \leq n$.
    Hence, $\F{P}'$ is well formed as well.
\end{proof}

\UsableRules*

\begin{proof}
    Let $\F{P}$ be a well-formed proof tree with $L_{\C{A}}(v) = \langle \P, \S \rangle$ and let $v_1,\ldots,v_k = v$ be the path from the root node $v_1$ to $v$.
    Moreover, let $L_{\C{A}}(w) = \langle \P', \S' \rangle$ for the only successor $w$ of $v$ in the proof tree.

    We first show that \eqref{eq:soundness-1} holds, i.e.,
    \begin{align*}
        \iota_{\langle \P, \SSS \rangle} \sqsubseteq L_{\C{C}}(v_1) \oplus ... \oplus L_{\C{C}}(v_{k-1}) \oplus \Pol_0 \oplus
        \iota_{\langle \P', \SSS' \rangle}
    \end{align*}
    Every $\P$-CT can also be seen as a $\P'$-CT, since in innermost reductions, 
    variables are always instantiated with normal forms. Thus,
  the only rules applicable to the right-hand sides of ADPs are the usable rules.
    Additionally, we start with basic terms, and hence, in every $\P$-CT only usable rules can be applied below the root of an annotated subterm.
    Thus, we have $\iota_{\langle \P, \SSS \rangle} = \iota_{\langle \P', \SSS'
      \rangle}$, which directly implies
    \eqref{eq:soundness-1}.
    
  The condition \eqref{eq:soundness-2} follows by the same reasoning:
    Every $\P'$-CT can also be seen as a $\P$-CT, so that 
    $\iota_{\langle \P', \P' \setminus \SSS' \rangle} = \iota_{\langle \P, \P \setminus
      \SSS \rangle}$. By well-formedness of
 $\F{P}$, we have
      $\iota_{\langle \P, \P \setminus \SSS \rangle} \sqsubseteq L_{\C{C}}(v_1) \oplus
    \dots \oplus L_{\C{C}}(v_{k-1})$. Thus, 
  $\iota_{\langle \P', \P' \setminus \SSS' \rangle} = \iota_{\langle \P, \P \setminus
      \SSS \rangle}$ implies 
\eqref{eq:soundness-2}.
\end{proof}

\renewcommand{\myfootnote}[1]{}
\DependencyGraph*
\renewcommand{\myfootnote}[1]{\footnote{#1}}

\begin{proof}
    Let $\F{P}$ be a well-formed proof tree with $L_{\C{A}}(v) = \langle \P, \S \rangle$, 
    let $v_1,\ldots,v_k = v$ be the path from the root node $v_1$ to $v$,
    and let $\PathComp_{(v_1,\ldots,v_k)} = L_{\C{C}}(v_1) \oplus ... \oplus L_{\C{C}}(v_{k-1})$.

    Let $\{\JJ_1, \ldots, \JJ_n\}$ be the set of all SCC-prefixes
    of the $\P$-depen\-den\-cy graph.
    To show that \eqref{eq:soundness-1} holds, it suffices to show
    \begin{align*}
        \iota_{\langle \P, \SSS \rangle} \sqsubseteq\;&  \iota_{\langle \P|_{\JJ_1}, \S|_{\JJ_1} \rangle} \oplus ... \oplus
        \iota_{\langle \P|_{\JJ_n}, \S|_{\JJ_n} \rangle} \oplus \PathComp_{(v_1,\ldots,v_k)}
    \end{align*}

    Let $\F{T} = (V,E,L)$ be a $\P$-CT that starts with the term $t^\sharp$ at the root.
    We will create $n$ trees $\F{T}_1, \ldots,\F{T}_n$ from $\F{T}$
    such that $\F{T}_i$ is a $\P|_{\JJ_i}$-CT that starts with $t^\sharp$ and
    \begin{align}
        &\edl_{\langle \P, \S \rangle}(\F{T})\label{dependency_graph_proof_split}\\
        \leq &\edl_{\langle \P|_{\JJ_1}, \S|_{\JJ_1} \rangle}(\F{T}_1) + \ldots + \edl_{\langle \P|_{\JJ_n}, \S|_{\JJ_n} \rangle}(\F{T}_n) + B \cdot Q + D,\nonumber
    \end{align}
    where $B,D \in \IN$ are constants and
    {\small\begin{align*}
        Q = \edl_{\langle \P|_{\JJ_1}, \S|_{\JJ_1} \rangle}(\F{T}_1) + \ldots + \edl_{\langle \P|_{\JJ_n}, \S|_{\JJ_n} \rangle}(\F{T}_n) + \edl_{\langle \P, \P \setminus \S \rangle}(\F{T}).
    \end{align*}}

    \noindent
    While ``$\edl_{\langle \P|_{\JJ_i}, \S|_{\JJ_i} \rangle}(\F{T}_i)$'' 
    considers all $(\mathbf{at})$- and $(\mathbf{af})$-steps performed with
    ADPs corresponding to DPs from $\JJ_i$,
    ``$B \cdot Q + D$'' considers all $(\mathbf{at})$- and $(\mathbf{af})$-steps that do not occur in any $\F{T}_i$ anymore.
    As \eqref{dependency_graph_proof_split}
holds for every $\P$-CT $\F{T}$, we get 
    $\iota_{\langle \P, \SSS \rangle} \sqsubseteq \iota_{\langle \P|_{\JJ_1}, \S|_{\JJ_1} \rangle} \oplus ... \oplus \iota_{\langle \P|_{\JJ_n}, \S|_{\JJ_n} \rangle} \oplus \iota_{\langle \P, \P \setminus \S \rangle}$
    and by well-formedness of $\F{P}$ we obtain $\iota_{\langle \P, \P \setminus \S
      \rangle} \sqsubseteq \PathComp_{(v_1,\ldots,v_k)}$, which implies
\eqref{eq:soundness-1}.

    We first present the construction to get from $\F{T}$ to $\F{T}_1, \ldots,\F{T}_n$,
    and afterwards, we show that the number of $(\mathbf{at})$- and $(\mathbf{af})$-steps performed in $\F{T}$
    that are not performed in any $\F{T}_1, \ldots,\F{T}_n$ is $\leq B \cdot Q + D$ for some constants $B,D \in \IN$, 
    which implies \eqref{dependency_graph_proof_split}

\smallskip
    
    \noindent
    \textbf{\underline{1. Every $\P$-CT gives rise to $n$ CTs $\F{T}_1, \ldots, \F{T}_n$ }}
    
    \noindent
    Let $1 \leq i \leq n$.
    We will construct the $\P|_{\JJ_i}$-CT $\F{T}_i = (V,E,L_i)$ using 
    the same underlying tree structure and an adjusted labeling such that $p_x^{\F{T}} =
    p_x^{\F{T'}}$ for all $x \in V$. Moreover, 
   the term at the root of $\F{T}_i$ will still be labeled with $(1:t^\sharp)$.
    
    We now recursively define the new labeling $L'$ for the $\P|_{\JJ_i}$-CT $\F{T}_i$. 
    Let $X \subseteq V$ be the set of nodes where we have already defined the labeling $L'$. 
    During our construction, we ensure that the following property holds for all $x \in X$
    \begin{align}
        \label{Cond-1} \flat(t_x) = \flat(t'_x) \land \pos_{\D^{\sharp}}(t_x) \setminus \text{Junk}(t_x, \JJ_i) \subseteq \pos_{\D^{\sharp}}(t'_x).
    \end{align}
    Here, for any annotated term $t_x$, let $\text{Junk}(t_x, \JJ_i)$ denote the set of all positions of annotations in $t_x$ 
    that will never be used for a rewrite step in $\F{T}$ with some ADP that corresponds to the DPs from $\JJ_i$.
    We define $\text{Junk}(t_x, \JJ_i)$ recursively:  
    For the term $t$ at the root, we define $\text{Junk}(t, \JJ_i) = \emptyset$.  
    For a node $y_j$ for some $1 \leq j \leq k$ with predecessor $x$ such that  
    $ t_x \tored{i}{}{\P} \{ \tfrac{p_{y_1}}{p_x} : t_{y_1}, \dots, \tfrac{p_{y_k}}{p_x} : t_{y_k} \}$
    at position $\pi$, we define
    $\text{Junk}(t_{y_j}, \JJ_i) = \{ \rho \mid \rho \in \text{Junk}(t_x, \JJ_i), \pi \nless \rho \}$
    if $\pi \notin \pos_{\D^{\sharp}}(t_x)$, and otherwise we define  
    $\text{Junk}(t_{y_j}, \JJ_i)$ to be the union of $\{ \rho \mid \rho \in \text{Junk}(t_x, \JJ_i), \pi \nless \rho \}$
    (all positions that were already in $\text{Junk}(t_x, \JJ_i)$ and are not below $\pi$), 
    and $\{ \pi.\rho \mid \rho \in \pos_{\SignatureA}(\alpha), \rho \notin \pos_{\SignatureA}(\alpha|_{\JJ_i}) \}$
    (all annotated positions where the annotation was removed by the dependency graph
    processor within the ADP). Here, as usual, $\pi < \rho$ means that $\pi$ is strictly above
    $\rho$ (i.e., $\pi$ is a proper prefix of $\rho$).
    
    We start with the same term $t$ at the root. 
    Here, our property (\ref{Cond-1}) is clearly satisfied. 
    As long as there is still an inner node $x \in X$ such that its successors are not contained in $X$, we do the following:  
    
    Let $xE = \{y_1, \dots, y_k\}$ be the set of its successors. 
    We need to define the terms for the nodes $y_1, \dots, y_k$ in $\F{T}_i$.
    Since $x$ is not a leaf and $\F{T}$ is a $\P$-CT, we have  
    $ t_x \tored{i}{}{\P} \{ \tfrac{p_{y_1}}{p_x} : t_{y_1}, \dots, \tfrac{p_{y_k}}{p_x} : t_{y_k} \}$.
    
    If we performed a step with $\tored{i}{}{\P} $ using the ADP $\alpha$,
    the position $\pi$ and the substitution $\sigma$ in $\F{T}$, then we can use the ADP 
    $\alpha|_{\JJ_i}$ with the same position $\pi$ and the same substitution $\sigma$. 
    
    Now, we directly obtain \eqref{Cond-1} for all $t_{y_j}$ with $1 \leq j \leq k$, 
    since the original rule contains the same terms with more annotations, 
    but all missing annotations are in $\text{Junk}(t_x, \JJ_i)$.

\smallskip \pagebreak[3]
    
    \noindent
    \underline{\textbf{2. Removed $(\mathbf{at})$- and $(\mathbf{af})$-steps are bounded by $B \cdot Q + D$.}}

    \noindent
    During the above construction, it can happen that
    $(\mathbf{at})$- or $(\mathbf{af})$-steps from $\F{T}$
    are replaced by only $(\mathbf{nt})$- and $(\mathbf{nf})$-steps in all $\F{T}_1, \ldots,\linebreak \F{T}_n$.
    The number of those rewrite steps is bounded by $B \cdot Q + D$.
    To be precise, we infer a bound the number of rewrite steps taking place at a term $t_x$ in node $x \in X$
    at a position $\pi \in \bigcup_{1 \leq i \leq n} \text{Junk}(t_x, \JJ_i)$ with an ADP from $\S$.
    Note that those rewrite steps belong to DPs that do not reach any SCC within the dependency graph, e.g., leaves.
    Moreover, we can bound the number of such rewrite steps, similar to the bound used for the knowledge propagation processor, 
    by its (not necessarily direct) predecessors from an SCC-prefix or the ADP
    initially used at the root.

    Let $\alpha \in \SSS$ be an ADP such that all DPs in $\nonprobDP^{\bot}(\alpha)$ cannot reach any SCC within the dependency graph.
    We will obtain a bound on $\edl_{\langle \P, \{ \alpha \} \rangle}(\F{T})$ using the complexity of the ``predecessor SCC-prefixes'',
    i.e., those SCC-prefixes that can lead to an $\alpha$-step at an annotated position later in the chain tree.

    Let $d$ be the maximal number of annotated symbols in any term on a right-hand side of an ADP from $\PP$.
    Every term in the multi-distribution of some right-hand side of some ADP $\beta$ with a DP in 
    $\nonprobDP^{\bot}(\beta)$ that is a direct predecessor
    of some DP in $\nonprobDP^{\bot}(\alpha)$
    can trigger at most $d$ $\alpha$-steps.
    Recall that even if $\alpha$ contains no annotations,
    we still have $\nonprobDP^{\bot}(\alpha) \neq \emptyset$.
    Moreover, every ADP with a DP in 
    $\nonprobDP^{\bot}(\beta)$ that is a $2$-step predecessor
    of some DP in $\nonprobDP^{\bot}(\alpha)$ 
    (i.e., we can reach a DP in $\nonprobDP^{\bot}(\alpha)$ from a DP in $\nonprobDP^{\bot}(\beta)$ in at most $2$ steps)
    can trigger at most $d^2$ $\alpha$-steps.
    In general, every $e$-step predecessor can trigger at most $d^e$ $\alpha$-steps.
    A path in the dependency graph starting at a DP from some SCC-prefix
    that reaches a DP from $\nonprobDP^{\bot}(\alpha)$
    without node repetition has a length of at most $|\nonprobDP(\P)|$.
    Moreover, we can also create $\alpha$-steps without going through an SCC-prefix,
    by following a path from the initially used ADP to $\alpha$.
    Overall, we get
    \begin{align}
        \edl_{\langle \P, \{ \alpha \} \rangle}(\F{T}) \leq
        d^{|\nonprobDP(\P)|} \cdot (1 + \sum_{i=1}^{n} \edl_{\langle \P|_{\JJ_{i}}, \P|_{\JJ_{i}} \rangle}(\F{T}_{i})).\label{eq:dep-graph-abschaetzung}
    \end{align}
    Since SCC-prefixes may contain nodes from $\nonprobDP(\S)$ and $\nonprobDP(\P \setminus \S)$
    we have 
    \begin{align}
            & \edl_{\langle \P|_{\JJ_{i}}, \P|_{\JJ_{i}} \rangle}(\F{T}_{i}) \nonumber \\
            =\;&
            \edl_{\langle \P|_{\JJ_{i}}, \S|_{\JJ_{i}} \rangle}(\F{T}_{i}) + \edl_{\langle
              \P|_{\JJ_{i}}, \P|_{\JJ_{i}} \setminus \S|_{\JJ_{i}} \rangle}(\F{T}_{i}) \nonumber \\
            \leq \; &
            \edl_{\langle \P|_{\JJ_{i}}, \S|_{\JJ_{i}} \rangle}(\F{T}_{i}) + \edl_{\langle
              \P, \P \setminus \S \rangle}(\F{T}). \label{DepGraphClaim}
                    \end{align}
    Hence, we even have 
    {\small\[\textstyle
        \begin{array}{r@{\;}l}
            &\edl_{\langle \P, \{ \alpha \} \rangle}(\F{T}) \\
         &\leq
            d^{|\nonprobDP(\P)|} \cdot (1 + \sum_{j=1}^{n} \edl_{\langle \P|_{\JJ_{i}},
              \P|_{\JJ_{i}} \rangle}(\F{T}_{i}))
\quad \text{by \eqref{eq:dep-graph-abschaetzung}}
            \\
            &\leq 
            d^{|\nonprobDP(\P)|} \cdot (1 + \sum_{j=1}^{n} \big( \edl_{\langle
              \P|_{\JJ_{i}}, \S|_{\JJ_{i}} \rangle}(\F{T}_{i}) + \edl_{\langle \P, \P
              \setminus \S \rangle}(\F{T}) \big))
\quad \text{by \eqref{DepGraphClaim}}
            \\
            &= 
            d^{|\nonprobDP(\P)|} \cdot (1 + n \cdot \edl_{\langle \P, \P \setminus \S \rangle}(\F{T}) + \sum_{j=1}^{n} \edl_{\langle \P|_{\JJ_{i}}, \S|_{\JJ_{i}} \rangle}(\F{T}_{i}))\\
            &\leq 
            d^{|\nonprobDP(\P)|} + d^{|\nonprobDP(\P)|} \cdot n \cdot (\edl_{\langle \P, \P \setminus \S \rangle}(\F{T}) + \sum_{j=1}^{n} \edl_{\langle \P|_{\JJ_{i}}, \S|_{\JJ_{i}} \rangle}(\F{T}_{i})).\!
        \end{array}
    \]}

    \noindent
    So the number of $(\mathbf{at})$- and $(\mathbf{af})$-steps with such an ADP $\alpha$ is $\leq d^{|\nonprobDP(\P)|} \cdot n \cdot Q + d^{|\nonprobDP(\P)|}$.
    Since there are at most $|\P|$ such ADPs $\alpha$,
    the final constants $B$ and $D$ for \eqref{dependency_graph_proof_split} are 
    $B = |\P| \cdot d^{|\nonprobDP(\P)|} \cdot n$ and $D = |\P| \cdot d^{|\nonprobDP(\P)|}$.

    \noindent
    Now we show that \eqref{eq:soundness-2} holds.
   For all $1 \leq i \leq n$ we have to show
    \begin{align*}
        \iota_{\langle \P|_{\JJ_i}, \P|_{\JJ_i} \setminus \SSS|_{\JJ_i} \rangle}
        \sqsubseteq\;&
    c_{(v_1,\ldots,v_k)}.
    \end{align*}
     Every $\P|_{\JJ_i}$-CT $\F{T}$ gives rise to a $\P$-CT $\F{T}'$ by using the same ADPs just with (possibly) more annotations.
    Hence, for every such $\F{T}$ we get $\edl_{\langle \P|_{\JJ_i}, \P|_{\JJ_i} \setminus \SSS|_{\JJ_i} \rangle}(\F{T}) \leq 
    \edl_{\langle \P, \P \setminus \SSS \rangle}(\F{T}')$,
    and thus, we have\linebreak 
    $\iota_{\langle \P|_{\JJ_i}, \P|_{\JJ_i} \setminus \SSS|_{\JJ_i} \rangle} 
    \sqsubseteq \iota_{\langle \P, \P \setminus \SSS \rangle}$.
    So \eqref{eq:soundness-2} holds by well-formedness of $\F{P}$, 
    which implies
    $\iota_{\langle \P, \P \setminus \SSS \rangle} \sqsubseteq
c_{(v_1,\ldots,v_k)}$.
\end{proof}

\ReductionPair*

\begin{proof}
    Let $\I_0^\sharp$ behave exactly as $\ISum^\sharp$ but in addition, it 
    maps every variable to $0$.
    Moreover, let $\text{poloC} : \IN \to \IN \cup \{\omega\}$ be the
    function that maps $n$ to the maximal interpretation of any basic term of size $\leq n$
    which is annotated at the root.
    Thus, $\text{poloC}(n) = \sup\{\I_0^\sharp(t^\sharp) \mid t \in \TB
    \text{ and } |t| \leq n\}$.

    As in \cite{FLOPS2024}, the conditions (1), (2), (3) from \cref{RP-proc} can be lifted 
    to rewrite steps with $\tored{}{}{\P}$ instead of just rules and, therefore, to edges of a CT.
    For each ADP $\alpha = \ell \to \{p_1:r_j, \ldots, p_k:r_k\}^m \in \P_{>}$ 
    we can find an $\varepsilon_{\alpha} > 0$ such that
    $\I(\ell^\sharp) 
    \geq \varepsilon_\alpha + \sum_{1\leq j \leq k} p_j \cdot \ISum^\sharp(r_j)$.
    Let $\varepsilon = \min\{\varepsilon_{\alpha} \mid \alpha \in \P_{>}\}$ be the smallest such number.
    After the lifting, we get:

    \begin{enumerate}
        \item For every $s \tored{}{}{\P} \mu$ with $(\mathbf{nt})$ or $(\mathbf{nf})$, we have 
        \[ \textstyle
        \ISum^\sharp(s) \geq \sum_{1\leq j \leq k} p_j \cdot \ISum^\sharp(t)
        \]
        \item For every $s \tored{}{}{\P} \mu$ with $(\mathbf{at})$ or $(\mathbf{af})$ and an ADP from $\PP_{\geq}$, we have 
        \[ \textstyle
        \ISum^\sharp(s) \geq \sum_{1\leq j \leq k} p_j \cdot \ISum^\sharp(t)
        \]
        \item For every $s \tored{}{}{\P} \mu$ with $(\mathbf{at})$ or $(\mathbf{af})$ and an ADP from $\PP_{>}$, we have 
        \[ \textstyle
        \ISum^\sharp(s) \geq \varepsilon + \sum_{1\leq j \leq k} p_j \cdot \ISum^\sharp(t) 
        \]
    \end{enumerate}

    We only have to prove that the complexity of $\P$ when counting only $\P_>$-rules is bounded by 
    $\iota_{\text{poloC}}$, i.e, $\iota_{\langle \P, \P_>\rangle} \sqsubseteq \iota_{\text{poloC}}$.
    Then, $\iota_{\text{poloC}} \sqsubseteq c$ follows by the same reasoning that was used in \Cref{subsec:poly-interpret} to find upper bounds
    on the polynomial interpretation of terms.
    Instead of basic terms, here we only need to consider terms that have an annotation at the root, 
    and constructor terms as proper subterms.

    Afterwards, we can conclude that \eqref{eq:soundness-1} holds, i.e.,
    \begin{align*}
        \iota_{\langle \P, \SSS \rangle} \sqsubseteq\;&  c \oplus \iota_{\langle \P, \S \setminus \P_> \rangle},
    \end{align*}
    because
    \begin{align*}
        &\iota_{\langle \P, \SSS \rangle}\\
        =\;& \iota_{\langle \P, \SSS \cap \P_> \rangle}
        \oplus \iota_{\langle \P, \SSS \setminus \P_> \rangle}\\
        \sqsubseteq\;&
        \iota_{\langle \P,  \P_> \rangle}
        \oplus \iota_{\langle \P, \SSS \setminus \P_> \rangle}\\
        \sqsubseteq\; &
        c \oplus \iota_{\langle \P,\SSS \setminus \P_> \rangle}\!
    \end{align*}

    \smallskip

    Similarly, one can also conclude that \eqref{eq:soundness-2} holds.
    Let $\F{P}$ be a well-formed proof tree with $L_{\C{A}}(v) = \langle \P, \S \rangle$,
    let $v_1,\ldots,v_k = v$ be the path from the root node $v_1$ to $v$,
    and let $\PathComp_{(v_1,\ldots,v_k)} = L_{\C{C}}(v_1) \oplus ... \oplus
    L_{\C{C}}(v_{k-1})$.
    Then we have
    \begin{align*}
        \iota_{\langle \P, \P \setminus (\SSS \setminus \P_>) \rangle} \sqsubseteq\;& \PathComp_{(v_1,\ldots,v_k)} \oplus c
    \end{align*}
    because $\iota_{\langle \P, \P \setminus (\SSS \setminus \P_>) \rangle} \sqsubseteq \iota_{\langle \P, (\P \setminus \SSS) \cup \P_> \rangle}$ and 
    \begin{align*}
        \iota_{\langle \P, (\P \setminus \SSS) \cup \P_> \rangle} = \iota_{\langle \P, \P \setminus \SSS\rangle} \oplus \iota_{\langle \P, \P_> \rangle} \sqsubseteq \PathComp_{(v_1,\ldots,v_k)} \oplus c\!
    \end{align*}

    \smallskip
    
    So it remains to show that
    $\iota_{\langle \P, \P_>\rangle} \sqsubseteq \iota_{\text{poloC}}$ holds. Consider a
    basic term $t$ and a
    $\PP$-CT $\F{T}$ with $t^\sharp$ at the root. We want to show that
    up to a factor, $\I_0^\sharp(t^\sharp)$ is a bound
on the expected derivation length 
of $\F{T}$.

    We have $\I_0^\sharp(t^\sharp)$ as the initial value at the root.
    Whenever we perform a rewrite step with $(\mathbf{nt})$ or $(\mathbf{nf})$ or with an ADP from $\PP_{\geq}$, 
    we weakly decrease the value in expectation.
    And whenever we perform a rewrite step with $(\mathbf{at})$ or $(\mathbf{af})$ and an ADP from $\PP_{>}$ at a node $v$, 
    we strictly decrease the value by at least $\varepsilon \cdot p_v$ in expectation.
    This behavior is illustrated in \Cref{fig:rpp-proof}.
    Since the value is bounded from below by $0$, we get:
       \begin{align}
         &\I_0^\sharp(t^\sharp)
         - \sum_{v \in V \setminus \ctleaf^{\F{T}}} \;
           \sum_{\substack{\P(v) \in \P_> \times \{(\mathbf{at}),(\mathbf{af})\}}} \epsilon
           \cdot p_v \geq 0 \nonumber\\
                    \Longleftrightarrow\;& \sum_{v \in V \setminus \ctleaf^{\F{T}}} \; \sum_{\substack{\P(v) \in \P_> \times \{(\mathbf{at}),(\mathbf{af})\}}} \epsilon \cdot p_v \leq \I_0^\sharp(t^\sharp) \label{cond-root-expected-value}
     \end{align}
    Therefore, we get:
     \begin{align*}
         &\epsilon \cdot \edl_{\langle \P, \P_> \rangle}(\F{T})\\
      =\;&  \sum_{v \in V \setminus \ctleaf^{\F{T}}}
 \;  \sum_{\substack{\P(v) \in \P_> \times \{(\mathbf{at}),(\mathbf{af})\}}}^{} \, \epsilon
 \cdot p_v \tag{\small definition of edl}\\ 
 \leq\;& \I_0^\sharp(t^\sharp)\tag{with \eqref{cond-root-expected-value}}  
     \end{align*}
    and therefore, $\edl_{\langle \P, \P_> \rangle}(\F{T}) \leq \nicefrac{1}{\varepsilon} \cdot \I_0^\sharp(t^\sharp)$ 
    and thus, $\iota_{\langle \P, \P_> \rangle} \sqsubseteq \iota_{\text{poloC}}$.
\end{proof}

\begin{figure}[t]
    \begin{center}
        \begin{small}
            \begin{tikzpicture}
                \tikzstyle{nodeStyle}=[thick, draw=black, fill=white, minimum width=12mm, minimum height=5mm, shape=rectangle split, rectangle split parts=2, rectangle split horizontal, inner sep=2pt]
                \tikzstyle{edgeStyle}=[->, thick]
                \tikzstyle{blueLabel}=[above, blue, font=\tiny]
                \tikzstyle{empty}=[rectangle,thick,minimum size=4mm]
                \tikzstyle{evalue}=[rectangle,thick,minimum size=4mm, font=\small]
                
                \node[nodeStyle, label=above:{$(\P_>, (\mathbf{at}))$}] (rootR) at (6, 3) {1 \nodepart{two} $\I_0^\sharp(t^\sharp)$};
                
                \node[nodeStyle, label=above left:{$(\P_>, (\mathbf{at}))$}] (n1R) at (5.2, 1.5) {$p_1$ \nodepart{two} $\I_0^\sharp(t_1)$};
                \node[nodeStyle, label=above right:{$(\P_\geq, (\mathbf{at}))$}] (n2R) at (6.8, 1.5) {$p_2$ \nodepart{two} $\I_0^\sharp(t_2)$};
                
                \node[empty] (n3R) at (4.6, 0) {$\ldots$};
                \node[empty] (n4R) at (5.8, 0) {$\ldots$};
                \node[empty] (n5R) at (6.2, 0) {$\ldots$};
                \node[empty] (n6R) at (7.4, 0) {$\ldots$};
                
                \draw[edgeStyle] (rootR) -- (n1R);
                \draw[edgeStyle] (rootR) -- (n2R);
                \draw[edgeStyle] (n1R) -- (n3R);
                \draw[edgeStyle] (n1R) -- (n4R);
                \draw[edgeStyle] (n2R) -- (n5R);
                \draw[edgeStyle] (n2R) -- (n6R);

                \node (A) at (10, 2.75) {};
                \node (B) at (10, 1.85) {};
                
                \node (C) at (10, 1) {};
                \node (D) at (10, 0) {};
                
                \draw[->, bend left] (A) to (B);
                \draw[->, bend left] (C) to (D);
                
                \node[evalue] (ev1) at (10, 3) {$\IE(\I_0^\sharp(t^\sharp)) = \I_0^\sharp(t^\sharp)$};
                \node[empty]  at (11, 2.325) {$> \epsilon \cdot 1$};
                \node[evalue] (ev2) at (10, 1.7) {$\IE(\I_0^\sharp (\{ p_1 : t_1, p_2 : t_2\}))$};
                \node[evalue] (ev21) at (10, 1.2) {$= p_1 \cdot \I_0^\sharp(t_1) + p_2 \cdot \I_0^\sharp(t_2)$};
                \node[empty]  at (11, 0.5) {$> \epsilon \cdot p_1$};
                \node[empty] (ev3) at (10, 0) {$\ldots$};
                
            \end{tikzpicture}
        \end{small}
    \end{center}
    \vspace*{-2mm}
    \caption{Expected decrease of $\I_0^\sharp$ in a chain tree}\label{fig:rpp-proof}
\end{figure}

\KnowledgeProp*

\begin{proof}
    Let $\F{P}$ be a well-formed proof tree with $L_{\C{A}}(v) = \langle \P, \S \rangle$, 
    let $v_1,\ldots,v_k = v$ be the path from the root node $v_1$ to $v$,
    and let $\PathComp_{(v_1,\ldots,v_k)} = L_{\C{C}}(v_1) \oplus ... \oplus L_{\C{C}}(v_{k-1})$.
    
    We first show that \eqref{eq:soundness-1} holds, i.e.,
    \begin{align*}
        \iota_{\langle \P, \SSS \rangle} \sqsubseteq\;& \PathComp_{(v_1,\ldots,v_k)} \oplus \Pol_0 \oplus \iota_{\langle \P, \S \setminus \{\alpha\} \rangle}
    \end{align*}
    Let $\F{T} = (V,E,L)$ be a $\P$-CT.
    Let $d$ be the maximal number of annotated symbols in any term on a right-hand side of an ADP from $\PP$.
    Recall that 
    \[\textstyle
        \begin{array}{rcl}
        \edl_{\langle \P, \{ \alpha \} \rangle}(\F{T}) &\leq& 1 +
        \sum_{v \in V \setminus \ctleaf^{\F{T}}, \;
        \P(v) \in \Pre(\alpha) \times  \{\bat, \baf\}} \, d \cdot p_v\\
        &=& 1 + d
        \cdot   \edl_{\langle \P, \Pre(\alpha) \rangle}(\F{T}),
        \end{array}
    \]
    which implies 
    $\eirc_{\langle \P,\{ \alpha \} \rangle}(n) \leq 1 + d
    \cdot \eirc_{\langle
    \P,\Pre(\alpha) \rangle}(n)$
    for all $n \in \IN$ and thus,
    $\iota_{\langle \P,\{ \alpha \} \rangle}  \sqsubseteq  \iota_{\langle \P,\Pre(\alpha)
      \rangle}$. Hence,
 \eqref{eq:soundness-1} holds, because of well-formedness of $\F{P}$, i.e.,
    \begin{align*}
        \iota_{\langle \P, \SSS \rangle} =\;&\iota_{\langle \P, \SSS \setminus \{\alpha\} \rangle} \oplus \iota_{\langle \P, \{\alpha\} \rangle} \sqsubseteq \iota_{\langle \P, \SSS \setminus \{\alpha\} \rangle} \oplus \iota_{\langle \P,\Pre(\alpha) \rangle}\\
        \sqsubseteq\;&\iota_{\langle \P, \SSS \setminus \{\alpha\} \rangle} \oplus \iota_{\langle \P,\P \setminus \SSS \rangle} \sqsubseteq \iota_{\langle \P, \SSS \setminus \{\alpha\} \rangle} \oplus \PathComp_{(v_1,\ldots,v_k)}\!
    \end{align*}
Now we show \eqref{eq:soundness-2}, i.e.,
    \begin{align*}
        \iota_{\langle \P, \P \setminus (\SSS \setminus \{\alpha\}) \rangle} \sqsubseteq\;& \PathComp_{(v_1,\ldots,v_k)} \oplus \Pol_0
    \end{align*}
    This holds again because $\F{P}$ is well formed:
    \begin{align*}
        \iota_{\langle \P, \P \setminus (\SSS \setminus \{\alpha \}) \rangle} =\;&\iota_{\langle \P, (\P \setminus \SSS) \cup \{\alpha\} \rangle} = \iota_{\langle \P, \P \setminus \SSS \rangle} \oplus \iota_{\langle \P, \{\alpha\} \rangle}\\
        \sqsubseteq\;& \PathComp_{(v_1,\ldots,v_k)}\oplus \PathComp_{(v_1,\ldots,v_k)} = \PathComp_{(v_1,\ldots,v_k)}\!
    \end{align*}
\end{proof}

Next, we prove the soundness of the probability removal processor.
Here, we follow the notation of \cite{noschinski2013analyzing} 
for the non-probabilistic dependency tuple framework.

In \cite{noschinski2013analyzing}, the runtime complexity for a DT problem $(\P,\SSS,\K,\R)$
is defined via \emph{$\langle \P, \R\rangle$-DT chain trees}.
Nodes of these trees are labeled by pairs $(\ell^\sharp \to [t^\sharp_1, \ldots, t^\sharp_n], \sigma)$
of a DT and a substitution such that $\ell^\sharp \sigma \in \NF_{\R}$,
and if a node $(\ell^\sharp \to [t^\sharp_1, \ldots, t^\sharp_n], \sigma)$ has children
$(\ell_1^\sharp \to \ldots, \delta_1), \ldots, (\ell_k^\sharp \to \ldots, \delta_k)$,
then there are pairwise different $i_1, \ldots, i_k \in \{1, \ldots, n\}$ 
such that $t^\sharp_{i_j} \sigma \ito_{\R}^* \ell_{j}^\sharp \delta_{j}$ for all $1 \leq j \leq k$.\footnote{So the branching in $\langle \P, \R\rangle$-DT chain trees is not due to probabilities but due to the right-hand sides of DTs containing several terms that can be evaluated.}
The $\S$-\emph{derivation length} $\idl_{\langle \P, \S, \R \rangle}(\F{T})$ of a $\langle \P, \R\rangle$-DT chain tree $\F{T}$ is 
the number of its nodes labeled by DTs from $\S$.
The \emph{derivation height} of a term $t^\sharp$ w.r.t.\ $\langle \P, \S, \R \rangle$
is the supremum over all $\S$-derivation lengths of all $\langle \P, \R\rangle$-DT chain trees starting with $t^\sharp$, i.e.,
{\small\[
  \begin{array}{l}
    \idh_{\langle \P, \S, \R \rangle}(t^\sharp) =\\
    \qquad \sup \{ \idl_{\langle \P, \S, \R\rangle}(\F{T}) \mid 
    \F{T} \text{ is a $\langle \P, \R\rangle$-DT chain tree starting with } t^\sharp \}.
    \end{array}
\]}
Finally, the runtime complexity function of $\langle \P, \S, \R \rangle$ is defined in a
similar way as before
\[
    \irc_{\langle \P, \S, \R \rangle}(n) = \sup \{ \idh_{\langle \P, \S, \R \rangle}(t^\sharp) \mid t \in \TB_\R, |t| \leq n \},
\]
and its \emph{runtime complexity} $\iota_{\langle \P, \S, \R \rangle}$ is $\iota(\irc_{\langle \P, \S, \R \rangle})$.

As mentioned before, in the non-probabilistic DT framework one has two components $\S$ and $\K$ for 
those DTs that we still need to count and for those DTs for which we
already have a complexity bound on how often
this DT can occur in a $\langle \P, \R \rangle$-DT chain tree, respectively.
The overall complexity of a DT problem $(\P,\SSS,\K,\R)$ is
$\iota_{(\P,\SSS,\K,\R)} = \iota_{\langle \P, \S, \R \rangle} \ominus \iota_{\langle \P,
  \K, \R \rangle}$.\footnote{We do not need such a $\ominus$-operation due to our notion of
well-formed proof trees that did not exist in \cite{noschinski2013analyzing}.}
Here, $c \ominus d = c$ if $d \sqsubset c$ and $c \ominus d = \Pol_0$ otherwise 
(so, e.g., $\Pol_2 \ominus \Pol_1 = \Pol_2$ and $\Pol_1 \ominus \Pol_2 = \Pol_0$)

\ProbabilityRemovalProc*

\begin{proof}
    First, every $\P$-CT $\F{T}$ gives rise to a 
    $\langle \nonprobDT(\P), \nonprob(\P) \rangle$-DT chain tree $\F{T}'$
    such that $\edh_{\langle \P, X \rangle}(\F{T}) = \idh_{\langle \nonprobDT(\P),
      \nonprobDT(X), \nonprob(\P)
      \rangle}(\F{T}')$
    for every $X \subseteq \P$, and vice versa.
    To see this, note that every $\PP$-CT is a single (not necessarily finite) path, due to the trivial probabilities in its rules.
    Moreover, the DTs $\alpha$ in $\F{T}'$ correspond to those ADPs $\beta$ that are used
    at
    annotated subterms, i.e.,
    with Case $(\mathbf{at})$ or $(\mathbf{af})$, and we have $\alpha = \nonprobDT(\beta)$.
    The substitution is the corresponding substitution used for the rewrite step at this node.

    To be precise, if we have $t \tored{}{}{\PP} \{1:s\}$ at the root of $\F{T}$,
    where we use an ADP $\alpha = \ell \to \{1:r\}^m \in \P$ at position $\pi$
    with the substitution $\sigma$ such that $\flat(t|_\pi) = \ell\sigma \in \ANF_{\P}$,
    then we start with $(\nonprobDT(\alpha), \sigma)$ at the root of $\F{T}'$.
    Let $\nonprobDT(\alpha) = \ell^\sharp \to [r^\sharp_1, \ldots, r^\sharp_n]$.
    If we eventually rewrite  at the subterm corresponding to some $r^\sharp_j$ in $\F{T}$
    with an ADP $\alpha' = \ell' \to \{1:r'\}^{m'} \in \P$ and the substitution $\sigma'$,
    and this subterm still contains its annotation at the root, 
    then we must have $r^\sharp_j \sigma \ito_{\nonprob(\P)} \ell'^\sharp \sigma'$.
    Hence, we can create a child node of $(\nonprobDT(\alpha), \sigma)$ 
    labeled by $(\nonprobDT(\alpha'), \sigma')$.
    We can construct the whole tree $\F{T}'$ inductively.
    Note that whenever we perform a rewrite step at an annotated position
in $\F{T}$,
    we create a corresponding node in $\F{T}'$, hence we have 
    $\edh_{\langle \P, X \rangle}(\F{T}) = \idh_{\langle \nonprobDT(\P), \nonprobDT(X), \nonprob(\P) \rangle}(\F{T}')$.

    For the converse, i.e., to get from a $\langle \nonprobDT(\P), \nonprob(\P) \rangle$-DT chain tree $\F{T}'$
    to the corresponding $\P$-CT $\F{T}$, we simply perform all rewrite steps with $\R$ in $\F{T}$ 
    that are omitted in $\F{T}'$.
    To be precise, if we have $(\nonprobDT(\alpha), \sigma)$ at the root of $\F{T}'$
    for some ADP $\alpha = \ell \to \{1:r\}^m \in \P$,
    then we can perform the rewrite step $\ell^\sharp \sigma \tored{}{}{\PP} \{1:r \sigma\}$ at the root of $\F{T}$.
    Let $\nonprobDT(\alpha) = \ell^\sharp \to [r^\sharp_1, \ldots, r^\sharp_n]$.
    If there is a child
 of the root in $\F{T}'$
    labeled by $(\nonprobDT(\alpha'), \sigma')$,
    then we have $r_j^\sharp \sigma \ito_{\nonprob(\P)} \ell'^\sharp \sigma'$ for some $1
    \leq j \leq n$,
where $\ell'$ is the left-hand side of $\alpha'$.
We can simply perform those rewrite steps in $\F{T}$.
Since $\nonprob(\P)$ only contains rules resulting from 
 ADPs with the flag $m = \ttrue$, this does not remove any annotations from 
    subterms that are not in normal form.
    Hence, we can do this iteratively for every such child.
    In this way, we can create the whole tree $\F{T}$ inductively.
    Again, for every node of the $\langle \nonprobDT(\P), \nonprob(\P) \rangle$-DT chain tree,
    we perform a rewrite step with Case $(\mathbf{at})$ or $(\mathbf{af})$ in the $\P$-CT.
    Therefore, 
    $\edh_{\langle \P, X \rangle}(\F{T}) = \idh_{\langle \nonprobDT(\P), \nonprobDT(X), \nonprob(\P)  \rangle}(\F{T}')$.

    As this relation holds for arbitrary $\P$-CTs and $\langle \nonprobDT(\P), \nonprob(\P) \rangle$-DT chain trees,
    we obtain $\iota_{\langle \P, X \rangle} = \iota_{\langle \nonprobDT(\P), \nonprobDT(X), \nonprob(\P) \rangle}$
    for every $X \subseteq \P$.

    We can now conclude soundness:
    Let $\F{P}$ be a well-formed proof tree with $L_{\C{A}}(v) = \langle \P, \S \rangle$, 
    let $v_1,\ldots,v_k = v$ be the path from the root node $v_1$ to $v$,
    and let $\PathComp_{(v_1,\ldots,v_k)} = L_{\C{C}}(v_1) \oplus ... \oplus L_{\C{C}}(v_{k-1})$.

    Since the resulting ADP problems are solved, we only have to show \eqref{eq:soundness-1}.
    Let $\K = \P \setminus \S$.
    If we have $\iota_{\langle \nonprobDT(\P), \nonprobDT(\S), \nonprob(\P) \rangle} 
    \sqsubseteq \iota_{\langle \nonprobDT(\P), \nonprobDT(\K), \nonprob(\P) \rangle}$,
    then
    $\iota_{\langle \P, \S \rangle} \sqsubseteq \iota_{\langle \P, \P \setminus \S \rangle} 
    \sqsubseteq \PathComp_{(v_1,\ldots,v_k)}$
    by well-formedness of $\F{P}$, which proves that 
    $\Proc_{\mathtt{PR}}(\langle \PP, \SSS \rangle) = (c, \emptyset)$ is sound for every $c \in \C{C}$.

    Next, consider the case $\iota_{\langle \nonprobDT(\P), \nonprobDT(\K), \nonprob(\P) \rangle}
    \sqsubset \iota_{\langle \nonprobDT(\P), \nonprobDT(\S), \nonprob(\P) \rangle}$.
    This implies that $\iota_{(\nonprobDT(\P), \nonprobDT(\S), \nonprobDT(\K), \nonprob(\P))} 
    = \iota_{\langle \nonprobDT(\P), \nonprobDT(\S), \nonprob(\P) \rangle} 
    \ominus \iota_{\langle \nonprobDT(\P), \nonprobDT(\K), \nonprob(\P) \rangle} 
    = \iota_{\langle \nonprobDT(\P), \nonprobDT(\S), \nonprob(\P) \rangle}$, and
    therefore, we have
$\iota_{\langle \P, \S \rangle} =  \iota_{\langle \nonprobDT(\P), \nonprobDT(\S),
      \nonprob(\P) \rangle}
    = 
    \iota_{\langle \nonprobDT(\P), \nonprobDT(\S), \nonprobDT(\K), \nonprob(\P) \rangle} \sqsubseteq c$,
    since $c$ is returned by the DT framework as a bound on the runtime complexity of
    $\langle \nonprobDT(\P), \nonprobDT(\S), \nonprobDT(\K), \nonprob(\P) \rangle$. 
    Again, this proves soundness of a processor with
    $\Proc_{\mathtt{PR}}(\langle \PP, \SSS \rangle) = (c, \emptyset)$.
\end{proof}

\subsection{Transformational Processors}\label{sec:trans}

Finally, to show how to adapt transformational processors \cite{arts2000termination,giesl2006mechanizing,noschinski2013analyzing} for our novel ADP framework,
we consider the \emph{narrowing processor}.

Let $\PP = \PP' \uplus \{\ell \ruleArr{}{}{} \{ p_1:r_{1}, \ldots, p_k: r_k\}^{m}\}$ be
a set of 
ADPs
and let  $t \trianglelefteq_{\sharp} r_j$ for some $1 \leq j \leq k$.
If we have to perform rewrite steps on (an instance of) 
$t$ in order to enable the next application of an ADP at an annotated position, 
then the idea of the narrowing processor is to perform the first step 
of this reduction already on the ADP $\ell \ruleArr{}{}{} \{ p_1:r_{1}, \ldots, p_k:
r_k\}^{m}$ via narrowing.
So whenever there is a $t \trianglelefteq_{\sharp} r_j$ and a non-variable position $\tau$ in $t$ 
such that $t|_\tau$ unifies with the left-hand side $\ell'$ of some (variable-renamed) ADP 
$\ell' \ruleArr{}{}{} \{p_1':r_{1}', \ldots, p_{k'}': r_{k'}'\}^{m'} \in \PP$ 
using an mgu $\delta$ such that $\ell \delta, \ell' \delta \in \ANF_{\PP}$, 
then $\delta$ is a \emph{narrowing substitution} of $t$.
This is analogous to the narrowing substitutions defined in \cite{noschinski2013analyzing} for DTs.

As shown in \cite{FLOPS-Journal},  in the probabilistic setting,
the narrowing processor can only be used 
in a weaker version. Hence, here
it was renamed to the \emph{rule overlap instantiation processor}.
While we can apply the narrowing substitution to the ADP, we cannot perform any rewrite steps.
To be precise, if $\delta_1, \ldots ,\delta_d$ are all narrowing substitutions of $t$, 
then we can replace $\ell \ruleArr{}{}{} \{ p_1:r_{1}, \ldots, p_k: r_k\}^{m}$ 
by $\ell \delta_e \to \{p_1:r_1 \delta_e, \ldots, p_k:r_k \delta_e\}$, 
for all $1 \leq e \leq d$, but we cannot perform any rewrite steps on $t\delta_e$ directly.

Moreover, there could be another subterm $t' \trianglelefteq_{\sharp} r_j$ (with $t' \neq t$) 
which was involved in a chain tree (i.e., $t'^\sharp \sigma \ito^*_{\nonprob(\PP)}
\tilde{\ell}\tilde{\sigma}$ for some substitutions $\sigma, \tilde{\sigma}$ 
and a left-hand side $\tilde{\ell}$ of an ADP), 
but this reduction
is no longer possible when replacing $t'$ by the instantiations
$t' \delta_1, \ldots, t' \delta_d$. 
We say that $t'$ is \emph{captured} by $\delta_1, \ldots, \delta_d$ 
if for each narrowing substitution $\rho$ of $t'$,
there is a $\delta_e$ with $1 \leq e \leq d$ such that $\delta_e$ is more general than $\rho$, 
i.e., $\rho = \delta_e \rho'$ for some substitution $\rho'$. 
So the narrowing processor has to add another ADP 
$\ell \ruleArr{}{}{} \{ p_1:\sharp_{\capt_1(\delta_1,\ldots,\delta_d)}(r_{1}), 
\ldots, p_k: \anno_{\capt_k(\delta_1,\ldots,\delta_d)}(r_k)\}^{m}$, 
where $\capt_j(\delta_1,\ldots,\delta_d)$ contains all positions of subterms 
$t' \trianglelefteq_{\sharp} r_j$ which are not captured by the narrowing substitutions $\delta_1, \ldots, \delta_d$ of $t$. 

\begin{restatable}[Rule Overlap Instantiation Processor]{theorem}{RuleInstantiation}\label{theorem:rule-overlap-inst}
  Let $\langle \PP, \SSS \rangle$ be an ADP problem with 
  $\PP = \PP' \uplus \{\alpha\}$
  for $\alpha = \ell \ruleArr{}{}{} \{ p_1:r_{1}, \ldots, p_k: r_k\}^{m}$, 
 let $1 \leq j \leq k$, and let  $t \trianglelefteq_{\sharp} r_j$.
  Let $\delta_1, \ldots, \delta_d$ be all narrowing substitutions of $t$.
  Then $\Proc_{\mathtt{ROI}}(\langle \PP, \SSS \rangle) = \!(\Pol_0, \{\langle \PP' \cup
  N, \widetilde{\SSS} \rangle\})$  
  is sound, where

\vspace*{-.5cm}
  
  {\small\[
    \arraycolsep=2pt
    \begin{array}{rcl}
      N & = & \phantom{\cup} \; \{ \ell \delta_e \to \{p_1:r_1 \delta_e, \ldots,
      p_k:r_k \delta_e\}^{m} \mid 1 \leq e \leq d \} \\
        && \cup  \; \{
          \ell \ruleArr{}{}{} \{ p_1:\anno_{\capt_1(\delta_1,\ldots,\delta_d)}(r_{1}), \ldots, p_k: \anno_{\capt_k(\delta_1,\ldots,\delta_d)}(r_k)\}^{m}\!
     \} \\[.1cm]
        \widetilde{\SSS} & = & \begin{cases} (\SSS \setminus \{\alpha\}) \cup N,& \text{if } \alpha \in \SSS \\ 
        \SSS,& \text{otherwise}\! \end{cases}\!
    \end{array}
  \]}
\end{restatable}

\begin{proof} 
    We first show that \eqref{eq:soundness-1} holds, i.e., we show
    \begin{align*}
        \iota_{\langle \P, \SSS \rangle} \sqsubseteq    \iota_{\langle \PP' \cup N, \widetilde{\SSS} \rangle}
    \end{align*}
    Every $\P$-CT $\F{T}$ gives rise to a $(\PP' \cup N)$-CT $\F{T}'$, 
    since every rewrite step with $\alpha$ 
    can also be done with a rule from $N$.
    Moreover, if $\alpha \in \SSS$, then $N \subseteq \widetilde{\SSS}$,
    so every rewrite step that counts for the expected derivation length of $\F{T}$ w.r.t.\ 
    $\langle \PP, \SSS \rangle$ still counts for the expected derivation length of $\F{T}'$ w.r.t.\ $\langle \PP' \cup N, \widetilde{\SSS} \rangle$.
    To be precise, we get $\edl_{\langle \PP, \SSS \rangle}(\F{T}) 
    \leq \edl_{\langle \PP' \cup N, \widetilde{\SSS} \rangle}(\F{T})$.\footnote{Note 
    that we only have ``$\leq$'' and not ``$=$'' since we may add rules to $\SSS$ that did not count before,
    but finding a tree of equal or greater expected derivation length suffices.}
    As this holds for every $\P$-CT, we obtain $\iota_{\langle \P, \SSS \rangle} \sqsubseteq \iota_{\langle \PP' \cup N, \widetilde{\SSS} \rangle}$,
    and \eqref{eq:soundness-1} holds.

    Now we explain the precise construction. 
    Let $\F{T} = (V,E,L)$ be a $\PP$-CT and let $\overline{\PP'} = \PP' \cup N$.
    We will create an $\overline{\PP'}$-CT $\F{T}' = (V,E,L')$.
    As mentioned, the core idea of this construction is that every rewrite step with $\alpha$ 
    can also be done with a rule from $N$.
    If we use $\ell \ruleArr{}{}{} \{ p_1:\anno_{\capt_1(\delta_1,\ldots,\delta_d)}(r_{1}), \ldots, p_k: \anno_{\capt_k(\delta_1,\ldots,\delta_d)}(r_k)\}^{m} \in N$, 
    we may create fewer annotations than we did when using the old ADP $\alpha$.
    However, we will never rewrite at the position of the annotations that got removed in the CT $\F{T}$, hence we can ignore them.
    We construct the new labeling $L'$ for the $\overline{\PP'}$-CT $\F{T}'$ inductively such that for all nodes 
    $x \in V \setminus \ctleaf$ with $xE = \{y_1, \ldots, y_m\}$ we have 
    $t'_x \tored{}{}{\overline{\PP'}} \{\tfrac{p_{y_1}}{p_x}:t'_{y_1}, \ldots, \tfrac{p_{y_m}}{p_x}:t'_{y_m}\}$.
    Let $X \subseteq V$ be the set of nodes $x$ where we have already defined the labeling $L'(x)$.
    During our construction, we ensure that the following property holds for all $x \in X$:
    \begin{equation} \label{soundness-roi-induction-hypothesis}
        \flat(t_x) = \flat(t'_x) \land \posT(t_x) \cap \PosDPoss(\flat(t_x)) \subseteq \posT(t'_x).
    \end{equation}
   Thus, all annotations of root symbols of subterms that are not in normal form in $t_x$ are still annotated in $t'_x$.

    For the construction, we start with the same term at the root.
    Here, \eqref{soundness-roi-induction-hypothesis} obviously holds.
    As long as there is still an inner node $x \in X$ such that its successors are not contained in $X$, we do the following.
    Let $xE = \{y_1, \ldots, y_m\}$ be the set of its successors.
    We need to define the corresponding terms $t_{y_1}', \ldots, t_{y_m}'$ for the nodes $y_1, \ldots, y_m$.
    Since $x$ is not a leaf and $\F{T}$ is a $\PP$-CT, we have $t_x \tored{}{}{\PP} \{\tfrac{p_{y_1}}{p_x}:t_{y_1}, \ldots, \tfrac{p_{y_m}}{p_x}:t_{y_m}\}$.
    We have the following three cases:
    \begin{enumerate}
        \item[(A)] If it is a step with $\tored{}{}{\PP}$ using an ADP that is different from $\alpha$ in $\F{T}$, 
        then we perform a rewrite step with the same ADP, the same redex, and the same substitution in $\F{T}'$.
        Then, it is easy to see that \eqref{soundness-roi-induction-hypothesis}  holds for the resulting terms.
        \item[(B)] If it is a step with $\tored{}{}{\PP}$ using $\alpha$ at a position $\pi \notin \posT(t_x)$ in $\F{T}$, 
        then we perform a rewrite step with the new ADP 
        $\ell \ruleArr{}{}{} \{ p_1:\anno_{\capt_1(\delta_1,\ldots,\delta_d)}(r_{1}), \ldots, p_k: \anno_{\capt_k(\delta_1,\ldots,\delta_d)}(r_k)\}^{m}$, 
        the same redex, same substitution, and same position in $\F{T}'$.
        Since the new rule has the same underlying terms as $\alpha$, it is easy to see
        that \eqref{soundness-roi-induction-hypothesis} holds for the resulting terms.
        Note that the rule that we use contains fewer annotations than the original rule, 
        but since $\pi \notin \posT(t_x)$, we remove all annotations from the rule during the application of the rewrite step anyway.
        \item[(C)] If it is a step with $\tored{}{}{\PP}$ using $\alpha$ at a position $\pi \in \posT(t_x)$ in $\F{T}$, 
        then we look at the specific successors to find a substitution $\delta$ such that 
        $\ell \delta \ruleArr{}{}{} \{ p_1:r_{1}\delta, \ldots, p_k: r_k\delta\}^{m} \in N$ or we detect that we can use the ADP 
        $\ell \ruleArr{}{}{} \{ p_1:\anno_{\capt_1(\delta_1,\ldots,\delta_d)}(r_{1}), \ldots,\linebreak p_k: \anno_{\capt_k(\delta_1,\ldots,\delta_d)}(r_k)\}^{m}$ 
        and perform a rewrite step with this new ADP, at the same position in $\F{T}'$.
    \end{enumerate}

    \noindent
   It remains to consider Case (C) in detail.
    Here, we have $t_x \tored{}{}{\PP} \{\tfrac{p_{y_1}}{p_x}:t_{y_1}, \ldots, \tfrac{p_{y_k}}{p_x}:t_{y_k}\}$ using the ADP $\alpha$, 
    the position $\pi \in \posT(t_x)$, and a substitution $\sigma$ such that $\flat(t_x|_{\pi}) = \ell \sigma \in \ANF_{\PP}$.

    We first consider the case where there is no successor $v$ of $x$ where an ADP is applied at an annotated position below or at $\pi$, 
    or an ADP is applied on a position strictly
    above $\pi$ before reaching such a node $v$.
    Then, we can use $\ell \ruleArr{}{}{} \{ p_1:\anno_{\capt_1(\delta_1,\ldots,\delta_d)}(r_{1}), \ldots, p_k: \anno_{\capt_k(\delta_1,\ldots,\delta_d)}(r_{k})\}^{m}$ 
    instead, because the annotations will never be used, so they do not matter.

    Otherwise, there exists a successor $v$ of $x$ where an ADP is applied at an annotated
    position below or at $\pi$, and no ADP is applied on a position strictly  above $\pi$ before.
    Let $v_1, \ldots, v_n$ be all (not necessarily direct) successors that rewrite below
    position $\pi$, or rewrite at position $\pi$, and on the path from $x$ to $v$ there is
    no other node with this property, and no node that performs a rewrite step strictly above $\pi$.
    Furthermore, let $t_1, \ldots, t_n$ be the used redexes and $\rho_1, \ldots, \rho_n$ be the used substitutions.
    \begin{itemize}
        \item \textbf{(C1) If} none of the redexes $t_1, \ldots, t_n$ is captured by $t$,
          then we use $\ell \ruleArr{}{}{} \{
          p_1:\anno_{\capt_1(\delta_1,\ldots,\delta_d)}(r_{1}), \ldots, p_k:
          \anno_{\capt_k(\delta_1,\ldots,\delta_d)}(r_k)\}^{m}$ with the position $\pi \in
          \posT(t_{x}) \cap \PosDPoss(\flat(t_{x})) \subseteq_{(IH)} \posT(t'_{x})$ and
          the substitution $\sigma$. 
       Again, \eqref{soundness-roi-induction-hypothesis} is satisfied for our resulting terms.

        \item \textbf{(C2) If} $t = t_i$ for some $1 \leq i \leq n$, then we can find a narrowing substitution $\delta_e$ of $t$ that is more general than $\sigma$, i.e., we have $\delta_e \gamma = \sigma$.
            Now, we use the ADP $\ell \delta_e \ruleArr{}{}{} \{ p_1:r_{1}\delta_e, \ldots, p_k: r_k\delta_e\}^{m}$ with the position $\pi \in \posT(t_{x}) \cap \PosDPoss(\flat(t_{x})) \subseteq_{(IH)} \posT(t'_{x})$ and the substitution $\gamma$ such that $\flat(t_x|_{\pi}) = \ell \delta_e \gamma = \ell \sigma \in \ANF_{\PP}$.
            Again, \eqref{soundness-roi-induction-hypothesis} is satisfied for our resulting terms.

        \item \textbf{(C3) If} $t \neq t_i$ for all $1 \leq i \leq n$ but there is an $1 \leq i \leq n$ such that $t_i$ is captured, then, since $t_i$ is captured, there exists a narrowing substitution $\delta_e$ of $t$ that is more general than $\rho_i$, i.e., there exists a substitution $\kappa_1$ with $\delta_e \kappa_1 = \rho_i$, and since we use $\rho_i$ later on we additionally have that $\rho_i$ is more general than $\sigma$, i.e., there exists a substitution $\kappa_2$ with $\rho_i \kappa_2 = \sigma$.
            Now, we use the ADP $\ell \delta_e \ruleArr{}{}{} \{ p_1:r_{1}\delta_e, \ldots, p_k: r_k\delta_e\}^{m}$ with the position $\pi \in \posT(t_{x}) \cap \PosDPoss(\flat(t_{x})) \subseteq_{(IH)} \posT(t'_{x})$ and the substitution $\kappa_1 \kappa_2$ such that $\flat(t_x|_{\pi}) = \ell \delta_e \kappa_1 \kappa_2 = \ell \sigma \in \ANF_{\PP}$.
            Again, \eqref{soundness-roi-induction-hypothesis} is satisfied for our resulting terms.
    \end{itemize}

    \smallskip
Now we show condition \eqref{eq:soundness-2}.
  Let $\F{P}$ be a well-formed proof tree with $L_{\C{A}}(v) = \langle \P, \S \rangle$, 
    let $v_1,\ldots,v_k = v$ be the path from the root node $v_1$ to $v$,
    and let $\PathComp_{(v_1,\ldots,v_k)} = L_{\C{C}}(v_1) \oplus ... \oplus L_{\C{C}}(v_{k-1})$.
    We have
    \begin{align*}
        \iota_{\langle \PP' \cup N, (\PP' \cup N) \setminus \widetilde{\SSS} \rangle} \sqsubseteq\;& \PathComp_{(v_1,\ldots,v_k)},
    \end{align*}
    because every $(\PP' \cup N)$-CT $\F{T}$ gives rise to a $\PP$-CT $\F{T}'$ with
    at least the same $(\mathbf{at})$- and $(\mathbf{af})$-steps.
    We can replace each usage of an ADP $\ell \delta_e \ruleArr{}{}{} \{ p_1:r_{1}\delta_e, \ldots, p_k: r_k\delta_e\}^{m}$
    with the more general ADP $\alpha$,
    and each ADP $\ell \ruleArr{}{}{} \{ p_1:\anno_{\capt_1(\delta_1,\ldots,\delta_d)}(r_{1}), \ldots, p_k: \anno_{\capt_k(\delta_1,\ldots,\delta_d)}(r_k)\}^{m}$ 
    can be replaced by $\alpha$ as well, leading to more annotations than before.

    If an $(\mathbf{at})$- or $(\mathbf{af})$-step is performed with an ADP $\beta \in (\PP' \cup N) \setminus \widetilde{\SSS}$ in $\F{T}$,
    we have the following cases:
    \begin{itemize}
        \item If $\beta \in \P' \setminus \widetilde{\SSS}$, and $\alpha \in \S$, then
        \begin{align*}
            \beta \in   \P' \setminus \widetilde{\SSS} = \P' \setminus ((\SSS \setminus \{\alpha\}) \cup N) = \P' \setminus (\SSS \cup N) \subseteq \P \setminus \SSS\!
        \end{align*}
        \item If $\beta \in \P' \setminus \widetilde{\SSS}$, and $\alpha \notin \S$, then
        \begin{align*}
            \beta \in \P' \setminus \widetilde{\SSS} = \P' \setminus \SSS \subseteq \P \setminus \SSS\!
        \end{align*}
        \item If $\beta \in N \setminus \widetilde{\SSS}$, then $\alpha \notin \S$, and we use $\alpha$ in $\F{T}'$.
    \end{itemize}
    In every case, if the step counts towards the expected derivation length w.r.t.\ $(\PP' \cup N) \setminus \widetilde{\SSS}$ in $\F{T}$,
    it will count for the expected derivation length w.r.t.\ $\PP \setminus \SSS$ in $\F{T}'$ as well.
    Therefore, we result in $\edl_{\langle \PP' \cup N, (\PP' \cup N) \setminus \widetilde{\SSS} \rangle}(\F{T}) \leq \edl_{\langle \PP, \P \setminus \S \rangle}(\F{T})$.
    As this holds for every $(\PP' \cup N)$-CT, we get 
    $\iota_{\langle \PP' \cup N,
(\PP' \cup N) \setminus \widetilde{\SSS} \rangle} \sqsubseteq \iota_{\langle \P, \P \setminus \SSS \rangle} \sqsubseteq \PathComp_{(v_1,\ldots,v_k)}$,
    by well-formedness of $\F{P}$, and \eqref{eq:soundness-2} holds.
\end{proof}

\begin{example}[Rule Overlap Inst.\ Proc.]\label{ex:roi}
  Consider $\R_{\mathtt{ROI}}$ from \cite{FLOPS-Journal}.
    \begin{align*}
        \tf(\td(x))  \to \{ \nicefrac{3}{4}: \te(\,\tf(\tg(x)),\,\tf(\th(x))\,), \;
        \nicefrac{1}{4}: \ta \}
        \quad 
        \tg(\ta) & \to \{1:\td(\ta)\} \\[-.08cm]
        \th(\tb) & \to \{1:\td(\tb)\}\!
    \end{align*}
  $\R_{\mathtt{ROI}}$
    has constant  runtime complexity, 
    i.e., $\iota_{\R_{\mathsf{ROI}}} = \Pol_0$,
    because for every instantiation, 
    at most one of the two recursive $\tf$-calls 
    in the right-hand side of the $\tf$-rule can be evaluated.
    The reason is that we can either use the $\tg$-rule if  $x$ is instantiated with $\ta$, 
    or we can apply the $\th$-rule if $x$ is instantiated with $\tb$, but not both.

    With the rule overlap instantiation processor,
our ADP framework can determine this constant complexity automatically.
Using the  dependency graph processor on the canonical ADP problem $\langle
\ADPair{\R_{\mathtt{ROI}}}, \ADPair{\R_{\mathtt{ROI}}} \rangle$  yields
 $\langle \PP_{\mathsf{ROI}}, \PP_{\mathsf{ROI}} \rangle$, where $\PP_{\mathsf{ROI}}$
consists of

\vspace*{-.3cm}

{\small
    \begin{align*}
        \tf(\td(x))  &\to \{ \nicefrac{3}{4}: \te(\,\tF(\tg(x)),\,\tF(\th(x))\,), \;
        \nicefrac{1}{4}: \ta \}^{\ttrue}\\[-.08cm]
        \tg(\ta) & \to \{1:\td(\ta)\}^{\ttrue} \\[-.08cm]
        \th(\tb) & \to \{1:\td(\tb)\}^{\ttrue}\!
    \end{align*}}

    We apply $\Proc_{\mathtt{ROI}}$ using the term $t = \tf(\tg(x))$ whose
 only narrowing substitution is $\delta = \{x/\ta\}$.
\pagebreak[3]    For the other subterm $\tF(\th(x))$ with annotated root, $\tf(\th(x))$ is not captured
    by $\delta$. 
Hence, we generate an additional ADP where this second subterm is annotated.
    Thus, we replace the former $\tf$-ADP by the following two new ADPs.

\vspace*{-.3cm}

{\small
    \begin{align*}
    \tf(\td(\ta)) & \to \{ \nicefrac{3}{4}: \te(\,\tF(\tg(\ta)),\,\tF(\th(\ta))\,), \;
    \nicefrac{1}{4}: \ta \}^{\ttrue} \\[-.08cm]
    \tf(\td(x))   & \to \{ \nicefrac{3}{4}: \te(\,\tf(\tg(x)),\,\tF(\th(x))\,), \; \nicefrac{1}{4}: \ta \}^{\ttrue}\!
\end{align*}}

\noindent
    Now one can remove the annotation of $\tF(\th(\ta))$ from the first ADP by the dependency graph processor 
    and then apply the reduction pair processor with the interpretation that maps 
    $\tF$ to 1 and all other symbols to 0 to remove all annotations. This
   proves $\iota_{\langle \PP_{\mathsf{ROI}}, \PP_{\mathsf{ROI}} \rangle} = \iota_{\R_{\mathsf{ROI}}} = \Pol_0$.
    Using a constant polynomial interpretation 
    would not be possible without the rule overlap instantiation processor.
\end{example}

}

\end{document}